\documentclass[12pt,reqno]{amsart}
\usepackage[pdfborder={0 0 0.5 [3 2]}, plainpages=false]{hyperref}%
\usepackage[left=1in,right=1in,top=1in,bottom=1in]{geometry}%
\usepackage[citation-order]{amsrefs}%
\usepackage{amsmath}
\usepackage{enumerate}
\usepackage{amssymb}                
\usepackage{amsfonts}
\usepackage{amsthm}
\usepackage{bbm}
\usepackage[table,xcdraw]{xcolor}
\usepackage{float}
\usepackage{mathtools}
\usepackage{cool}
\usepackage{booktabs}
\usepackage{graphicx,epsfig}

\usepackage{subcaption}
\usepackage[capitalize,nameinlink]{cleveref}
\crefname{section}{section}{sections}
\crefname{subsection}{subsection}{subsections}
\Crefname{section}{Section}{Sections}
\Crefname{subsection}{Subsection}{Subsections}

\Crefname{figure}{Figure}{Figures}

\crefformat{equation}{\textup{#2(#1)#3}}
\crefrangeformat{equation}{\textup{#3(#1)#4--#5(#2)#6}}
\crefmultiformat{equation}{\textup{#2(#1)#3}}{ and \textup{#2(#1)#3}}
{, \textup{#2(#1)#3}}{, and \textup{#2(#1)#3}}
\crefrangemultiformat{equation}{\textup{#3(#1)#4--#5(#2)#6}}%
{ and \textup{#3(#1)#4--#5(#2)#6}}{, \textup{#3(#1)#4--#5(#2)#6}}{, and \textup{#3(#1)#4--#5(#2)#6}}

\Crefformat{equation}{#2Equation~\textup{(#1)}#3}
\Crefrangeformat{equation}{Equations~\textup{#3(#1)#4--#5(#2)#6}}
\Crefmultiformat{equation}{Equations~\textup{#2(#1)#3}}{ and \textup{#2(#1)#3}}
{, \textup{#2(#1)#3}}{, and \textup{#2(#1)#3}}
\Crefrangemultiformat{equation}{Equations~\textup{#3(#1)#4--#5(#2)#6}}%
{ and \textup{#3(#1)#4--#5(#2)#6}}{, \textup{#3(#1)#4--#5(#2)#6}}{, and \textup{#3(#1)#4--#5(#2)#6}}

\crefdefaultlabelformat{#2\textup{#1}#3}

\hypersetup{hidelinks}

\def\R{{\mathbb R}}

\def\C{{\mathbb C}}
\def\Z{{\mathbb Z}}

\def\per{\textrm{per}}
\def\calL{\mathcal{L}}
\def\calH{\mathcal{H}}

\def\calM{\mathcal{M}}

\newcommand{\evec}{\mathbf{e}}
\newcommand{\uvec}{\mathbf{u}}
\newcommand{\vvec}{\mathbf{v}}
\newcommand{\wvec}{\mathbf{w}}

\newcommand{\yvec}{\mathbf{y}}

\DeclareMathOperator{\spn}{span}

\graphicspath{ {images/} }

\newtheorem{lemma}{Lemma}
\newtheorem{theorem}{Theorem}
\newtheorem{corollary}{Corollary}

\newtheorem{hypothesis}{Hypothesis}

\theoremstyle{definition}

\newtheorem{remark}{Remark}

\begin{document}

\title[Revisiting Multi-breathers in the discrete Klein-Gordon equation]{Revisiting Multi-breathers in the discrete Klein-Gordon equation: A Spatial Dynamics Approach}

\author{Ross Parker}
\address{Department of Mathematics, Southern Methodist University, 
Dallas, TX 75275, USA}
\email{rhparker@smu.edu}

\author{Jes\'us Cuevas-Maraver}
\address{Grupo de F\'{\i}sica No Lineal, Departamento de F\'{\i}sica Aplicada I,
Universidad de Sevilla. Escuela Polit\'{e}cnica Superior, C/ Virgen de Africa, 7, 41011-Sevilla, Spain}
\address{Instituto de Matem\'{a}ticas de la Universidad de Sevilla (IMUS). Edificio
Celestino Mutis. Avda. Reina Mercedes s/n, 41012-Sevilla, Spain}
\email{jcuevas@us.es}

\author{P.\,G. Kevrekidis} 
\address{Department of Mathematics and Statistics, University of Massachusetts, Amherst MA 01003, USA}
\email{kevrekid@math.umass.edu}

\author{Alejandro Aceves}
\address{Department of Mathematics, Southern Methodist University, 
Dallas, TX 75275, USA}
\email{aaceves@smu.edu}

\begin{abstract}
	We consider the existence and spectral stability of multi-breather structures in the discrete Klein-Gordon equation, both for soft and hard symmetric potentials. To obtain analytical results, we project the system onto a finite-dimensional Hilbert space consisting of the first $M$ Fourier modes, for arbitrary $M$. On this approximate system, we then take a spatial dynamics approach and use Lin's method to construct multi-breathers from a sequence of well-separated copies of the primary, single-site breather. We then locate the eigenmodes in the Floquet spectrum associated with the interaction between the individual breathers of such multi-breather	states by reducing the spectral problem to a matrix equation. Expressions for these eigenmodes for the approximate, finite-dimensional system are obtained in terms of the primary breather and its kernel eigenfunctions, and these are found to be in very good agreement with the numerical Floquet spectrum results. This is supplemented with results from numerical timestepping experiments, which are interpreted using the spectral computations.
\end{abstract}

\maketitle

\section{Introduction}

Dynamical models on a one-dimensional lattice have been a topic of interest for well over 50 years. Perhaps the most famous example is the Fermi–Pasta–Ulam–Tsingou (FPUT) model \cites{FPUT,Zabusuy1965}, which was one of the first problems to be studied using numerical simulations. In this work, we will examine the discrete Klein-Gordon (DKG) equation~\cites{braun2004,SGbook,p4book,kivsharmalomed}
\begin{equation}\label{eq:introDKG}
\ddot{u}_n = d (\Delta_2 u)_n - f(u_n),
\end{equation}
which describes the dynamics of an infinitely long, one-dimensional lattice of particles. The quantity $u_n$ represents the displacement of the $n$th particle in the integer lattice as a function of the time $t$. Each particle is harmonically coupled to its two nearest neighbors via the discrete second difference operator $\Delta_2$, and the strength of this coupling is quantified by the parameter $d$. The particles are subject to an external, nonlinear, on-site potential $V(u)$, such that $f(u) = V'(u)$, which can be of different
type depending on the model~\cites{Karachalios,imamat}. Common nonlinearities are shown in \cref{table:V}. (See \cref{sec:DKGbreather} for the definition of hard and soft potentials).

\begin{table}
\begin{tabular}{lll}\toprule
Equation & $V(u)$ & $f(u)$ \\ \midrule
sine-Gordon & $1 - \cos u$ & $\sin u$ \\
$\phi^4$ (soft) & $\frac{1}{2}u^2 - \frac{1}{4}u^4$ & $u(1-u^2)$ \\
$\phi^4$ (hard) & $\frac{1}{2}u^2 + \frac{1}{4}u^4$ & $u(1+u^2)$ \\
Morse & $\frac{1}{2}(1 - e^{-u})^2$ & $e^{-u}(1 - e^{-u})$ \\ \bottomrule
\end{tabular}
\caption{Common nonlinearities for the discrete Klein-Gordon equation.}
\label{table:V}
\end{table}

The DKG equation is the discrete analogue of the nonlinear Klein-Gordon partial differential equation
\begin{equation*}
u_{tt} = u_{xx} - f(u),
\end{equation*}
which is a prototypical model in the study of nonlinear waves and solitons. One of the the most well-studied forms of this equation is the sine-Gordon equation \cites{braun2004,SGbook,p4book,kivsharmalomed}, which has periodic nonlinearity $f(u)=\sin(u)$, and is completely integrable. The transition between the continuous and discrete models in the sine-Gordon case has been extensively discussed in \cite{SGchapter}. The discrete sine-Gordon model, also known as the Frenkel-Kontorova model, was originally devised as a model for describing the dynamics of a crystal lattice near a dislocation core \cites{braun1998,braun2004}, and has been  
used since its inception in numerous additional applications (see, for example, \cite{braun2004}*{Chapter 2}), including a mechanical model for a chain of pendula \cites{Scott1969,english}, arrays of Josephson junctions \cites{Ustinov1992,Floria1998}, and DNA dynamics \cites{Yomosa1983,Yakushevich1998,DeLeo2011}.

Two major classes of coherent structures in the nonlinear Klein-Gordon equation (both discrete and continuous) are of particular interest: kinks, which are heteroclinic structures resembling ``wave fronts'' that connect two adjacent minima of the potential $V(u)$, and breathers, which are structures that are spatially localized and oscillatory in time. For the continuum sine-Gordon equation, exact, analytical solutions for both of these structures have been found \cite{SolitonBook1}. 
For discrete systems, the existence and stability of static kinks have been well-studied (see \cites{PEYRARD198488,KevrekidisWeinstein2000,SGchapter}, as well as \cite{Parker2021} for results on multi-kink solutions), and there has also been interest in moving kinks \cites{Aigner2003,Iooss2006,Cisneros2008,Cisneros2011}. We will concern ourselves herein with discrete breather solutions.

Discrete breathers have been studied in Hamiltonian \cite{Flach1998} and dissipative systems \cite{Flach2008a}, and have applications in areas such as laser scanning microscopy and coupled optical waveguides \cites{Flach2008,LEDERER20081}. Existence and stability of discrete breather solutions in Klein-Gordon lattices were first studied by MacKay and Aubry \cites{MacKay1994,Aubry1997} by considering the system near the anti-continuum (AC) limit ($d=0$), in which the individual sites in the lattice are uncoupled. The advantage of this approach is that the solution for a single site is known at the AC limit, and this can be continued to small $d>0$ using the implicit function theorem. Some results about asymptotic stability of these breathers can be found in \cite{Bambusi2013}. A similar approach has been used for multi-site solitons in the discrete nonlinear Schr\"odinger equation (DNLS) \cites{Pelinovsky2005,KALOSAKAS200644}, which demonstrated that the only 
potentially stable multi-solitons are those in which adjacent peaks are excited out-of-phase. Analysis of multi-breathers in DKG, in which a finite number of sites in the integer lattice are excited at the AC limit, was done in \cites{Archilla2003,Koukouloyannis2009}, but this was restricted to the case where the excited sites are adjacent, so that a concrete result could be deduced for the dynamics of the relative phases between oscillators at adjacent sites.
Indeed, the two methods (the Aubry band method~\cite{Aubry1997} and the MacKay
effective Hamiltonian method~\cite{sepulchre}) explored in the above two publications
were shown to yield the same results near this limit 
in~\cite{doi:10.1142/S0218127411029690}.
In this situation, for small $d>0$, out-of-phase multi-breathers are stable for soft potentials, and in-phase multi-breathers are stable for hard potentials~\cite{Archilla2003}*{Theorem 6}. Results on the existence and spectral stability of multi-breathers were extended in \cite{Pelinovsky2012} to multi-site breathers where any arbitrary, finite set of lattice points is excited at the AC limit; in particular, this allows the excited sites to be separated in the lattice, i.e., there can be ``holes'' in the lattice. However, spectral stability of multi-breathers does not necessarily imply nonlinear stability; even if the multi-breather is spectrally stable, nonlinear instabilities can result from resonance between internal eigenmodes and the continuous spectrum band, as long as the two have opposite Krein signatures \cite{cuevas-maraver2016}
(see also the simpler example of the discrete
nonlinear Schr{\"o}dinger model in~\cite{PhysRevLett.114.214101}). A recent result uses the Schauder fixed point theorem to prove the existence of discrete Klein-Gordon breathers in the setting of a convex on-site potential \cite{hennig2021}. 
Other related work concerns breather solutions in infinite Fermi-Pasta-Ulam-Tsingou (FPUT) lattices, in which the nonlinear potential involves inter-site terms (see \cite{FPUbook} and the references therein). Results on the existence of time periodic solutions in FPUT lattices can be found in \cite{FPUbook}*{Chapter 2}, as well as the more recent work in \cites{Arioli2019,yoshimura2021}.

In this paper, we take a different approach to multi-breathers. We start with a single-site breather, which we call the primary breather. We take the existence of such a breather for a particular $d>0$ as a hypothesis. We then construct a multi-breather by joining together multiple, well-separated copies of this primary breather. Consecutive copies of the primary breather can be either in-phase or out-of-phase. In effect, we replace the condition that the coupling parameter $d>0$ is small with the condition that the individual copies of the primary breather are well separated. 
The construction of complex coherent structures from simple building blocks has a rich mathematical history (see \cite{Sandstede1998}, and the references therein). The mathematical technique that we will use to accomplish this is an implementation of the Lyapunov-Schmidt reduction known as Lin's method. This method has been successfully employed in many systems, including semilinear parabolic PDEs \cites{Sandstede1998,doi:10.1137/0150029} and discrete dynamical systems \cite{Knobloch2000}. We used a similar method to construct multi-solitons in DNLS \cite{Parker2020} and multi-kinks in DKG \cite{Parker2021}. 

Heuristically, we use Lin's method to construct a double breather from the primary breather as follows (see \cite{Parker2020}*{Section 1} for a similar construction of double pulses in DNLS). Let $q_n(t)$ be the primary breather solution to \cref{eq:introDKG} with period $T$, i.e. $q_n(t+T) = q_n(t)$ for all integers $n$. Then a $T$-periodic solution $u_n(t)$ is a double breather if there is an integer $N\gg 1$ such that
\[
\sup_{n\leq 0,\:t\in[0,T]} |u_n(t) - q_{n+N}(t)| + \sup_{n\geq0,\:t\in[0,T]} |u_n(t) - q_{n-N}(t)|
\]
is small, i.e. $u_n(t)$ resembles (to leading order) the sum of two copies of the primary breather $q_n(t)$ translated by $N$ units to the left and to the right, respectively. We note that $u_n(t)$ must be close to the sum of the two translates of the primary breather for all $t \in [0,T]$. See the top panel of \cref{fig:lin} for an illustration. The double breather (\cref{fig:linb}) resembles two sequential copies of the primary breather (\cref{fig:lina}). Importantly, however, it is not identical to two adjoined translates of the primary breather; compare the center node $(n=0)$ of the double breather in \cref{fig:linb} to the corresponding node $(n=2)$ of the primary breather in \cref{fig:lina}. For each $N\gg1$, Lin's method yields a piecewise double breather that comprises three pieces
\begin{equation}\label{eq:piecewisebreather}
\begin{cases}
u_n^1(t) & n \in (-\infty, -N] \\
u_n^2(t) & n \in [-N, N] \\
u_n^3(t) & n \in [N, \infty),
\end{cases}
\end{equation}
as shown in Figure~\ref{fig:lin}. This piecewise function will be a genuine double breather if and only if these pieces coincide at $n=\pm N$, i.e. $u_{-N}^1(t) = u_{-N}^2(t)$ and $u_{N}^2(t) = u_{N}^3(t)$ for $t \in [0,T]$ (see schematic in \cref{fig:linc}). This approach reduces the existence problem for double breathers to solving these two jump conditions. 

\begin{figure}
	\begin{center}
	\begin{subfigure}{0.45\linewidth}
		\caption{}
		\includegraphics[width=7.5cm]{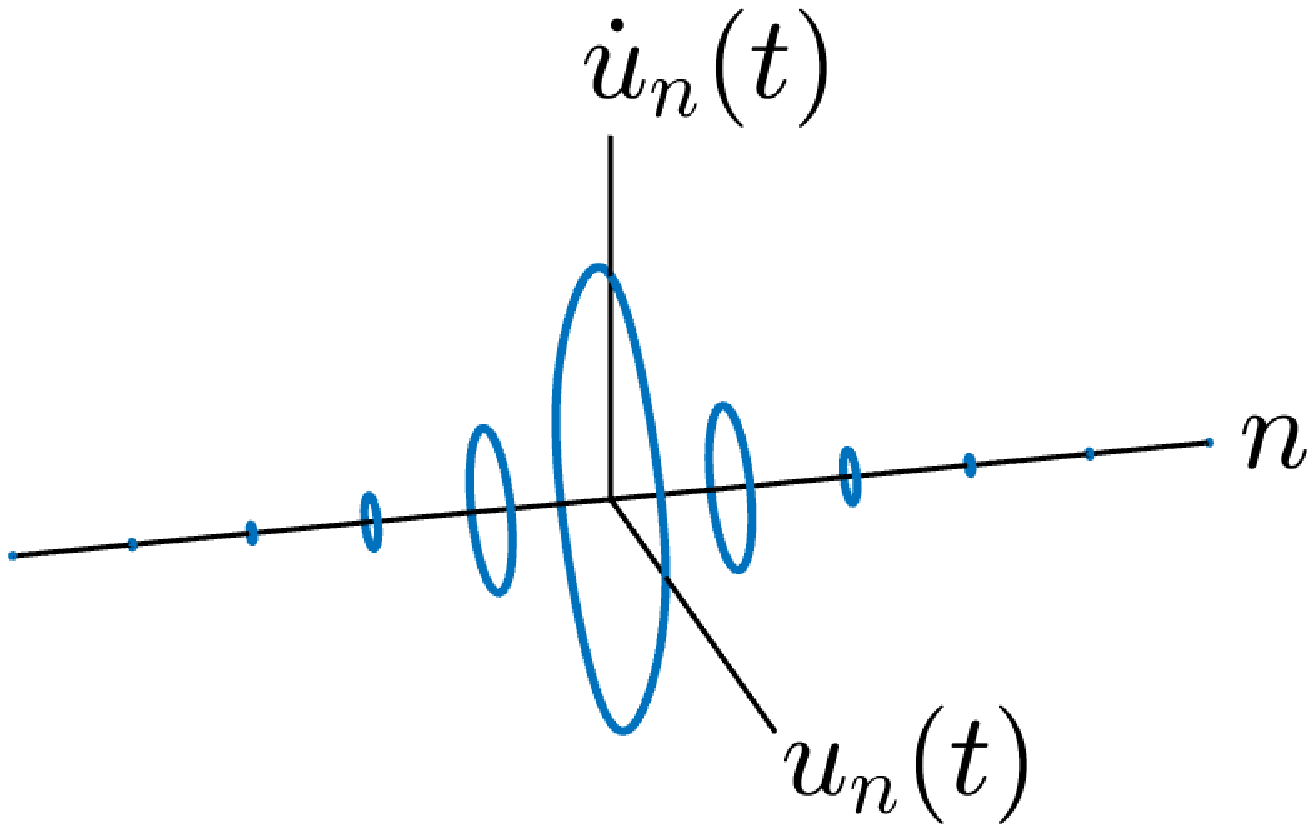}
		\label{fig:lina}
	\end{subfigure}
	\begin{subfigure}{0.45\linewidth}
		\caption{}
		\includegraphics[width=7.5cm]{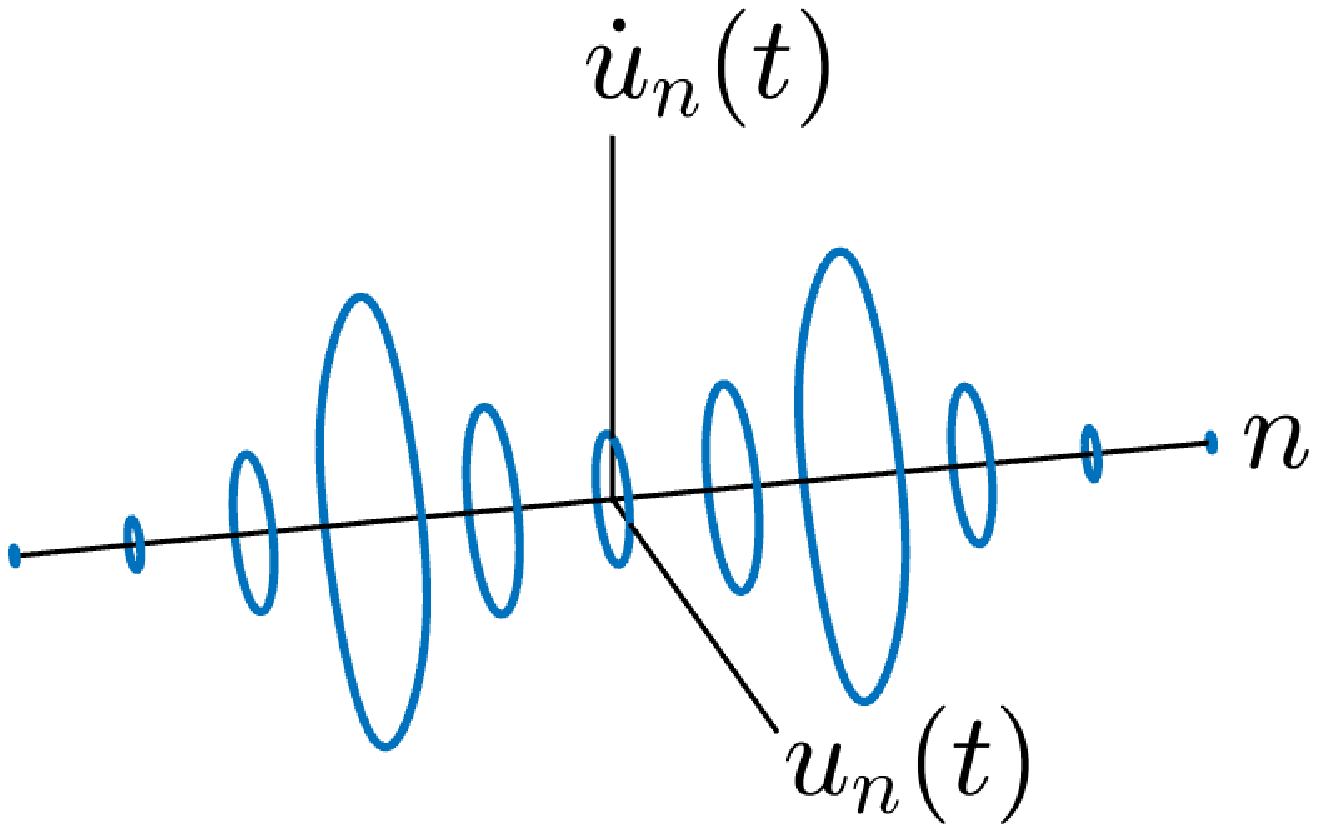}
		\label{fig:linb}
	\end{subfigure}
	\begin{subfigure}{0.9\linewidth}
		\caption{}
		\hspace{1cm}
		\includegraphics[width=13cm]{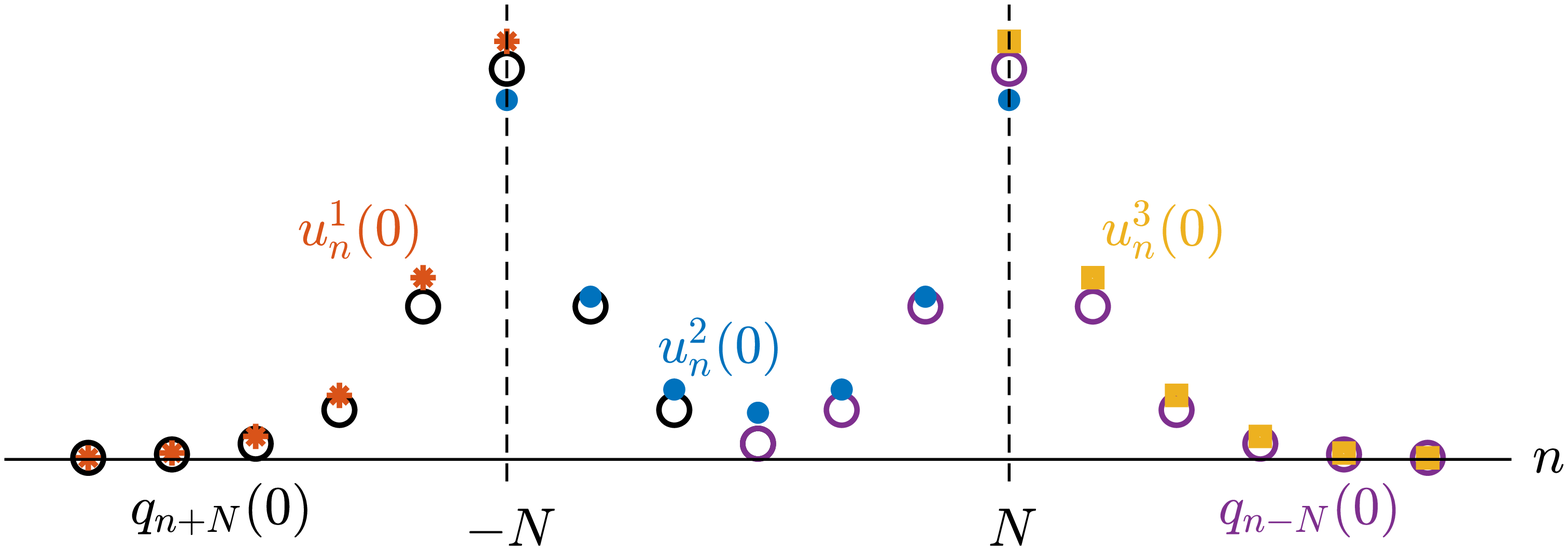}
		\label{fig:linc}
	\end{subfigure}
	\end{center}
	\caption{
	Top panel plots $\dot{u}_n(t)$ vs. $u_n(t)$ at each lattice site $n$ for a primary breather (a) and an in-phase double breather (b) for discrete sine-Gordon. Lattice site $n=0$ is located where the axes meet.
	Period $T=9.8944$, coupling parameter $d=0.25$. (c) is a schematic showing the piecewise construction of a double breather \cref{eq:piecewisebreather} with Lin's method, illustrated at $t=0$. The piecewise double breather $u_n$ comprises three pieces $u_n^1$, $u_n^2$, and $u_n^3$ (red stars, blue filled circles, and yellow squares, from left to right). Two copies of the primary breather $q_n$ are placed at $n=-N$ and $n=N$, respectively (black open circles and purple open circles, from left to right). This is a genuine double breather if and only if the two jumps at $n=\pm N$ are zero for all $t \in [0,T].$}
	\label{fig:lin}
\end{figure}

Lin's method can also be used to determine spectral stability. For a multi-pulse, multi-kink, or multi-breather, there are specific elements in the point spectrum, which we will hereafter term interaction eigenmodes, that result from nonlinear interactions between neighboring copies of the primary coherent structure. For breathers, which are periodic in time, these interaction eigenmodes assume the form of Floquet multipliers or Floquet exponents.
(We note that these are called internal modes in \cite{cuevas-maraver2016}; since we find from numerical simulations that other internal modes split off from the continuous spectrum bands as the coupling parameter $d$ is increased, we use the term interaction eigenmodes to refer to this specific set of internal modes).
In a similar fashion to how we constructed the multi-breather,
we can use Lin's method to find these interaction eigenmodes by constructing the corresponding eigenfunction. To leading order, this eigenfunction is a piecewise linear combination of translates of the kernel eigenfunction associated with the primary breather.
This reduces the spectral problem to a finite-dimensional matrix eigenvalue problem, which can then be solved. While the matrix obtained this way has the same form as that in \cites{Pelinovsky2012,cuevas-maraver2016}, its elements are computed using the primary breather solution for a specific $d$ rather than at the AC limit, thus the interaction eigenmodes can be computed more accurately for a greater range of $d$, i.e. further from the AC limit.

For the DKG equation, we follow a similar approach to that which we used in \cite{Parker2020}. We reformulate the equation using a spatial dynamics approach by recasting it as a lattice dynamical system on an appropriate function space, and then we use Lin's method to both construct multi-breathers and determine the interaction eigenmodes associated with these solutions. The major limitation with this method, however, is that for a periodic breather on $[0,T]$, the most appropriate function space is $C^\infty_\per([0,T])$, which is not a closed subspace of $L^2_\per([0,T])$. It is not straightforward to adapt the necessary mathematical tools, such as the stable manifold theorem and exponential dichotomies, to this space. As an alternative, since we are interested in smooth solutions, we project the system onto the finite-dimensional subspace $X_M$ of $L^2_\per([0,T])$ consisting of the first $M$ Fourier modes of the standard orthonormal basis, and consider the problem restricted to this finite-dimensional Hilbert space.

Although this means that the results we obtain are only approximate, we believe this is a reasonable approximation, since the Fourier coefficients for smooth functions decay exponentially, and this method is valid for arbitrary, finite $M$ (see \cref{fig:freqdecay} below).
That being said, it is not easy to quantify how good this approximation actually is. For example, as in many other instances, if resonances in the system occur, they may be affected by this truncation.
Since numerical discretization of periodic solutions is often done using a Fourier spectral method, this approximation applies directly to this discretization. Finally, the results from this method are in very good agreement with the results obtained by directly computing the Floquet spectrum of the corresponding multi-breather, both for the soft sine-Gordon and the hard $\phi^4$ potentials. In addition, the spectral pattern agrees qualitatively with that found using the approach in \cites{Pelinovsky2012,cuevas-maraver2016} for small $d$, and the matrix reduction found by using Lin's method and a spatial dynamics approach has the same form as that in those references.
For both soft and hard potentials, the eigenvalue pattern is determined by the phase differences (in-phase vs. out-of-phase) between adjacent copies of the primary breather. For the hard potential, it also depends on whether the distances between copies of the primary breather are an even or odd number of lattice points.

This paper is organized as follows. In \cref{sec:bg}, we present the mathematical background for the discrete Klein-Gordon equation, including breather solutions, the continuous spectrum, and the reformulation using spatial dynamics. In \cref{sec:findim}, we formulate our finite-dimensional approximation, which we then use in \cref{sec:multi} to prove results about the existence and spectrum of multi-breathers in the approximate system. The proof of the spectral results is deferred to \cref{app:specproof}. Numerical results are presented in \cref{sec:numerics}, which are in very good agreement with the predictions from the main theorems. Results are presented for both the soft sine-Gordon potential and the hard $\phi^4$ potential. In addition, we present results from timestepping simulations to illustrate the effect of the spectrum on the dynamical evolution of the system. We end with a brief concluding section, which suggests avenues for future research.

\section{Mathematical background}\label{sec:bg}

We will consider the discrete Klein-Gordon (DKG) equation with on-site nonlinearity $f(u)$
\begin{equation}\label{eq:DKG}
\ddot{u}_n = d (\Delta_2 u)_n - f(u_n),
\end{equation}
on the integer lattice $\Z$, where $t \in \R$ is the evolution time, $u_n(t) \in \R$ is the displacement of the $n$th particle in the lattice, $(\Delta_2 u)_n = u_{n+1} - 2 u_n + u_{n-1}$ is the discrete second difference operator, and $f(u) = V'(u)$ for a smooth, on-site potential function $V(u)$. We use the following assumptions for the potential $V$:
\begin{enumerate}[(i)\leftmargin=\parindent]
\item $V$ is an even function, and $V(0) = 0$. This implies $V'(0) = 0$.
\item $V''(0)>0$.
\end{enumerate}
We note that the Morse potential, which is considered in \cite{cuevas-maraver2016}, does not satisfy the first assumption, since it is not an even function. For any time $t$, we take the displacements $\{u_n(t)\}_{n \in \Z} \in \ell^2(\Z)$, and we denote this sequence by $\uvec(t)$. Existence and uniqueness of solutions to \cref{eq:DKG} is discussed in \cite{cuevas-maraver2016}. Since $f(u)$ is an odd function, if $\uvec(t)$ is a solution to \cref{eq:DKG}, then $-\uvec(t)$ is as well. Equation \cref{eq:DKG} is Hamiltonian \cites{KevrekidisWeinstein2000,cuevas-maraver2016}, and can be written as
\begin{equation}\label{eq:Hform}
\frac{d}{dt}\begin{pmatrix} u_n \\ v_n \end{pmatrix} = 
\begin{pmatrix} 0 & 1 \\ -1 & 0 \end{pmatrix}\begin{pmatrix} \partial \calH / \partial u_n \\ \partial \calH / \partial v_n \end{pmatrix},
\end{equation}
where $v_n = \dot{u}_n$ is the velocity of the $n$th particle in the lattice, and $\calH$ is the conserved energy
\begin{equation}\label{eq:H}
	\calH(\uvec) = \sum_{n=-\infty}^\infty 
	\left( \frac{1}{2} v_n^2 + \frac{d}{2} (u_{n+1} - u_n)^2 + V(u_n) \right).
\end{equation}

\subsection{Breathers}\label{sec:DKGbreather}

We are interested in breather solutions to \cref{eq:DKG}, which are periodic in time and exponentially localized in their spatial
profile over the lattice. Specifically, a breather solution is a function $\uvec \in \ell^2(\Z, H^2_\per[0,T])$, where $H^2_\per[0,T]$ is the Hilbert-Sobolev space of periodic, real-valued functions on $[0,T]$. 
The fundamental period $T$ is the smallest positive real number for which $\uvec(t+T) = \uvec(t)$ for all $t$. At the AC limit ($d = 0$), the individual sites in the lattice are decoupled. At each site, $u_n(t)$ is a $T$-periodic solution to the nonlinear oscillator equation
\begin{equation}\label{eq:singlesiteAC}
\ddot{\phi} + V'(\phi) = 0,
\end{equation}
which has conserved energy $E = \frac{1}{2}\dot{\phi}^2 + V(\phi)$. For fixed energy $E$, equation \cref{eq:singlesiteAC} has a unique, even solution $\phi(t)$, which we will call the fundamental AC breather \cite{Pelinovsky2012}. This solution satisfies the initial conditions $\phi(0) = a$ and $\dot{\phi}(0) = 0$, where $a$ is the smallest, positive root of $V(a) = E$. The fundamental period $T$ of $\phi(t)$ is a function of the energy $E$, and is given by
\begin{equation}\label{eq:TE}
T(E) = \sqrt{2}\int_{-a}^a \frac{du}{\sqrt{E - V(u)}}.
\end{equation}
The potential $V(u)$ is a hard potential if the period $T$ decreases as the energy $E$ increases, and a soft potential if $T$ increases as $E$ increases \cites{Pelinovsky2012,cuevas-maraver2016}.
Pelinovsky and Sakovich prove the existence of multi-site breathers close to the AC limit, i.e. for $d$ small, which are even functions of $t$ \cite{Pelinovsky2012}. Specifically, for a finite set of lattice sites $S = \{ k_1, \dots, k_N \}$, with $k_i < k_{i+1}$, they start with a solution
\begin{equation}
\uvec^{(0)}(t) = \sum_{i=1}^N \sigma_i \phi(t) \evec_{k_i}
\end{equation}
at the AC limit, where $\phi(t)$ is the fundamental AC breather, $\evec_{k_i}$ is the unit vector for site $k_i$ in the integer lattice, and $\sigma_i = \pm 1$ is the phase factor for the oscillator at site $k_i$. Adjacent oscillators are in-phase if $\sigma_i \sigma_{i+1} = 1$, and out-of-phase if $\sigma_i \sigma_{i+1} = -1$. They then use the implicit function theorem to prove the existence of a multi-site breather $\uvec^{(d)}(t)$ to \cref{eq:DKG} for sufficiently small $d$ \cite{Pelinovsky2012}*{Theorem 1}. 
The fact that in-phase and out-of-phase structures
are the only ones available in discrete Klein-Gordon 
lattices with nearest-neighbor interactions
has been shown in~\cite{KOUKOULOYANNIS20132022}.
In that light, we will restrict our considerations
to such configurations in what follows. Nevertheless,
it is relevant to mention in passing that the examination 
of so-called phase-shift multi-breathers (with relative
phases different than $0$ or $\pi$) in lattices with
interactions beyond nearest-neighbor ones remains an
active topic of investigation in Klein-Gordon (and DNLS)
settings~\cite{PENATI201992}.

\subsection{Linearization}\label{sec:DKGlinear}

For a specific coupling constant $d$ and fundamental period $T$, let $\uvec(t)$ be a breather solution to \cref{eq:DKG}. To study the spectral stability of $\uvec(t)$, we linearize equation \cref{eq:DKG} about $\uvec(t)$ by substituting the perturbation ansatz $\uvec(t) + \epsilon \vvec(t)$ and keeping terms of order $\epsilon$ to obtain the linearized equation
\begin{equation}\label{eq:DKGlinear}
\ddot{v}_n = d (\Delta_2 v)_n - f'(u_n)v_n,
\end{equation}
which can be written as the first order linear system
\begin{equation}\label{eq:DKGlinear1}
\frac{d}{dt} \begin{pmatrix} v_n \\ w_n \end{pmatrix} = 
\begin{pmatrix}
w_n \\ 
d (\Delta_2 v)_n - f'(u_n)v_n
\end{pmatrix}
\end{equation}
by letting $w_n = \dot{v}_n$. Since $\uvec(t)$ has period $T$, it follows from Floquet theory that its spectral stability depends on the Floquet multipliers, which are the spectrum of the monodromy operator $\calM = \Phi(0, T)$, where $\Phi(s, t)$ is the evolution operator for \cref{eq:DKGlinear1}. 
If $\mu$ is a Floquet multiplier, then the corresponding Floquet exponent $\lambda$ (which is unique modulo $2 \pi i/T$) is related to $\mu$ by $\mu = e^{\lambda T}$. 
For every Floquet exponent $\lambda$, there is a corresponding solution $\vvec(t) = e^{\lambda t} \wvec(t)$ to the linearized equation \cref{eq:DKGlinear}, where $\wvec(t)$ is periodic with period $T$ (see, for example, \cite{Kapitula2013}*{Lemma 2.1.29}). Substituting this ansatz into \cref{eq:DKGlinear}, we obtain the Floquet eigenvalue problem
\begin{equation}\label{eq:DKGeig}
d (\Delta_2 w)_n - f'(u_n)w_n - \ddot{w}_n = 2 \lambda \dot{w}_n + \lambda^2 w_n,
\end{equation}
where $\wvec \in \ell^2(\Z, H^2_\per[0,T]) \subset \ell^2(\Z, L^2_\per[0,T])$, and we use the inner product
\begin{equation}\label{eq:IP1}
\langle \uvec, \vvec \rangle_{\ell^2(\Z, L^2_\per[0,T])} = \sum_{n=-\infty}^\infty \int_0^T u_n(s) \overline{v_n(s)} ds
\end{equation}
on $\ell^2(\Z, L^2_\per[0,T])$. We can write equation \cref{eq:DKGeig} as 
\begin{equation}\label{eq:DKGeigL}
\calL(\uvec)\wvec = (2 \lambda \partial_t + \lambda^2 )\wvec,
\end{equation}
where the linear operator $\calL(\uvec)$ is defined by the LHS of \cref{eq:DKGeig}. Since \cref{eq:DKG} is Hamiltonian, the Floquet exponents must come in quartets $\lambda = \pm \alpha \pm \beta i$. It follows that the Floquet multipliers $\mu$ can only occur in one of three patterns: a pair $\{ \mu, \overline{\mu} \}$ on the unit circle; a pair $\{ \mu, \mu^{-1} \}$ on the real line; or a quartet $\{ \mu, \overline{\mu}, \mu^{-1}, \overline{\mu}^{-1} \}$ off of the unit circle. Therefore, the breather solution $\uvec(t)$ is spectrally unstable unless all of its Floquet multipliers lie on the unit circle.

There is always a Floquet exponent at 0 (corresponding to the Floquet multiplier $\mu = 1$), since $\dot{\uvec}$ is a solution to \cref{eq:DKGeig}, i.e. $\calL(\uvec)\dot{\uvec} = 0$, which can be verified by differentiating \cref{eq:DKG} with respect to $t$. Furthermore, there exists a solution $\yvec \in \ell^2(\Z, H^2_\per[0,T])$ which solves 
\begin{equation}
\calL(\uvec)\yvec = 2 \ddot{\uvec},
\end{equation}
and can be chosen to be perpendicular to $\dot{\uvec}$ with respect to the inner product \cref{eq:IP1} (see \cite{Pelinovsky2012}*{Section 3}, noting that the corresponding linear operator to $\calL(\uvec)$ in \cite{Pelinovsky2012} has the opposite sign). In fact, we can actually compute $\yvec(t)$. 
Letting $\omega = 2 \pi / T$ be the frequency of the breather, and normalizing the period of the breather to $2 \pi$ by rescaling the time variable to $\tau = \omega T$, as in \cite{kevrekidis2016}, equation \cref{eq:DKG} becomes
\begin{equation}\label{eq:DKGomega}
\omega^2 \partial_\tau^2 u_n = d (\Delta_2 u)_n - f(u_n).
\end{equation}
Differentiating with respect to $\omega$, we obtain
\begin{equation}\label{eq:DKGdiffw}
d (\Delta_2 \partial_\omega u)_n - f'(u_n)\partial_\omega u_n 
- \omega^2 \partial_\tau^2 ( \partial_\omega u_n) = 2 \omega \partial_\tau^2 u_n.
\end{equation}
Changing variables back to $t$, this becomes $\calL(\uvec)\partial_\omega \uvec = \frac{2}\omega \ddot{\uvec}$, thus $\yvec = \omega \partial_\omega \uvec$.

\subsection{Continuous spectrum}

The continuous spectrum of $\calL(\uvec)$ is the set of all $\lambda$ for which the limiting problem 
\begin{equation}\label{eq:DKGeigcont}
d (\Delta_2 w)_n - f'(0)w_n - \ddot{w}_n = 2 \lambda \dot{w}_n + \lambda^2 w_n
\end{equation}
has a bounded solution in $n$. 
Following the procedure in \cite{cuevas-maraver2016}*{Section 2.1}, the continuous spectrum consists of the bands
\begin{equation}\label{eq:contspec}
\begin{aligned}
\lambda &= i\left( \pm \omega(\theta) - \frac{2 \pi m}{T} \right) && \qquad m \in \Z \\
\omega(\theta) &= \sqrt{ f'(0) + 4 d \sin^2\left( \frac{\theta}{2} \right) } && \qquad\theta \in [-\pi, \pi]
\end{aligned}
\end{equation}
on the imaginary axis. The corresponding Floquet multipliers comprise two bands on the unit circle,  which are symmetric about the real axis, and are given by
\begin{equation}\label{eq:contspecmult}
\begin{aligned}
\mu = \exp \left( \pm i \sqrt{f'(0) + 4 d \sin^2 \left(\frac{\theta}{2}\right) }T \right)  && \qquad \theta \in [0, \pi].
\end{aligned}
\end{equation}
Let $m_0$ be the largest nonnegative integer such that $T - 2 \pi m_0 > 0$. At the AC limit, the bands consist of the two points $\mu = \exp(\pm i \theta_0)$, where $\theta_0 =  T - 2 \pi m_0$. For $d>0$, the bands stretch from $\mu = \exp(\pm i \theta_0)$ to $\mu = \exp(\pm i \theta_1)$, where $\theta_1 = \sqrt{f'(0) + 4 d }T - 2 \pi m_0$. Following the analysis in \cite{cuevas-maraver2016}*{Section 2.2}, if $T \in (n \pi, (n+1)\pi)$ for $n$ even, the upper band has positive Krein signature, and the lower band has negative Krein signature. As $d$ is increased, the bands grow towards $(-1,0)$ (see \cref{fig:bands}, left). The ends of the bands meet when $\theta_1 = \pi$, which occurs when $d = \frac{1}{4} \left( \frac{(1 + 2 m_0)^2 \pi^2}{T^2} - 1\right)$, at which point they merge into a single band. They meet again when $\theta_1 = 0$, which occurs when $d = \frac{1}{4} \left( \frac{(2 + 2 m_0)^2 \pi^2}{T^2} - 1\right)$, at which point they comprise the entire unit circle.
Conversely, if $T \in (n \pi, (n+1)\pi)$ for $n$ odd, the upper band has negative Krein signature, the lower band has positive Krein signature, and the bands grow towards (1,0) as $d$ is increased (see \cref{fig:bands}, right).

\begin{figure}
	\begin{center}
	\begin{subfigure}{0.45\linewidth}
		\caption{}
		\includegraphics[width=7.5cm]{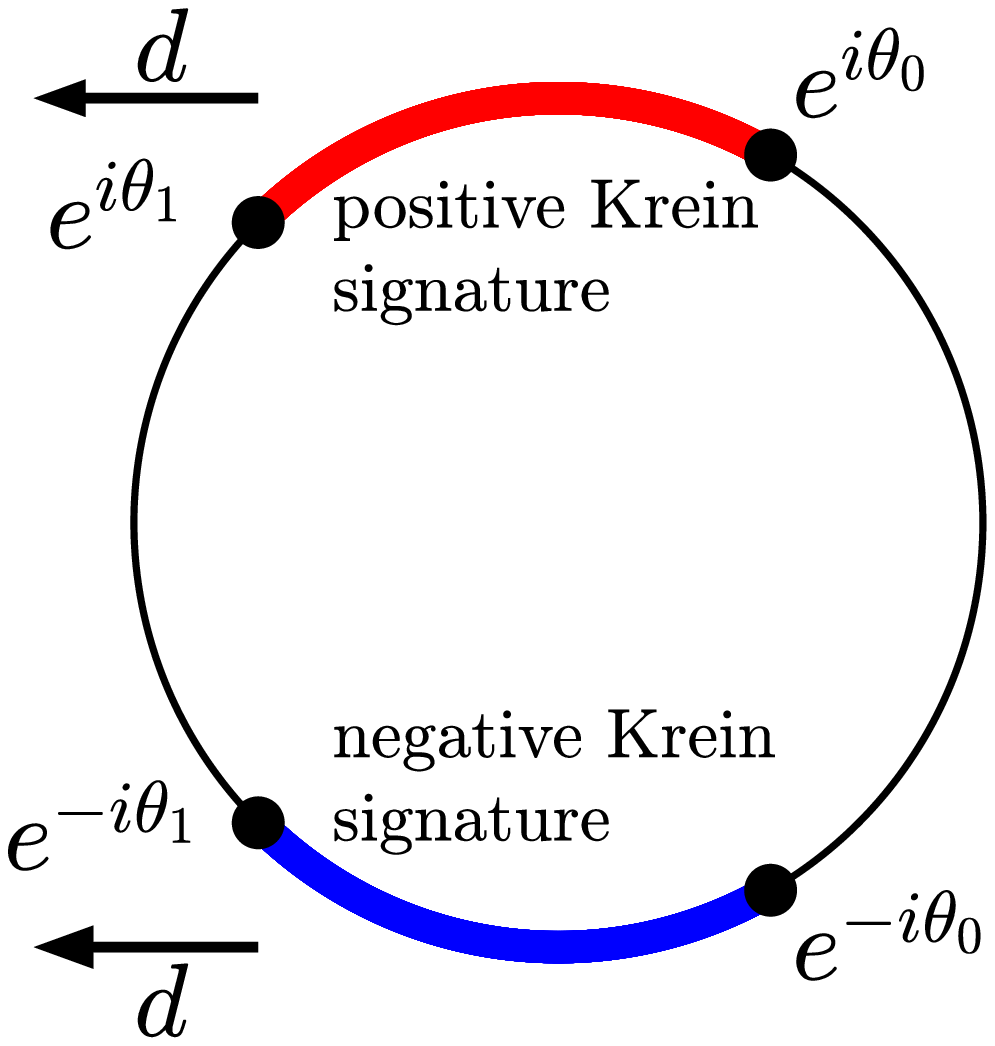}
	\end{subfigure}
	\begin{subfigure}{0.45\linewidth}
		\caption{}
		\includegraphics[width=7.5cm]{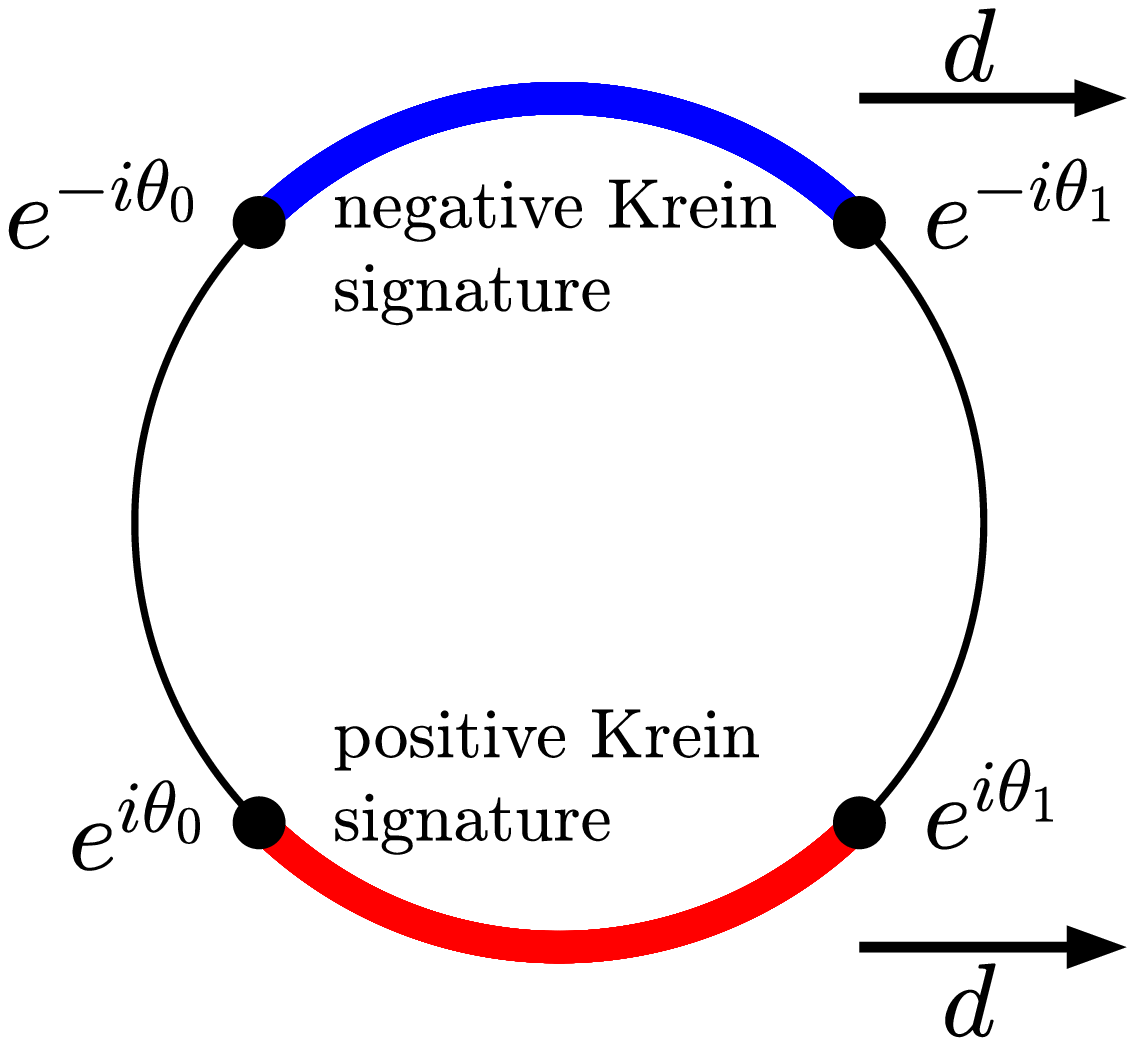}
	\end{subfigure}
	\end{center}
	\caption{Schematic showing continuous spectrum bands on the unit circle for $T \in (n \pi, (n+1)\pi)$ with $n$ even (a) and $n$ odd (b). Krein signature of bands is indicated, and bands grow in the direction of the arrow with increasing $d$.}
	\label{fig:bands}
\end{figure}

\subsection{Spatial dynamics}

We now reformulate both the DKG equation \cref{eq:DKG} and the eigenvalue problem \cref{eq:DKGeig} using a spatial dynamics approach, as in \cites{Parker2020,Parker2021}. Let $\uvec(t)$ be a breather solution to \cref{eq:DKG} with period $T$, and define $U(n) = (u(n), \tilde{u}(n)) = ( u_n, u_{n-1} )$. Then equation \cref{eq:DKG} can be written as the lattice dynamical system
\begin{equation}\label{eq:dynEq}
U(n+1) = F(U(n)),
\end{equation}
where
\begin{equation}\label{eq:F}
F\begin{pmatrix}u \\ \tilde{u} \end{pmatrix} = 
\begin{pmatrix}2u  + \dfrac{1}{d}f(u) + \dfrac{1}{d} \partial_t^2 u - \tilde{u} \\
u
\end{pmatrix}.
\end{equation}
We note that since $f$ is a odd function, if $U(n)$ is a solution to \cref{eq:dynEq}, then $-U(n)$ is as well. The Floquet eigenvalue problem \cref{eq:DKGeig} can similarly be written as 
\begin{equation}\label{eq:dynEVP}
W(n+1) = \left[ DF(U(n)) + (2 \lambda \partial_t + \lambda^2) B \right] W(n),
\end{equation}
where
\begin{equation}\label{eq:DFU}
DF(U(n)) = \begin{pmatrix}
2 + \dfrac{f'(u(n))}{d} + \dfrac{1}{d}\partial_t^2  & -1 \\ 1 & 0
\end{pmatrix}, \qquad
B = \begin{pmatrix} 1 & 0 \\ 0 & 0 \end{pmatrix}.
\end{equation}
The zero function is an equilibrium solution to \cref{eq:dynEq}. At this point, the standard procedure (see, for example, \cites{Parker2021,Parker2020,Sandstede1998}) is to consider $U(n)$ as a homoclinic orbit of the equilibrium at 0. The complication here is that for the difference equation \cref{eq:dynEq} to be well-posed, we require $U(n) \in C_\per^\infty([0,T],\R^2)$ for all $n$, rather than just $U(n) \in H^2_\per([0,T], \R^2)$, since each application of $F$ involves differentiating twice with respect to $t$, and equation \cref{eq:dynEq} involves applying $F$ an arbitrary number of times. In essence, the spatial dynamics formulation in equation \cref{eq:dynEq} is more restrictive than the original system.
Since $C_\per^\infty([0,T])$ is not a closed subspace of $L^2([0,T])$, it is not straightforward to adapt the stable manifold theorem and results on exponential dichotomies to this problem, even if $DF(0)$ has the desired spectral properties. As an alternative, we will consider a finite-dimensional approximation, where we project the problem onto a finite-dimensional subspace of $L_\per^2([0,T])$. Roughly, this subspace consists of the first $M$ Fourier modes of the standard Hilbert basis for $L_\per^2([0,T])$. Although this is a limited result, we note that since we require $U(n)$ to be smooth, its Fourier coefficients will decay exponentially. Since we can take $M$ as large as we like, this should be a very good approximation. In addition, since many of the numerical simulations are performed using Fourier spectral methods, this theory applies directly to such numerical discretization. Before we formulate our approximation, we prove some important results about the spectrum of $DF(0)$.

\subsection{Spectrum of \texorpdfstring{$DF(0)$}{DF(0)}}

The linearization of \cref{eq:dynEq} about the equilibrium at 0 is the constant coefficient linear operator 
\begin{equation}\label{eq:DF0}
DF(0) = \begin{pmatrix}
\dfrac{1}{d}\partial_t^2 + \dfrac{f'(0)}{d} + 2 & -1 \\ 1 & 0
\end{pmatrix},
\end{equation}
which is invertible with inverse
\begin{equation}\label{eq:DF0inv}
DF(0)^{-1} = \begin{pmatrix}
0 & 1 \\ -1 & \dfrac{1}{d}\partial_t^2 + \dfrac{f'(0)}{d} + 2
\end{pmatrix}
\end{equation}
First, we determine the eigenvalues and eigenfunctions of $DF(0)$.

\begin{lemma}\label{lemma:DF0eigs}
The set of eigenvalues of $DF(0)$ is given by $\bigcup_{k \in \Z} \{\lambda_k, \lambda_k^{-1} \}$, where 
\begin{equation}\label{eq:DF0lambdak}
\lambda_k = \frac{1}{2}\left( r_k + \sqrt{r_k^2 - 4} \right), \quad r_k = -\frac{4 k^2 \pi^2}{d T^2} + \frac{f'(0)}{d} + 2.
\end{equation}
For $k \neq 0$, these have algebraic multiplicity 2, since $\lambda_{-k} = \lambda_k$. For each $k \in \Z$, $\{\lambda_k, \lambda_k^{-1} \}$ is either real, or a complex conjugate pair on the unit circle. The eigenfunctions corresponding to $\left\{ \lambda_k, \lambda_k^{-1} \right\}$ are $\left\{ U_k(t), U_k^{-1}(t) \right\}$, which are defined, up to constant multiple, by 
\begin{equation}\label{eq:DF0eigenfns}
\begin{aligned}
U_k(t) &= \begin{pmatrix}v_k(t) \\ \lambda_k^{-1}  v_k(t) \end{pmatrix}, \quad
U_k^{-1}(t) = \begin{pmatrix}v_k(t) \\ \lambda_k v_k(t) \end{pmatrix}, \quad
v_k(t) = \frac{1}{T} \exp\left( i \frac{2 \pi k t}{T} \right).
\end{aligned}
\end{equation}
\end{lemma}
\begin{proof}
Consider the eigenvalue problem $DF(0) U(t) = \lambda U(t)$ on $H^2_\per([0,T],\R^2)$, where $U(t) = (v(t), w(t))^T$. We note that $\lambda = 0$ is not an eigenvalue, since that implies $v = w = 0$. The eigenvalue problem then reduces to the system of equations
\begin{align}\label{eq:DF0EVPsystem}
\left( \frac{1}{d}\partial_t^2 + \frac{f'(0)}{d} + 2 \right) v(t) = \left( \lambda + \frac{1}{\lambda} \right) v(t), \quad
w(t) = \frac{1}{\lambda} v(t).
\end{align}
Letting $r = \lambda + \frac{1}{\lambda}$ and using the periodic boundary conditions $v(T) = v(0)$, the set of solutions to \cref{eq:DF0EVPsystem} is given by
\begin{align}
v_k(t) &= \frac{1}{T} \exp\left( i \frac{2 \pi k t}{T} \right), \quad r_k = -\frac{4 k^2 \pi^2}{d T^2} + \frac{f'(0)}{d} + 2 && k \in \Z,
\end{align}
where the functions $v_k(t)$ have been normalized. The corresponding eigenvalues of $DF(0)$ are then given by $\left\{ \lambda_k, \lambda_k^{-1} \right\}$, where $\lambda_k$ is defined by \cref{eq:DF0lambdak}, and the corresponding eigenfunctions are given by \cref{eq:DF0eigenfns}. The pair $\left\{ \lambda_k, \lambda_k^{-1} \right\}$ is real if $|r_k| \geq 2$, and is complex with modulus 1 if $|r_k| < 2$.
\end{proof}

We note that the spectrum of $DF(0)$ depends on both the coupling parameter $d$ and the period $T$. It follows from \cref{lemma:DF0eigs} that the spectrum of $DF(0)$ is bounded away from the unit circle provided $|r_k| > 2$ for all $k$. The following lemma gives some conditions on $T$ and $d$ to guarantee that this is the case.

\begin{lemma}\label{lemma:DF0hyp}
The spectrum of $DF(0)$ is bounded away from the unit circle if, for a specific nonnegative integer $k$, $T$ and $d$ are chosen so that
\begin{equation}\label{eq:Tdpair}
\frac{2 k \pi}{\sqrt{f'(0)}} < T < \frac{2 (k+1) \pi}{\sqrt{f'(0)}} , \qquad 0 < d < \frac{(k+1)^2\pi^2}{T^2} - \frac{f'(0)}{4}.
\end{equation}
\begin{proof}
Since $f'(0) > 0$, it follows that $r_0 = 2 + \frac{1}{d}f'(0) > 2$. In addition, $r_k$ is strictly decreasing in $k$, with $r_k \rightarrow -\infty$ as $k \rightarrow \infty$. Thus $|r_k| > 2$ for all $k$ if both $r_k > 2$ and $r_{k+1} < -2$ for some nonnegative integer $k$, from which the conditions \cref{eq:Tdpair} follow.
\end{proof}
\end{lemma}

We take the following assumption on the spectrum of $DF(0)$, which is the analogue to hyperbolicity in the finite-dimensional case.

\begin{hypothesis}\label{hyp:hyp}
The coupling constant $d$ and period $T$ are chosen so that the spectrum of $DF(0)$ is bounded away from the unit circle.
\end{hypothesis}

\section{Finite dimensional approximation}\label{sec:findim}

In this section, we define our finite dimensional approximation for \cref{eq:dynEq}. For $M \geq 1$, let 
\begin{equation}\label{eq:XM}
X_M = \spn\left\{ \bigcup_{k = -M}^M v_k(t) \right\}, \qquad
v_k(t) = \frac{1}{T} \exp\left( i \frac{2 \pi k t}{T} \right)
\end{equation}
be the $(2M+1)$-dimensional subspace of $L_\per^2([0,T])$ spanned by the Fourier basis functions with wavenumber $|k| \leq M$, and let $P_M: L_\per^2([0,T]) \rightarrow X_M$, defined by
\begin{equation}\label{eq:PM}
P_M u(t) = \sum_{k=-M}^M \langle u, v_k \rangle_{L^2([0,T])} v_k(t)
= \sum_{k=-M}^M \left( \int_0^T u(s) \overline{v_k(s)} ds \right) v_k(t)
\end{equation}
be the corresponding projection operator. In addition, let $X_{M,e}$ be the $(M+1)$-dimensional subspace of $X_M$ comprising functions which are even in $t$, i.e.
\begin{equation}
X_{M,e} = \left\{ f \in X_M : f(-t) = f(t) \text{ for all }t \in \R \right\}.
\end{equation}

For $M \geq 1$, define the approximation of the discrete Klein-Gordon equation \cref{eq:DKG} on $\ell^2(\Z, X_M)$ by
\begin{equation}\label{eq:DKGapprox}
\begin{aligned}
\ddot{u}_n &= d (\Delta_2 u)_n - g(u_n) && \qquad u_n \in X_M,
\end{aligned}
\end{equation}
where $g: X_M \rightarrow X_M$ is given by 
\begin{equation}g(u) = P_M f(u).
\end{equation} 
When $u = 0$, $g(0) = P_M f(0) = 0$, and $g'(0) = P_M f'(0) = f'(0)$, since $f'(0)$ is a constant function, which is in $X_M$ for all $M$.
Furthermore, since $f$ is an odd function, $g$ is as well. In general, it not the case that $P_M f(u) = f(P_M u)$, thus we cannot simply obtain a solution to \cref{eq:DKGapprox} by projecting a solution of \cref{eq:DKG} onto $X_M$. However, since the Fourier coefficients of a smooth, $T$-periodic function $u(t)$ decay exponentially, equation \cref{eq:DKGapprox} should be a reasonable approximation to \cref{eq:DKG} for large $M$. Linearization of \cref{eq:DKGapprox} about a solution $\uvec \in \ell^2(\Z, X_M)$ yields the eigenvalue problem
\begin{equation}\label{eq:DKGMeig}
\begin{aligned}
d (\Delta_2 w)_n - g'(u_n)w_n - \ddot{w}_n = 2 \lambda \dot{w}_n + \lambda^2 w_n,
\end{aligned}
\end{equation}
which we write as
\begin{equation}\label{eq:DKGMeigL}
\calL_M(\uvec)\wvec = (2 \lambda \partial_t + \lambda^2 )\wvec,
\end{equation} 
where $\calL_M(\uvec)$ is the linear operator on $\ell^2(\Z, X_M)$ defined by the LHS of \cref{eq:DKGMeigL}. As in \cref{sec:DKGlinear}, $\calL_M(\uvec) \dot{\uvec} = 0$, and there exists a solution $\yvec$ to $\calL_M(\uvec) \yvec = 2 \ddot{\uvec}$, where $\yvec = \omega \partial_\omega \uvec$ and $\omega = 2 \pi / T$.

Reformulating \cref{eq:DKGapprox} from a spatial dynamics perspective yields the discrete dynamical system on $X_M^2$
\begin{align}\label{eq:dynEqM}
U(n+1) &= F_M(U(n)) && U(n) \in X_M^2,
\end{align}
where
\begin{equation}\label{eq:FM}
F_M\begin{pmatrix}u \\ \tilde{u} \end{pmatrix} = 
\begin{pmatrix}2u  + \dfrac{1}{d}g(u) + \dfrac{1}{d} \partial_t^2 u - \tilde{u} \\
u
\end{pmatrix}
\end{equation}
Since $F_M(-U) = -F_M(U)$, it follows that if $U(n)$ is a solution to \cref{eq:dynEqM}, so is $-U(n)$. Similarly, the eigenvalue problem \cref{eq:DKGMeig} can be reformulated as
\begin{equation}\label{eq:dynEVPM}
W(n+1) = \left[ DF_M(U(n)) + (2 \lambda \partial_t + \lambda^2) B \right] W(n),
\end{equation}
where $B$ is defined in \cref{eq:DFU}. Since $g'(0) = f'(0)$, the linear operator $DF_M(0)$ on $X_M^2$ is also given by \cref{eq:DF0}. 

Finally, let $\{\tau(s) : s \in \R\}$ be the one-parameter group of unitary translation operators on $X_M^2$, defined by $[\tau(s)]U(\cdot) = U(\cdot - s)$, which has infinitesimal generator $\tau'(0) = \partial_t$. We note that this group is well-defined on $X_M^2$, since 
\[
\tau(s) v_k(t) =
\frac{1}{T} \exp\left( i \frac{2 \pi k (t-s)}{T}\right) 
= \exp\left( -i \frac{2 \pi k s}{T} \right) \frac{1}{T} \exp\left( i \frac{2 \pi k t}{T}\right) 
= \exp\left( -i \frac{2 \pi k s}{T} \right) v_k(t),
\]
i.e., the group action multiplies a basis element by a constant. In fact, the eigenvalue of \cref{eq:DKGMeig} at 0 is a result of this translational symmetry. The function $F_M$ from \cref{eq:dynEqM} (as well as $F$ from the full system \cref{eq:dynEq}) commutes with this one-parameter group, i.e. $F(\tau(s) U) = \tau(s) F(U)$. It is crucial to note that this symmetry is lost if we consider the problem \cref{eq:dynEqM} on $X_{M,e}^2$, since the space of even functions is not translation invariant. 

\section{Multi-breathers}\label{sec:multi}

Our strategy will be to first prove that multi-breathers exist on the subspace $X_{M,e}^2$ of even functions. Since there is no translational symmetry, the stable and unstable manifolds of the origin will intersect transversely, which greatly simplifies the analysis. This parallels the restriction in \cite{Pelinovsky2012} to breathers which are even in $t$, and is consistent with the odd symmetry of the nonlinearity $f$. 
Once that is accomplished, we will return to the full space $X_{M}^2$ for the eigenvalue problem, and use Lin's method as in \cites{Parker2021,Parker2020,Sandstede1998} to construct the interaction eigenfunctions as piecewise linear combinations of the eigenfunction corresponding to translation symmetry. This technique is similar to the one we employed in \cite{Parker2020} for DNLS, where, to prove the existence of multi-pulses, we removed the gauge symmetry by restricting the problem to real-valued solutions.

\subsection{Primary breather}

As noted above, we will take the existence of a primary, single-site breather as a hypothesis. 
Fix $d$ and $T$ such that \cref{hyp:hyp} holds. First, we consider the system \cref{eq:dynEqM} on $X_{M,e}^2$. By \cref{lemma:DF0eigs}, the $2M+2$ eigenvalues of $DF_M(0)$ on $X_{M,e}^2$ are given by $S_M = \bigcup_{k=0}^M \{\lambda_k, \lambda_k^{-1} \}$, where these are defined in the statement of \cref{lemma:DF0eigs}. By \cref{hyp:hyp}, the spectrum of $DF_M(0)$ is real and does not intersect the unit circle, thus 0 is a hyperbolic equilibrium point of \cref{eq:dynEqM}. Define the stable and unstable subsets of the spectrum of $DF_M(0)$ by
\[
S_M^s = \{ \lambda \in S_M : |\lambda| < 1\}, \qquad S_M^u = \{ \lambda \in S_M : |\lambda| > 1\}.
\]
By symmetry, $|S_M^s| = |S_M^u| = M+1$. Define
\begin{equation}\label{eq:defrM}
r_M = \min \{ |\lambda| : \lambda \in S_M^u, |\lambda| > 1 \}.
\end{equation}
Since $X_{M,e}^2$ is finite-dimensional, the stable manifold theorem holds. For each $M \geq 1$, let $W_{M,e}^s(0)$ and $W_{M,e}^u(0)$ be the $(M+1)$-dimensional stable and unstable manifolds of the equilibrium at 0, which are subsets of $X_{M,e}^2$. A breather solution to \cref{eq:dynEqM} is a homoclinic orbit which lies in the intersection of the stable and unstable manifolds. We take the existence of such a solution as a hypothesis.

\begin{hypothesis}\label{hyp:breather}
Let $d$ and $T$ be chosen according to \cref{hyp:hyp}. There exists a positive integer $M_0$ such that for all $M \geq M_0$, the stable and unstable manifolds $W_{M,e}^s(0)$ and $W_{M,e}^u(0)$ intersect transversely in $X_{M,e}^2$ in a homoclinic orbit $Q_M(n) = (q(n), \tilde{q}(n))^T = (q_n, q_{n-1})^T$. 
\end{hypothesis}
The first component $q_n$ is a breather solution to \cref{eq:DKGapprox}. As a consequence of the stable manifold theorem, we have the estimate 
\begin{equation}\label{eq:U1decayest}
\|Q_M(n)\|_{L^2_\per([0,T])} \leq C r_M^{-|n|}.
\end{equation}

We now consider \cref{eq:dynEqM} on the larger space $X_M^2$. The spectrum of $DF_M(0)$ on $X_M^2$ is exactly the same as that on $X_{M,e}^2$, except the eigenvalues corresponding to $k = 1, \dots, M$ have multiplicity of 2. It follows that 0 is also a hyperbolic equilibrium of \cref{eq:dynEqM} on $X_M^2$. Let $W_M^s(0)$ and $W_M^u(0)$ be the $(2M+1)$-dimensional stable and unstable manifolds of the equilibrium at 0, which this time are subsets of $X_M^2$. For $M \geq M_0$, $Q_M(n)$ is also a homoclinic orbit connecting these stable and unstable manifolds.  
This intersection, however, will not be transverse. Due to translational symmetry, these manifolds will have an intersection which is at least one-dimensional.
In the next hypothesis, we assume that this intersection is non-degenerate, i.e. this intersection is exactly one-dimensional. (We note that systems with higher dimensional intersections do exist; one example is the Ablowitz-Ladik discretization of the 
nonlinear Schr{\"o}dinger equation, where it stems from the existence of
exact single soliton solutions~\cites{Ablowitz1975,Ablowitz1976,Kapitula2001}).

\begin{hypothesis}\label{hyp:breathernondegen}
Let $d$ and $T$ be chosen according to \cref{hyp:hyp}, and let $M_0$ and $Q_M(n)$ be as in \cref{hyp:breather}. Then for all $M \geq M_0$, the stable and unstable manifolds $W_M^s(0)$ and $W_M^u(0)$ have an intersection in $Q_M(n)$ which is exactly one-dimensional.
\end{hypothesis}

The variational equation is the linearization of \cref{eq:dynEqM} on $X_M^2$ about the homoclinic orbit solution $Q_M(n)$, which is given by
\begin{equation}\label{eq:vareq}
\begin{aligned}
W(n+1) &= DF_M(Q_M(n)) W(n) && \qquad W(n) \in X_M^2.
\end{aligned}
\end{equation}
Since the tangent spaces of $W_M^s(0)$ and $W_M^u(0)$ have a one-dimensional intersection by \cref{hyp:breathernondegen}, it follows that $\dot{Q}_M(n) = (\dot{q}_n, \dot{q}_{n-1})$ is the unique, bounded solution to \cref{eq:vareq}, up to scalar multiples. ($\dot{Q}_M(n)$ is not a solution to \cref{eq:vareq} on $X_{M,e}^2$, since it is an odd function). We can thus decompose the tangent spaces to $W_M^u(0)$ and $W_M^s(0)$ at $Q_M(0)$ as
\begin{equation}\label{eq:TWdecomp}
T_{Q_M(0)}W^u_M(0) = \R \dot{Q}_M(0) \oplus Y_M^-, \qquad  
T_{Q_M(0)}W^s_M(0) = \R \dot{Q}_M(0) \oplus Y_M^+,
\end{equation}
where $\dim Y_M^- = \dim Y_M^+ = 2M$. In addition, the adjoint variational equation
\begin{equation}\label{eq:adjvareq}
Z(n) = DF_M(Q_M(n))^* Z(n+1)
\end{equation}
has a unique bounded solution, given by
\begin{equation}\label{eq:Z1}
Z_M(n) = (-\dot{q}_{n-1}, \dot{q}_n)^T,
\end{equation}
and $Z_M(0) \perp \R \dot{Q}_M(0) \oplus Y_M^- \oplus Y_M^+$ by \cite{Parker2020}*{Lemma 1}. We can thus decompose $X_M^2$ as
\begin{equation}\label{eq:Xdecomp}
X_M^2= \R \dot{Q}_M(0) \oplus Y_M^- \oplus Y_M^+ \oplus \R Z_M(0).
\end{equation}

\subsection{Existence}

We construct a multi-breather on $X_{M,e}^2$ by splicing together multiple copies of the primary breather $Q_M(n)$ in a chain. We characterize a multi-breather in the following way. Let $m > 1$ be the total number of copies of the primary breather in the chain. Let $N_i$ ($i = 1, \dots, m-1$) be the distances (in lattice points) between the center point of each breather. We seek to construct a solution $U(n)$ which can be written piecewise in the form
\begin{equation}\label{eq:Upiecewise}
\begin{aligned}
U_i^-(n) &= \sigma_i Q_M(n) + \tilde{U}_i^-(n) && n \in [-N_{i-1}^-, 0] && \quad i = 1, \dots, m\\
U_i^+(n) &= \sigma_i Q_M(n) + \tilde{U}_i^+(n) && n \in [0, N_i^+] && \quad i = 1, \dots, m,
\end{aligned}
\end{equation}
where $\sigma_i \in \{1, -1\}$ represents the orientation of each copy of the primary breather, $N_i^+ = \lfloor \frac{N_i}{2} \rfloor$, $N_i^- = N_i - N_i^+$, and $N_0^- = N_m^+ = \infty$. Adjacent copies of the primary breather are in-phase if $\sigma_i \sigma_{i+1} = 1$, or out-of-phase if $\sigma_i \sigma_{i+1} = -1$. (As in \cite{Pelinovsky2012}, other phase relations are not considered here; see also~\cite{KOUKOULOYANNIS20132022}). The functions $\tilde{U}_i^\pm$ in \cref{eq:Upiecewise} are remainder terms, which will be small,
so that the multi-breather resembles a sequence of well-separated copies of the primary breather, to leading order.
We also define the characteristic distance
\begin{equation}\label{defN}
N = \frac{1}{2} \min\{ N_i \},
\end{equation}
which will be used in the estimates of the remainder terms $\tilde{U}_i^\pm$. The individual pieces $U_i^\pm(n)$ are joined together end-to-end as in \cites{Sandstede1998,Knobloch2000,Parker2020,Parker2021} to create the multi-breather $U(n)$, which can be written in piecewise form as
\begin{equation}
\begin{aligned}
U(n) &= \begin{cases}
U_i^-\left( n - \sum_{j=1}^{i-1}N_j \right) & \sum_{j=1}^{i-1}N_j - N_{i-1}^- + 1 \leq n \leq \sum_{j=1}^{i-1}N_j \\
U_i^+\left( n - \sum_{j=1}^{i-1}N_j \right) & \sum_{j=1}^{i-1}N_j + 1 \leq n \leq \sum_{j=1}^{i-1}N_j + N_i^+
\end{cases}
&& i = 1, \dots, m,
\end{aligned}
\end{equation}
where we define $\sum_{j=1}^0 N_j = 0$. We have the following theorem concerning the existence of multi-breathers on $X_{M,e}^2$. 
We note that this is a solution to the finite dimensional approximation \cref{eq:dynEqM} on $X_{M}^2$ for arbitrary, fixed $M$.

\begin{theorem}\label{th:multi-breathers}
Assume \cref{hyp:hyp} and \cref{hyp:breather}, and let $M \geq M_0$, where $M_0$ is defined in \cref{hyp:breather}. 
Let $Q_M(n)$ be the primary breather solution to the discrete dynamical system \cref{eq:dynEqM} on $X_M^2$ from \cref{hyp:breather}.
Then there exists a positive integer $N^*$ with the following property. For all $m > 1$ and distances $N_i \geq N^*$, there exists a unique solution $U(n)$ to the discrete dynamical system \cref{eq:dynEqM} on $X_M^2$, where $U(n) \in X_{M,e}^2$. This solution comprises, to leading order, $m$ sequential copies of the primary breather $Q_M(n)$, and can be written piecewise in the form \cref{eq:Upiecewise}. For the remainder terms $\tilde{U}_i^\pm(n)$, we have the estimates
\begin{equation}\label{eq:Uestimates}
\begin{aligned}
\tilde{U}_i^+(N_i^+) &= \sigma_{i+1} Q_M(-N_i^-) + \mathcal{O}(r_M^{-2N}) \\
\tilde{U}_{i+1}^-(-N_i^-) &= \sigma_{i} Q_M(N_i^+) + \mathcal{O}(r_M^{-2N}) \\ 
\| \tilde{U}_i^-(n)\|_{X_M} &\leq C r_M^{-N_{i-1}^-} r_M^{-(N_{i-1}^- + n)} && \qquad n = 2, \dots, m\\
\|\tilde{U}_i^+(n) \|_{X_M} &\leq C r_M^{-N_i^+} r_M^{-(N_i^+ - n)} && \qquad n = 1, \dots, m-1 \\
\| \tilde{U}_1^-(n)\|_{X_M} &\leq C r_M^{-2N} r_M^{n} \\
\|\tilde{U}_m^+(n) \|_{X_M} &\leq C r_M^{-2N} r_M^{-n},
\end{aligned}
\end{equation}
which hold as well for derivatives of $\tilde{U}_i^\pm(n)$.
\begin{proof}
The proof is a straightforward adaptation of the proofs of Theorems 1 and 3 in \cite{Parker2020}, and is very similar to the proof of \cite{Parker2021}*{Theorem 1}. Since the stable and unstable manifolds $W_{M,e}^s(0)$ and $W_{M,e}^u(0)$ intersect transversely in $X_{M,e}$, the implementation of Lin's method does not involve jump conditions.
\end{proof}
\end{theorem}

\subsection{Spectral stability}

For a multi-breather consisting of $m$ components, as long as the individual copies of the primary breather are sufficiently well separated, there will be $2(m-1)$ eigenvalues in spectrum of \cref{eq:DKGMeig}, i.e., one pair per additional copy of the primary breather, which will be located near the origin. We call these interaction eigenmodes, since they result from nonlinear interactions between the tails of adjacent copies of the primary breather.
As in \cites{Parker2020,Sandstede1998}, we locate these eigenmodes by using Lin's method to construct the corresponding eigenfunctions (see \cref{app:specproof} for details). We note that this theorem gives the Floquet exponents $\lambda$, which are the eigenvalues of \cref{eq:DKGMeig}. The Floquet multipliers $\mu$ are related to these by $\mu = e^{\lambda T}$. 
We also note that the eigenvalues found this way are for the finite dimensional approximation \cref{eq:dynEqM} on $X_{M}^2$ for arbitrary, fixed $M$.
The proof is given in \cref{app:specproof}.

\begin{theorem}\label{th:spectrum}
Assume \cref{hyp:hyp}, \cref{hyp:breather}, and \cref{hyp:breathernondegen}, and let $M \geq M_0$, where $M_0$ is defined in \cref{hyp:breather}. Let $Q_M(n) = (q_n, q_{n-1})$ be the primary breather solution 
to the discrete dynamical system \cref{eq:dynEqM} on $X_M^2$
from \cref{hyp:breather}, and let $Y(n) = (y_n, y_{n-1}) = \omega \partial_\omega Q_M(n)$.
Let $U(n)$ be an $m$-component multi-breather constructed as in \cref{th:multi-breathers} with distances $N_i$ and orientation parameters $\sigma_i$. Then there exists $\delta > 0$ with the following property. There exists a bounded, nonzero solution $W(n)$ of the eigenvalue problem \cref{eq:dynEVPM} on $X_M^2$
for $|\lambda| < \delta$ if and only if $E(\lambda) = 0$, where
\begin{equation}\label{Elambda}
E(\lambda) = \det\left(A + \frac{1}{d}K \lambda^2 I + R(\lambda)\right).
\end{equation}
$A$ is the tri-diagonal $m \times m$ matrix
\begin{align}\label{eq:matrixA}
A &= \begin{pmatrix}
-a_1 & a_1 & & & \\
-\tilde{a}_1 & \tilde{a}_1 - a_2 & a_2 \\
& -\tilde{a}_2 & \tilde{a}_2 - a_3 & a_3 \\
& \ddots & & \ddots \\
& & & -\tilde{a}_{m-1} & \tilde{a}_{m-1}  \\
\end{pmatrix},
\end{align}
with
\begin{equation}\label{eq:ai}
\begin{aligned}
a_i &= \sigma_i \sigma_{i+1} \int_0^T \left( \dot{q}_{N_i^+}\dot{q}_{-N_i^- - 1} 
- \dot{q}_{N_i^+ - 1}\dot{q}_{-N_i^-} \right) dt \\
\tilde{a}_i &= \sigma_i \sigma_{i+1} \int_0^T \left( \dot{q}_{-N_i^-} \dot{q}_{N_i^+ - 1} 
- \dot{q}_{-N_i^- - 1}\dot{q}_{N_i^+} \right) dt,
\end{aligned}
\end{equation}
\begin{align}\label{eq:M}
K =
\sum_{n = -\infty}^\infty \int_0^T \left( \dot{q}_n^2 + 2 \dot{y}_n \dot{q}_n \right) dt,
\end{align}
and the remainder term has uniform bound
\begin{equation}\label{eq:Rbound}
\|R(\lambda)(c)\|_{X_M^2} \leq C \left( r_M^{-N} + |\lambda|\right)^3.
\end{equation}
\end{theorem}

\begin{remark}\label{remark:solvability}
Substituting $y_n = \omega \partial_\omega q_n$, we can write \cref{eq:M} as 
\begin{align}\label{eq:M2}
K =
\sum_{n = -\infty}^\infty \int_0^T \dot{q}_n \left( 2 \omega \partial_\omega \dot{q}_n + \dot{q}_n \right) dt,
\end{align}
which is the solvability condition from \cite{kevrekidis2016}. Integrating by parts, we can also write \cref{eq:M} as 
\begin{align}\label{eq:M3}
K =
\sum_{n = -\infty}^\infty \int_0^T \left( \dot{q}_n^2 - 2 y_n \ddot{q}_n \right) dt =
\sum_{n = -\infty}^\infty \int_0^T \left( \dot{q}_n^2 - 2 \omega (\partial_\omega q_n) \ddot{q}_n \right) dt.
\end{align}
In the AC limit, since the primary breather comprises a single excited site $\phi(t)$ at one of the lattice points, equation \cref{eq:M3} becomes
\begin{align}\label{eq:MAC}
K = \int_0^T \left( \dot{\phi}^2 - 2 \omega(\partial_\omega \phi) \ddot{\phi} \right) dt = - \frac{T^2(E)}{T'(E)},
\end{align}
by the proof of \cite{Pelinovsky2012}*{Lemma 2}, where $T(E)$ is defined in \cref{eq:TE}, noting that $-\omega \partial_\omega \phi$ corresponds to $v$ in that reference.
\end{remark}

\begin{corollary}\label{corr:even}
If the primary breather $Q_M = (q_n, q_{n-1})$ has the even spatial symmetry $q_{-n} = q_n$, the matrix $A$ in \cref{th:spectrum} simplifies to
\begin{align}\label{eq:matrixAsymm}
A &= \begin{pmatrix}
-a_1 & a_1 & & & \\
a_1 & -a_1 - a_2 & a_2 \\
& a_2 & -a_2 - a_3 & a_3 \\
& \ddots & & \ddots \\
& & & a_{m-1} & -a_{m-1}  \\
\end{pmatrix},
\end{align}
with $a_i = \sigma_i \sigma_{i+1} b_i$ and
\begin{align*}
b_i &= \begin{cases}
\int_0^T \dot{q}_{\frac{N_i}{2}}\left( \dot{q}_{\frac{N_i}{2}+1} - \dot{q}_{\frac{N_i}{2}-1}\right) dt & N_i \text{ even} \\
\int_0^T \left( \dot{q}_{\frac{N_i-1}{2}}\dot{q}_{\frac{N_i+3}{2}} - \dot{q}_{\frac{N_i+1}{2}}\dot{q}_{\frac{N_i-3}{2}} \right) dt & N_i \text{ odd.}
\end{cases}
\end{align*}
The matrix $A$ has one eigenvalue at 0, and the remaining $(m-1)$ eigenvalues $\{\mu_1, \dots, \mu_{m-1}\}$ are real and distinct. As long as the continuous spectrum does not interfere, there are $(m-1)$ pairs of interaction eigenmodes, which are given by 
\begin{align}\label{eq:inteigs}
\lambda_j &= \pm \sqrt{-\frac{d \mu_j}{K}} + \mathcal{O}(r^{-2N}) && j = 1, \dots, m-1.
\end{align}
These are either real or imaginary by Hamiltonian symmetry.
\begin{proof}
This follows from \cref{th:spectrum} using same argument as in \cite{Parker2020}*{Theorem 5}.
\end{proof}
\end{corollary}

Two special cases to consider are the double breather ($m=2$), and the multi-breather where all copies of the primary breather are equally spaced. For a double breather, the pair of interaction eigenmodes is given by
\begin{align}\label{eq:inteigsdouble}
\lambda &= \sqrt{\frac{2 \sigma_1 \sigma_2 b_1 d}{K}} + \mathcal{O}(r^{-2N}).
\end{align}
If $b_1$ and $K$ have the same sign, the in-phase double breather is spectrally unstable, and the out-of-phase double breather is spectrally neutrally stable. The reverse is true if $b_1$ and $K$ have opposite signs. For an $m$-site multi-breather, if $N_i = N_1$ for all $i$, i.e., the excited sites are equally spaced, then $b_i = b_1$ for all $i$. The matrix $A$ from \cref{eq:matrixAsymm} reduces to $A = b_1 S$, where 
\begin{align}\label{eq:matrixequal}
S &= \begin{pmatrix}
-\sigma_1 \sigma_2 & \sigma_1 \sigma_2 & & & \\
\sigma_1 \sigma_2 & -\sigma_2(\sigma_1+\sigma_3) & \sigma_2 \sigma_3 \\
& \sigma_2 \sigma_3 & -\sigma_3(\sigma_2+\sigma_4) & \sigma_3 \sigma_4 \\
& \ddots & & \ddots \\
& & & \sigma_{m-1}\sigma_m & -\sigma_{m-1}\sigma_m  \\
\end{pmatrix}.
\end{align}
The eigenvalues of $S$ depend only on the phase differences $\sigma_i \sigma_{i+1}$ between consecutive excited sites, and not on the individual phases $\sigma_i$. We note that this matrix reduction of the eigenvalue problem has the same form as that in equation (35) of \cite{Pelinovsky2012}*{Lemma 2}. In particular, the matrix $S$ in \cref{eq:matrixequal} and the matrix in that lemma are similar (with similarity transformation given by $P = \text{diag}(\sigma_1, \dots, \sigma_m)$), and thus have the same eigenvalues. The quantity $b_1$ corresponds to $K_N$ in that lemma. For $i = 1, \dots, m-1$, let $p_i = \sigma_i\sigma_{i+1}$ be the phase differences between consecutive copies of the primary breather, where $p_i = 1$ indicates in-phase, and $p_i = -1$ indicates out-of-phase. Let $n_-$ be the number of negative $p_i$, and $n_+$ the number of positive $p_i$. By \cite{Sandstede1998}*{Lemma 4} (see also the proof of \cite{Pelinovsky2012}*{Lemma 3}), $S$ has a single eigenvalue at 0, $n_+$ positive eigenvalues, and $n_-$ negative eigenvalues. If $b_1$ and $K$ have the same sign, there are $n_+$ pairs of unstable interaction eigenmodes, and $n_-$ pairs of neutrally stable interaction eigenmodes. Thus, the only neutrally stable multi-breather is one in which each pair of consecutive copies of the primary breather is out-of-phase. The reverse is true if $b_1$ and $K$ have opposite signs. For multi-breathers in which the excited sites are not equally spaced, the stability pattern depends on the relative signs of each $b_i$ and $K$, following \cite{Sandstede1998}*{Lemma 4}.

\section{Numerical results}\label{sec:numerics}

We construct both the primary breather and multi-breathers by numerical parameter continuation from the AC limit using the software package AUTO \cite{auto07p}. To do this, we first choose an energy level $E$, from which we compute the period $T$ using \cref{eq:TE}, and then we determine the AC breather solution $\phi(t)$ by solving \cref{eq:singlesiteAC} with the initial conditions given in \cref{sec:DKGbreather}. For the starting solution for the parameter continuation at $d = 0$, we take $u_k(t) = \pm \phi(t)$ at a finite set of lattice points, and $u_k(t) = 0$ everywhere else. In general, we will take the initial excited sites to be well separated in the lattice, since the results above only hold for well-separated copies of the primary breather. That being said, this numerical method is effective for any starting configuration at the AC limit (see \cref{fig:SGintersite}, for example, for a double breather comprising two adjacent excited sites).
For the parameter continuation, we use the equation
\begin{equation}\label{eq:HformAUTO}
\frac{d}{dt}\begin{pmatrix} u_n \\ v_n \end{pmatrix} = 
\begin{pmatrix} 0 & 1 \\ -1 & 0 \end{pmatrix}\begin{pmatrix} \partial \calH / \partial u_n \\ \partial \calH / \partial v_n \end{pmatrix} + 
\epsilon \begin{pmatrix} \partial \calH / \partial u_n \\ \partial \calH / \partial v_n \end{pmatrix},
\end{equation}
where $\calH$ is the Hamiltonian defined by \cref{eq:H}, and $\epsilon$ is a small parameter used to break the Hamiltonian structure. The parameter $\epsilon$ is necessary since homoclinic orbits and periodic orbits in Hamiltonian systems are generally codimension zero phenomena, i.e. they persist as a system parameter (the coupling parameter $d$ in this case) is varied; the additional parameter $\epsilon$ converts the problem into a codimension one problem (see \cites{Champneys1997,Beyn1990} for more details on this technique). As the coupling parameter $d$ is varied in \cref{eq:HformAUTO} by numerical parameter continuation, the Hamiltonian-breaking parameter $\epsilon$ is approximately of order $10^{-15}$, thus is negligible.
To avoid continuing in the direction of translation symmetry, we add the phase condition
\[
\langle \dot{u}_k^\text{old}, u_k - u_k^\text{old} \rangle =
\int_0^T v_k^\text{old}( u_k - u_k^\text{old}) dt = 0,
\]
where $u_k^\text{old}$ refers to the function at the previous continuation step. The index $k$ is chosen to be one of the lattice points at which the initial condition is $\phi(t)$. We approximate the integer lattice with a finite lattice of length $2L+1$, numbered from $-L$ to $L$, and we take Dirichlet boundary conditions at the ends of the lattice, i.e. $u_{L+1} = u_{-L-1} = 0$. Thus equation \cref{eq:Hform} becomes an ODE in $4L+2$ dimensions, with periodic boundary conditions imposed on each variable. Once a breather solution has been found, its Floquet spectrum is found numerically by using Matlab's eigenvalue solver \texttt{eig} to compute the eigenvalues of the monodromy operator associated with the linearized system \cref{eq:DKGlinear1}. Since \cref{eq:DKGlinear1} results from writing a second order linear ODE in $2L+1$ dimensions as first order linear system in $4L+2$ dimensions, the corresponding eigenvectors are of the form $(v_n, w_n)$, where $v_n$ and $w_n$ are $(2L+1)$-dimensional.

\subsection{Discrete sine-Gordon equation}

We first consider the discrete sine-Gordon equation, where the nonlinearity is given by $f(u) = \sin u$. For all of the plots in the section, unless otherwise noted, we use a lattice size $L = 15$, and the energy at the AC limit is $E = 0.25$, which corresponds to a period of $T = 7.91$. We note that $2 \pi < T < 3 \pi$, thus the upper continuous spectrum band has positive Krein signature, and the lower band has negative Krein signature (see \cref{fig:bands}, left).
\cref{fig:singlea} shows the primary breather solution for $d = 0.25$ at $t = 0$, which is symmetric on the spatial lattice. The corresponding Floquet spectrum is shown in \cref{fig:singleb}. There is a single Floquet multiplier (with multiplicity 2) at 1, which comes from translation symmetry in $t$. The continuous spectrum bands are a discrete set of points, which is a result of the finiteness
of the lattice and of the spatial discretization, and have effectively merged into a single band by this value of $d$. 
\cref{fig:singlec} shows the exponential decay of the $L^2$ norm on $[0,T]$ of the primary breather solution $u_n(t)$ at lattice site $n$, as $n$ is increased. \cref{fig:singled} plots the relative error of this decay rate compared to $r_M$, which is defined in \cref{eq:defrM}. The exponential decay of the coefficients of the Fourier series expansion (in $t$) of the primary breather is shown in \cref{fig:freqdecaya}. As expected, since the breather is even in $t$, only even frequencies are represented in the frequency spectrum. Furthermore, the scaled amplitudes for all frequencies $k\geq 20$ are below $10^{-10}$, suggesting that the finite dimensional approximation is reasonable. 
\cref{fig:freqdecayb} shows the exponential decay of the $L^2$ norm difference $\| Q_M(n) - Q_{512}(n) \|_{\ell^2(\Z, L^2_\per[0,T])}$ between the solution $Q_M(n)$ truncated after $M$ Fourier nodes and the solution $Q_{512}(n)$ with 512 Fourier nodes, suggesting that an approximation with 32 Fourier nodes is more than adequate for this set of parameters.

\begin{figure}
	\begin{center}
	\begin{subfigure}{0.45\linewidth}
		\caption{}
		\includegraphics[width=7.5cm]{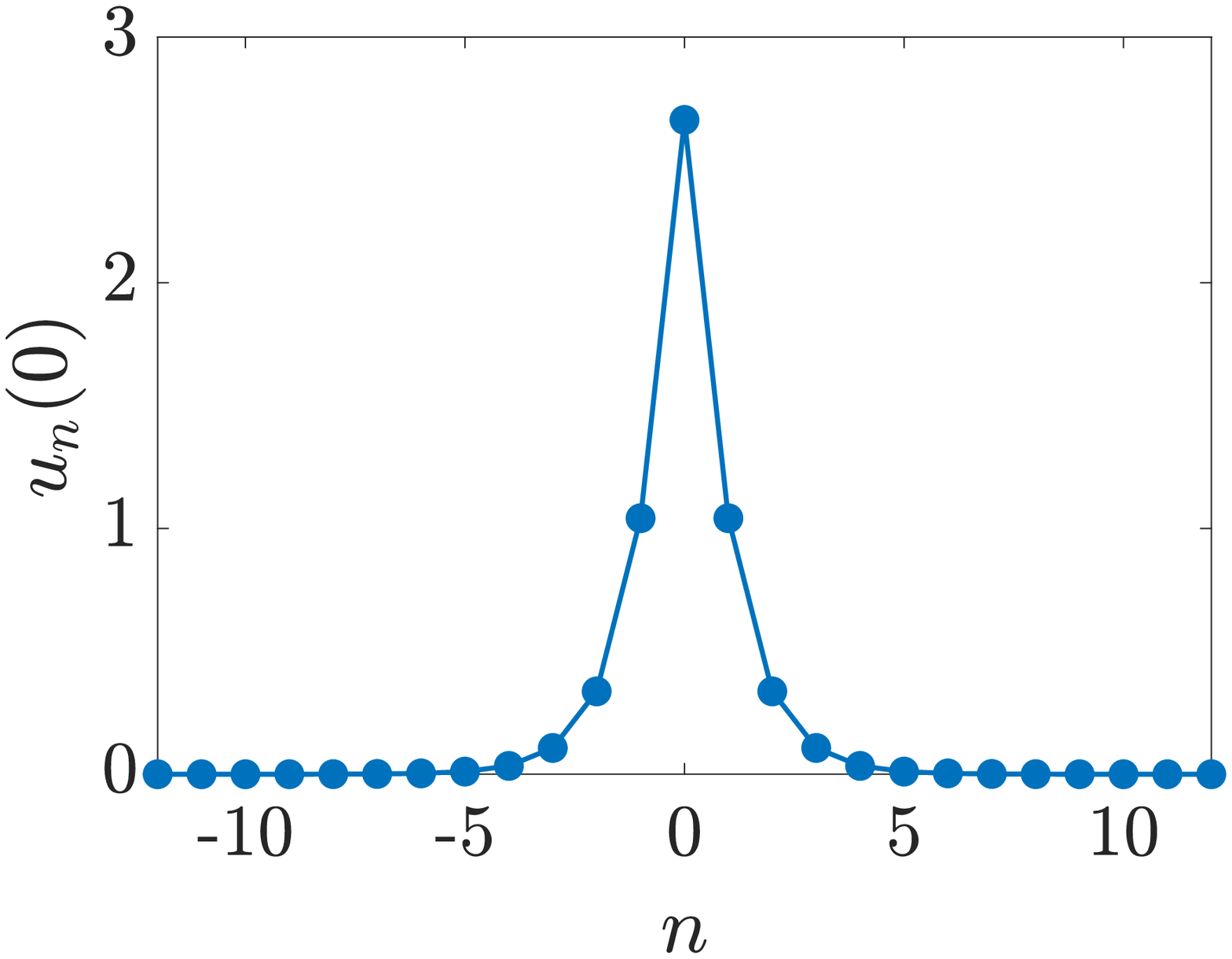}
		\label{fig:singlea}
	\end{subfigure}
	\begin{subfigure}{0.45\linewidth}
		\caption{}
		\includegraphics[width=7.5cm]{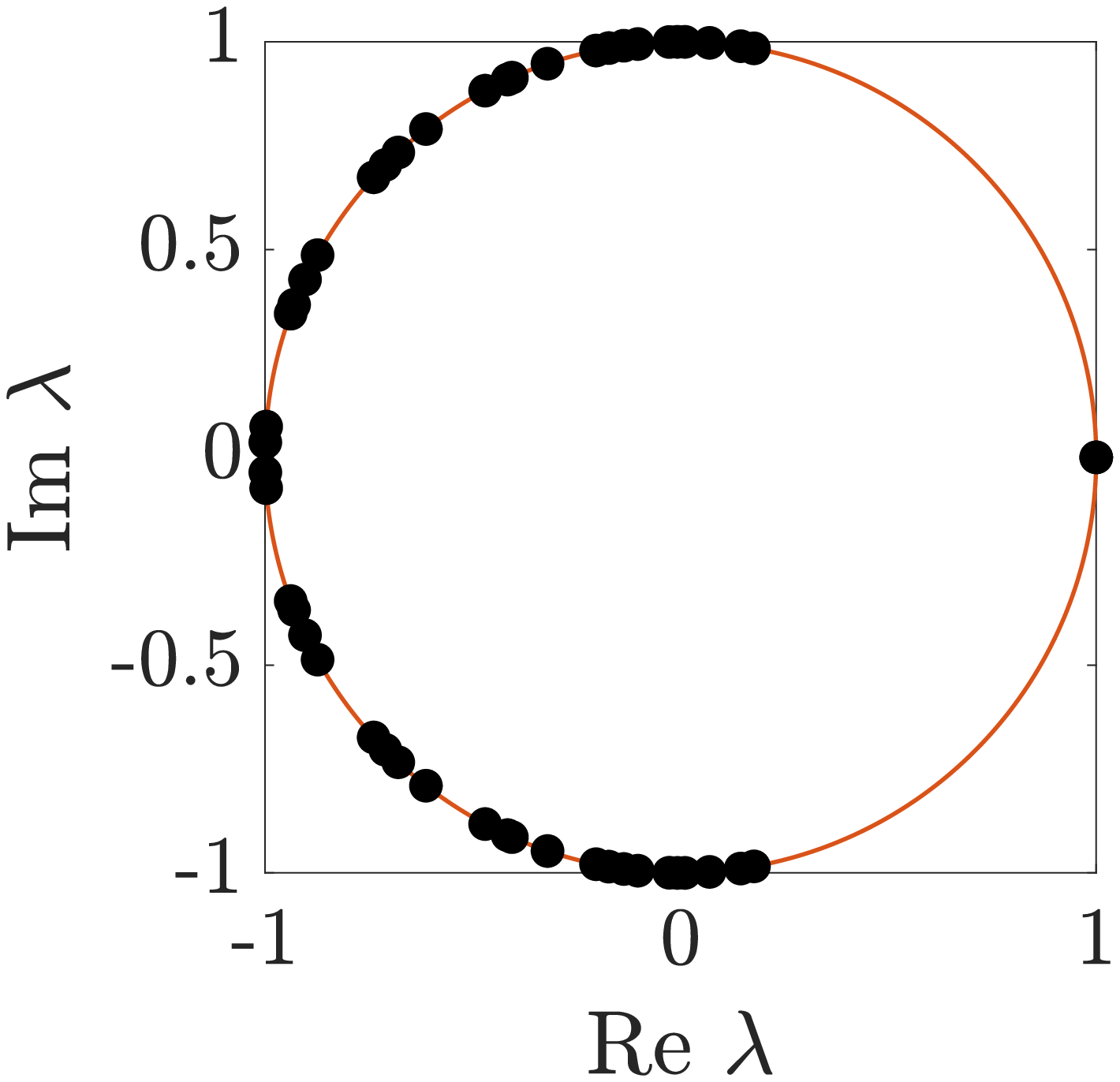}
		\label{fig:singleb}
	\end{subfigure}
	\begin{subfigure}{0.45\linewidth}
		\caption{}
		\includegraphics[width=7.5cm]{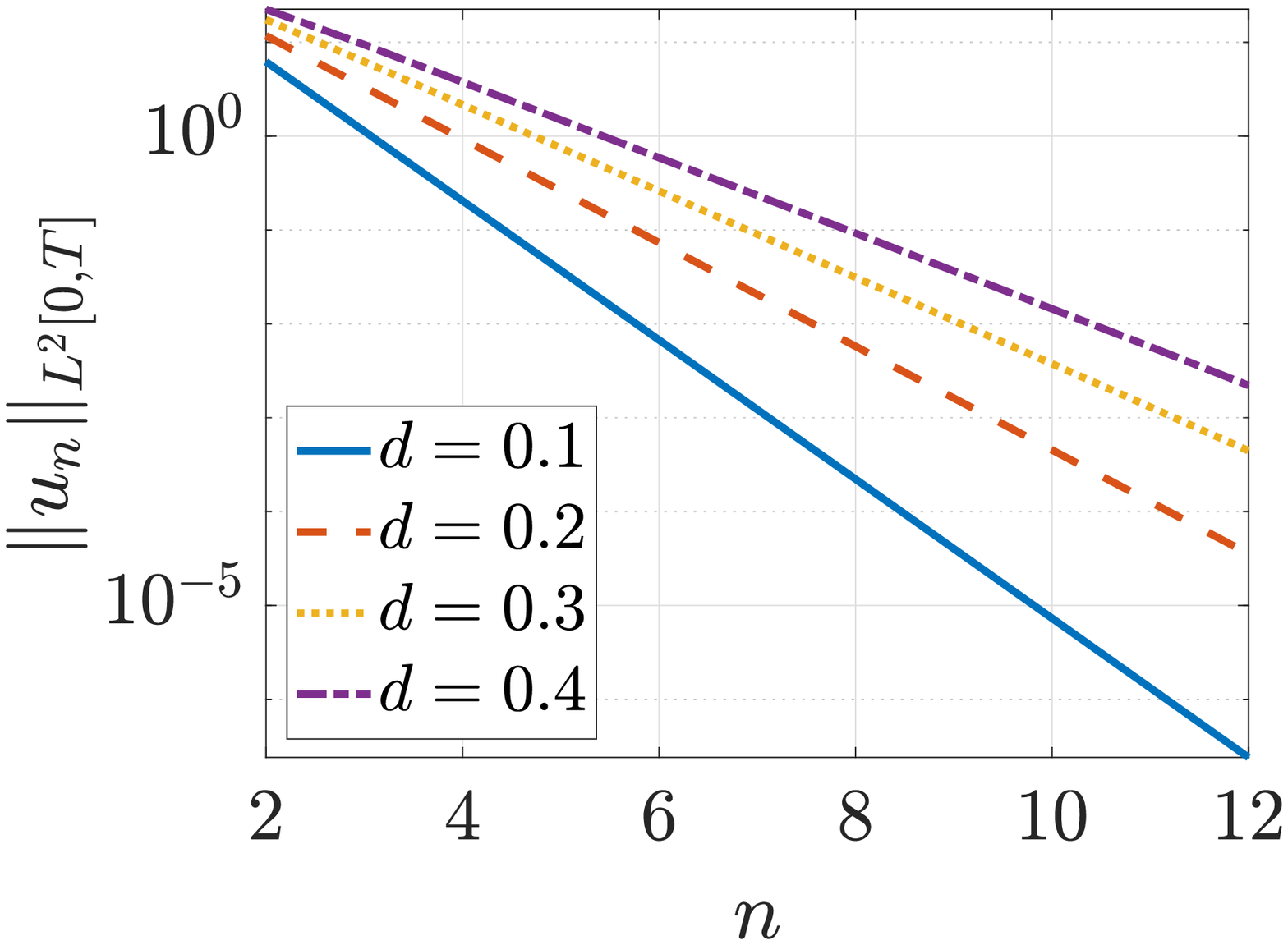}
		\label{fig:singlec}
	\end{subfigure}
	\begin{subfigure}{0.45\linewidth}
		\caption{}
		\includegraphics[width=7.5cm]{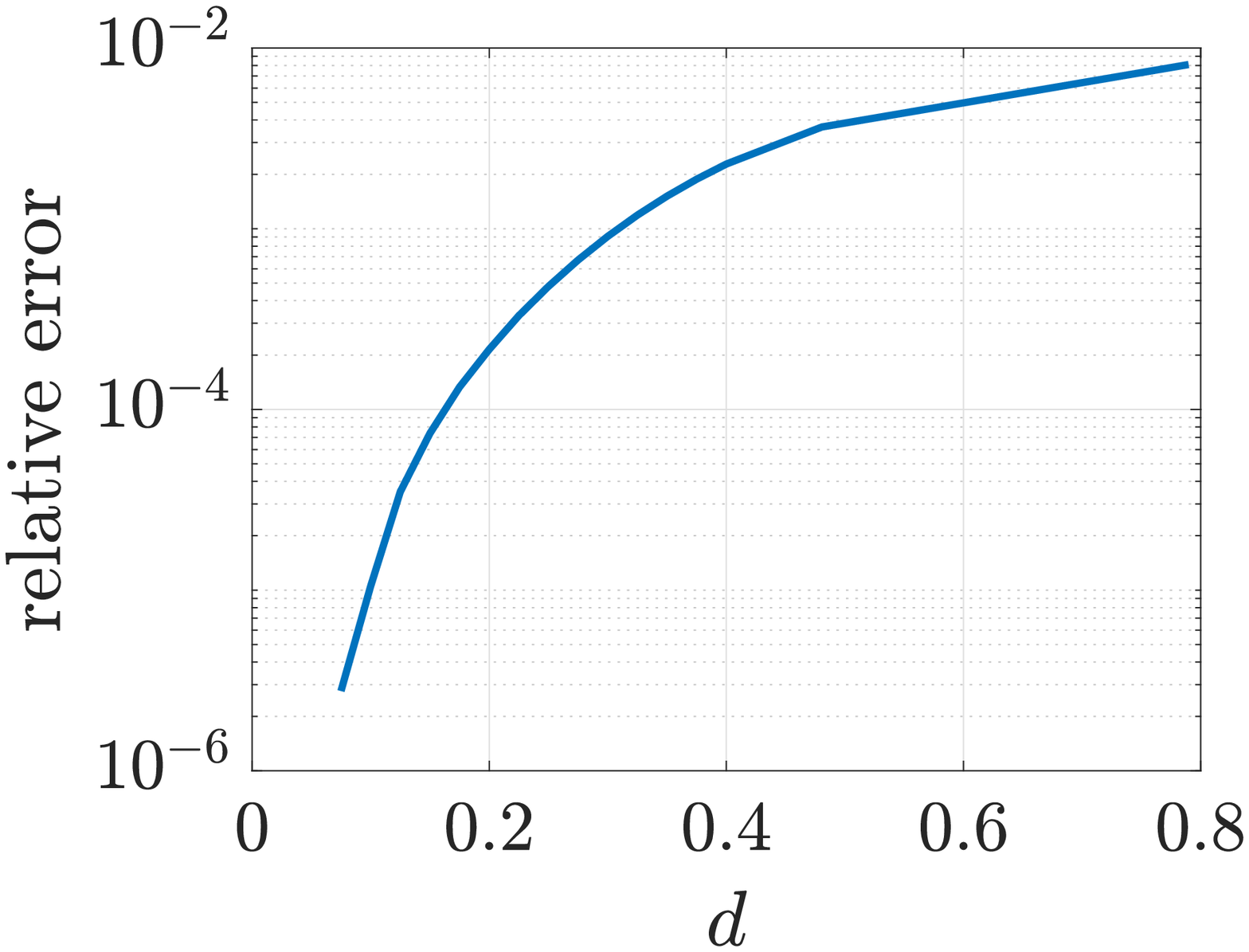}
		\label{fig:singled}
	\end{subfigure}
	\end{center}
	\caption{(a) Initial condition $u_n(0)$, and (b) Floquet spectrum, for primary breather for $d = 0.25$ for discrete sine-Gordon. (c) Semilog plot of $L^2$ norm of solution $u_n(t)$ over one period $[0,T]$ vs. the lattice index $n$ for varying coupling parameter $d$. (d) Relative error of exponential decay rate seen in (c) versus that predicted by equation \cref{eq:U1decayest}. }
	\label{fig:single}
\end{figure}

\begin{figure}
	\begin{center}
	\begin{subfigure}{0.45\linewidth}
		\caption{}
		\includegraphics[width=7.5cm]{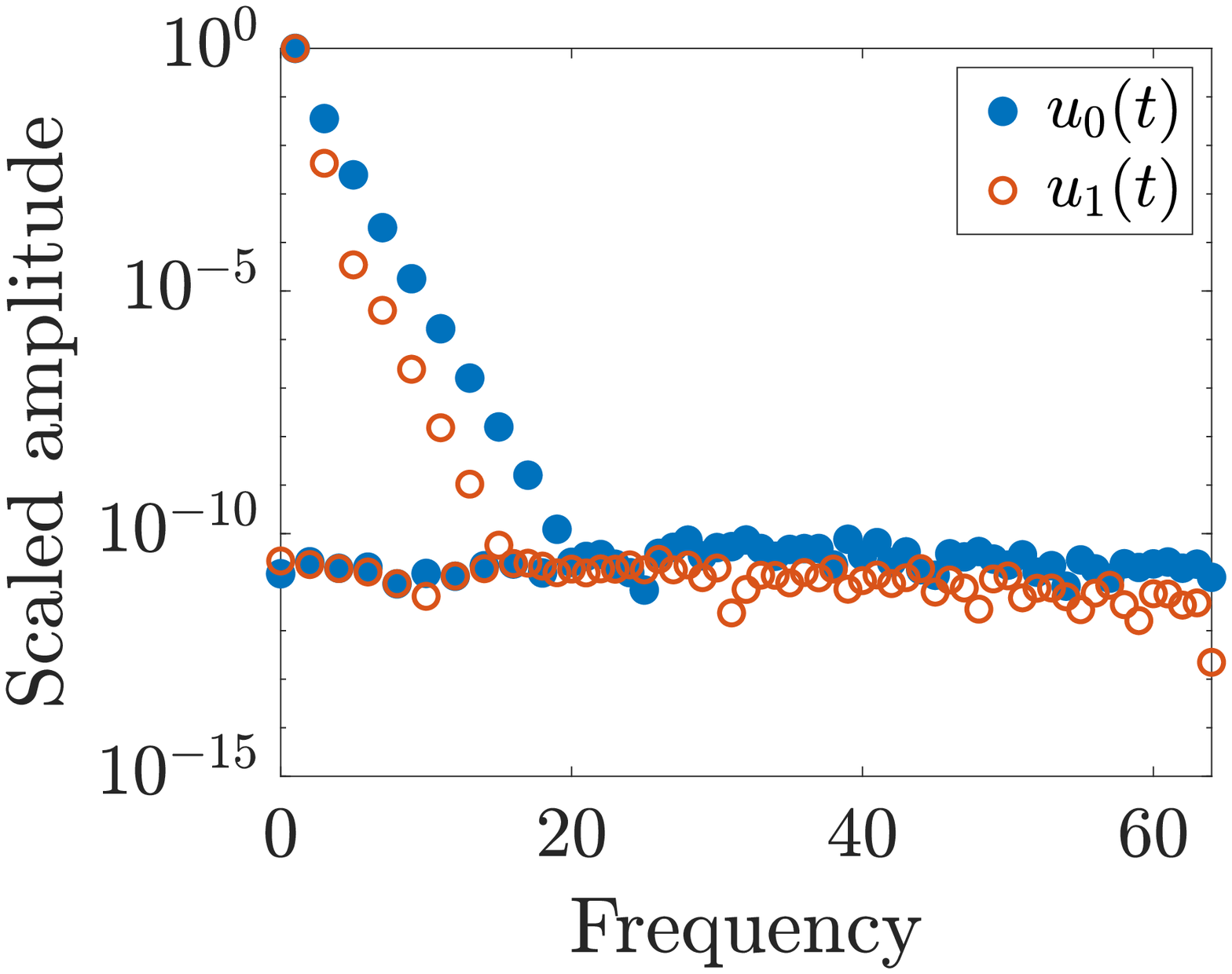}
		\label{fig:freqdecaya}
	\end{subfigure}
	\begin{subfigure}{0.45\linewidth}
		\caption{}
		\includegraphics[width=7.5cm]{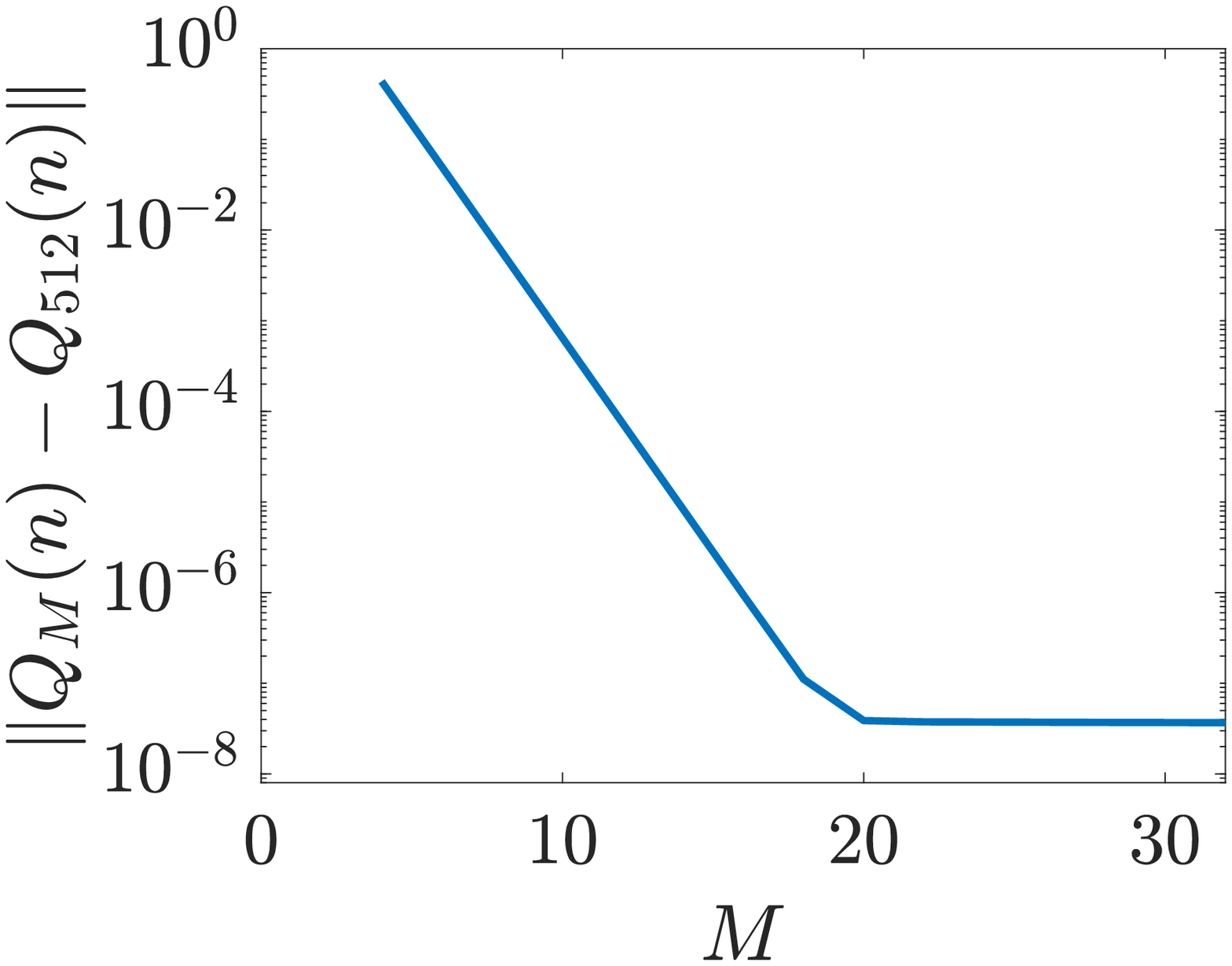}
		\label{fig:freqdecayb}
	\end{subfigure}
	\end{center}
	\caption{(a) Semilog plot showing the decay of the coefficients of the Fourier series expansion for a single-site breather at the center site $u_0(t)$ and neighboring site $u_1(t)$ for discrete sine-Gordon. The vertical axis is rescaled so that the fundamental frequency has amplitude of 1. 
	(b) Semilog plot showing the $L^2$ norm difference between the solution $Q_M(n)$ truncated after $M$ Fourier nodes and the solution $Q_{512}(n)$ with 512 Fourier nodes. Coupling parameter $d=0.25$ in both cases.}
	\label{fig:freqdecay}
\end{figure}

\cref{fig:double} shows the initial conditions (solution at $t=0$) for the out-of-phase (left) and in-phase (right) double breather for $N_1 = 6$ at $d = 0.25$, together with their Floquet spectra and the eigenfunction corresponding to the interaction eigenmode. The out-of-phase double breather has a pair of interaction eigenmodes on the unit circle, thus is spectrally neutrally stable, while the in-phase double breather has a pair of real interaction eigenmodes off of the unit circle and on the positive real axis, thus is spectrally unstable. This same pattern holds for all $N_1$, both even and odd. We note that for the solvability condition, we have $K < 0$ for the sine-Gordon potential, and that the quantity $b_1 < 0$ in \cref{eq:inteigsdouble} for both even and odd $N_1$. Since $K$ and $b_1$ have the same sign for all $N_1$, this agrees with the stability predictions following equation \cref{eq:inteigsdouble}.

\cref{fig:eigerrora} shows the exponential decay of the Floquet exponents $\lambda$ (as computed from the monodromy matrix) for in-phase double breathers as the separation distance $N_1$ is increased. (Recall that the Floquet exponents $\lambda$ and Floquet multipliers $\mu$ are related by $\mu = e^{\lambda T}$).
This is similar for out-of-phase double breathers (not shown). The remaining panels in \cref{fig:eigerror} show the relative error between the computation of these Floquet eigenmodes using the formula \cref{eq:inteigsdouble} and the computation of these modes from the monodromy matrix. As expected, the relative error increases with the coupling parameter $d$, and decreases with the separation distance $N_1$. \cref{fig:SGintersite} shows the initial condition $u_n(0)$ and Floquet spectra for two special two-site breathers, for which the two excited sites at the AC limit are in adjacent lattice nodes. We will call these inter-site centered breathers, in analogy to the inter-site centered soliton in DNLS \cite{Kevrekidis2009}.
As predicted by \cite{cuevas-maraver2016}, the in-phase inter-site centered breather is unstable, and the out-of-phase inter-site centered breather is spectrally stable, at least for small $d$. Since the two copies of the primary breather in the inter-site centered breather are not well separated, the eigenvalue results of \cref{sec:multi} do not apply. This is a
distinguishing feature between the theory presented herein
and the earlier results, e.g., of~\cites{Archilla2003,Koukouloyannis2009}. The latter are
well-tailored to this case, while the theory presented herein
is suitable for the case of well-separated breathers, as 
described above.

\begin{figure}
	\begin{center}
	\begin{subfigure}{0.45\linewidth}
		\caption{}
		\includegraphics[width=7.5cm]{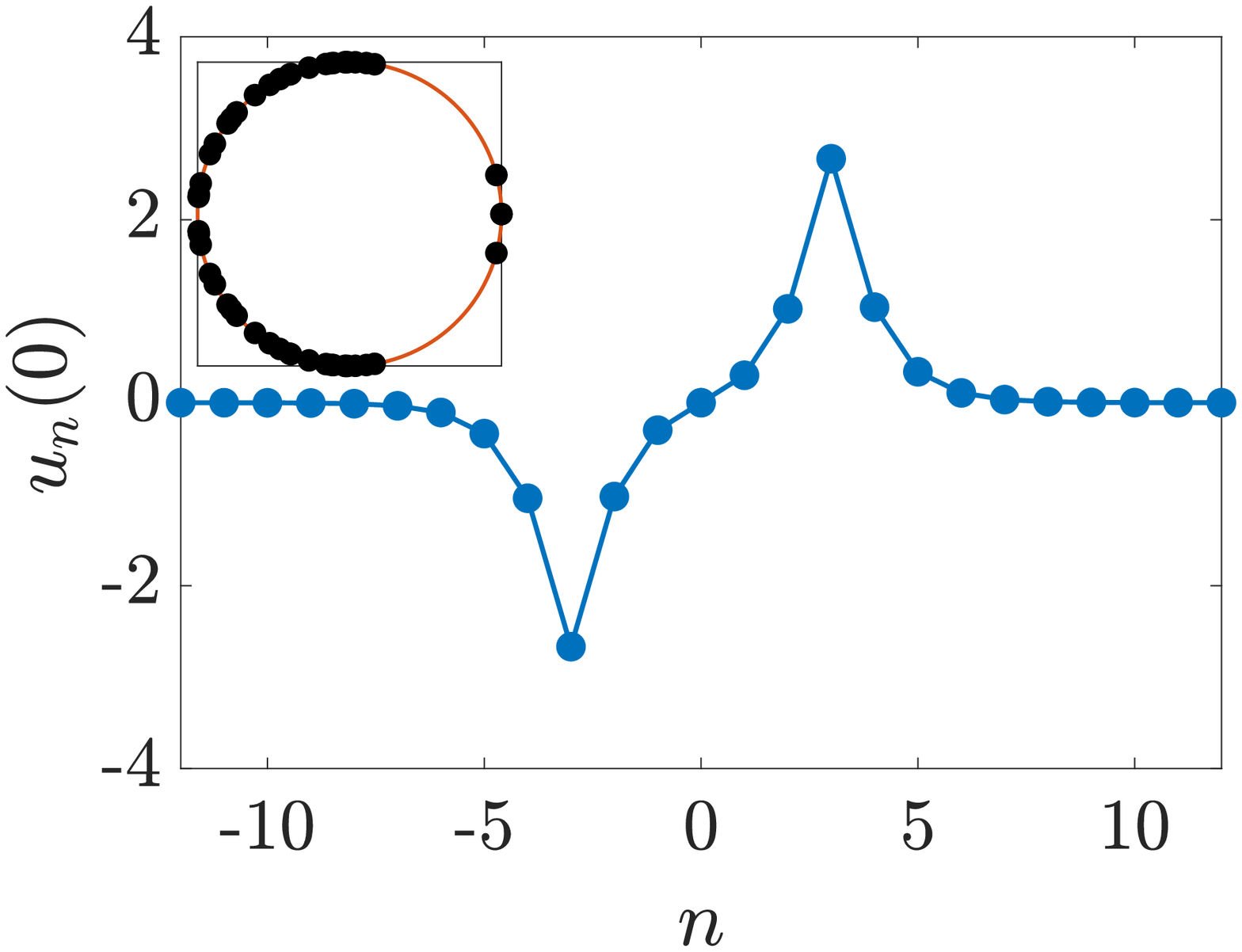} \hspace{-0.5cm}
		\label{fig:doublea}
	\end{subfigure}
	\begin{subfigure}{0.45\linewidth}
		\caption{}
		\includegraphics[width=7.5cm]{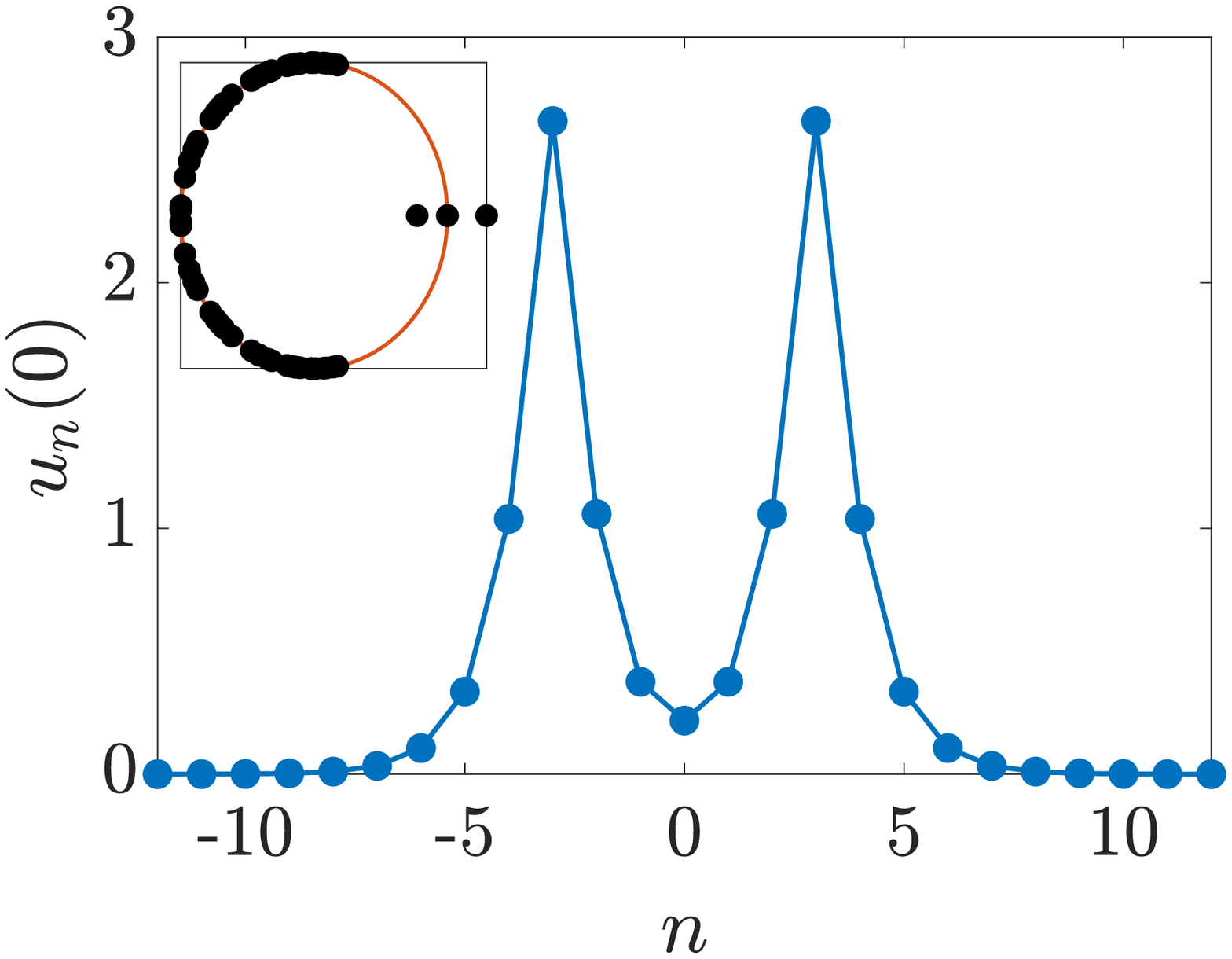}
		\label{fig:doubleb}
	\end{subfigure}
	\begin{subfigure}{0.45\linewidth}
		\caption{}
		\includegraphics[width=7.5cm]{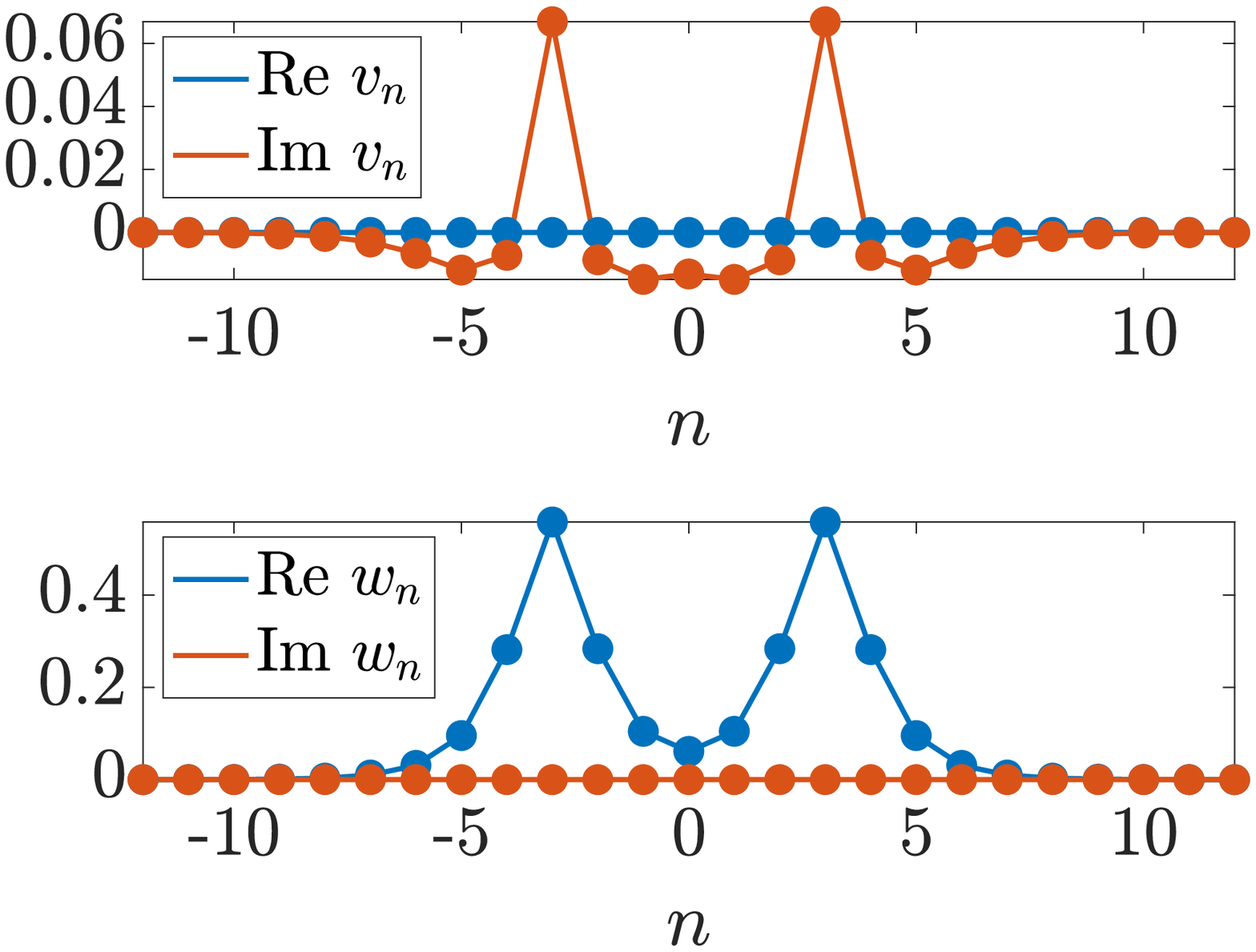} \hspace{-0.5cm}
		\label{fig:doublec}
	\end{subfigure}
	\begin{subfigure}{0.45\linewidth}
		\caption{}
		\includegraphics[width=7.5cm]{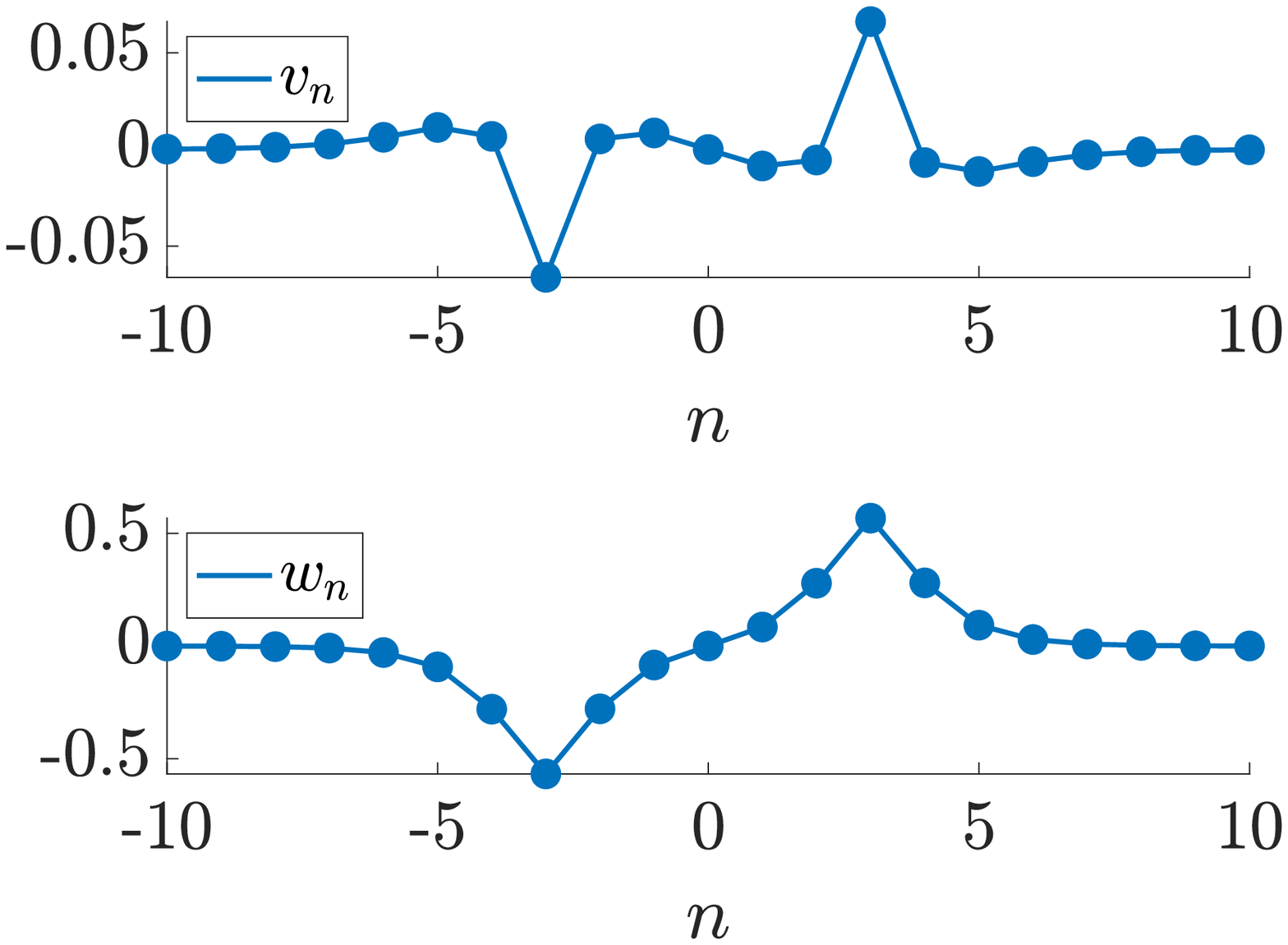}
		\label{fig:doubled}
	\end{subfigure}
	\end{center}
	\caption{Initial condition $u_n(0)$ with Floquet spectrum in inset (top) and eigenfunction $(v_n, w_n)$ corresponding to interaction eigenmode (bottom) for out-of-phase double breather (left) and in-phase double breather (right) for discrete sine-Gordon. Orange solid line in Floquet spectrum plot corresponds to unit circle. Coupling parameter $d = 0.25$, breather distance $N_1 = 6$.}
	\label{fig:double}
\end{figure}

\begin{figure}
	\begin{center}
	\begin{subfigure}{0.45\linewidth}
		\caption{}
		\includegraphics[width=7.5cm]{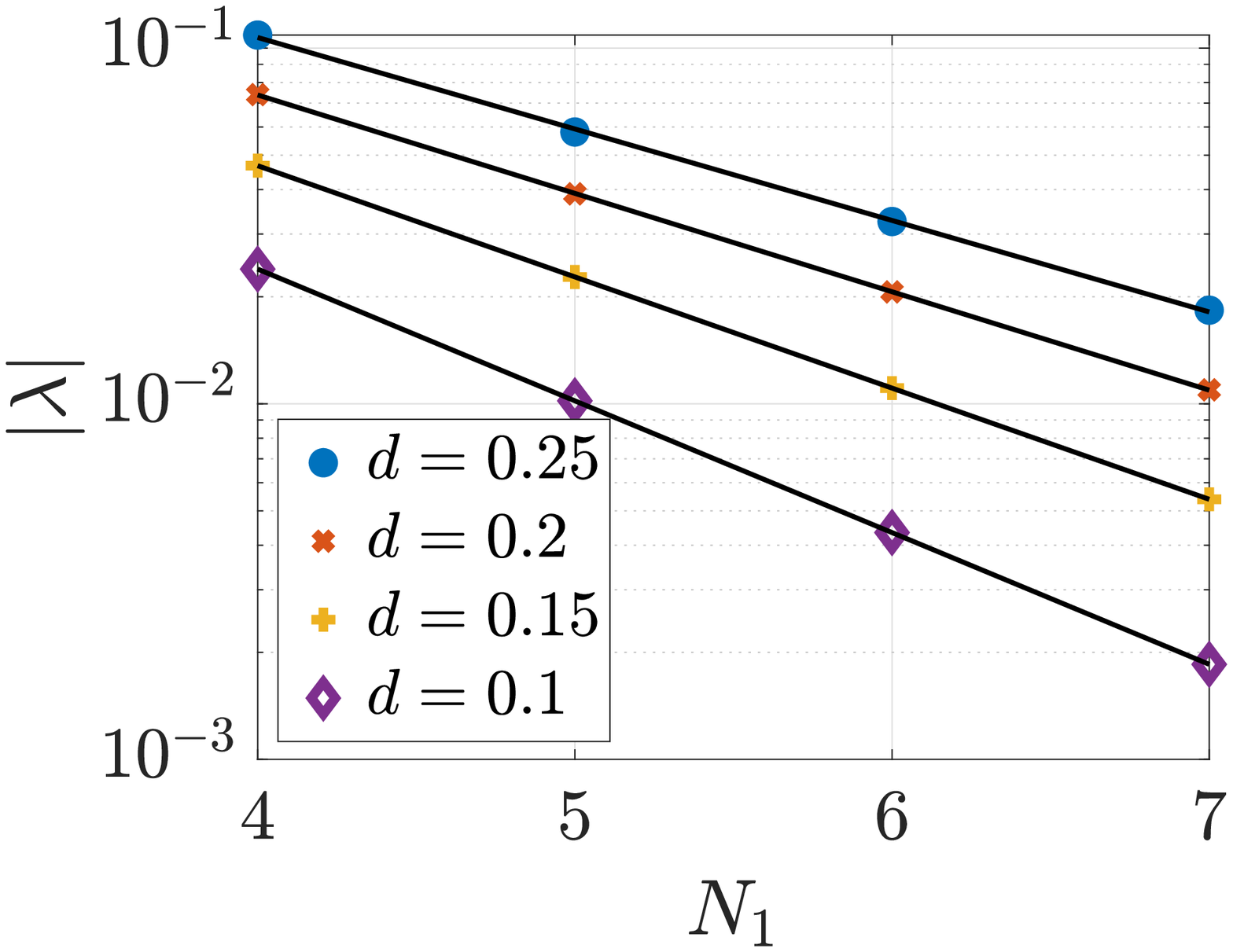}
		\label{fig:eigerrora}
	\end{subfigure}
	\begin{subfigure}{0.45\linewidth}
		\caption{}
		\includegraphics[width=7.5cm]{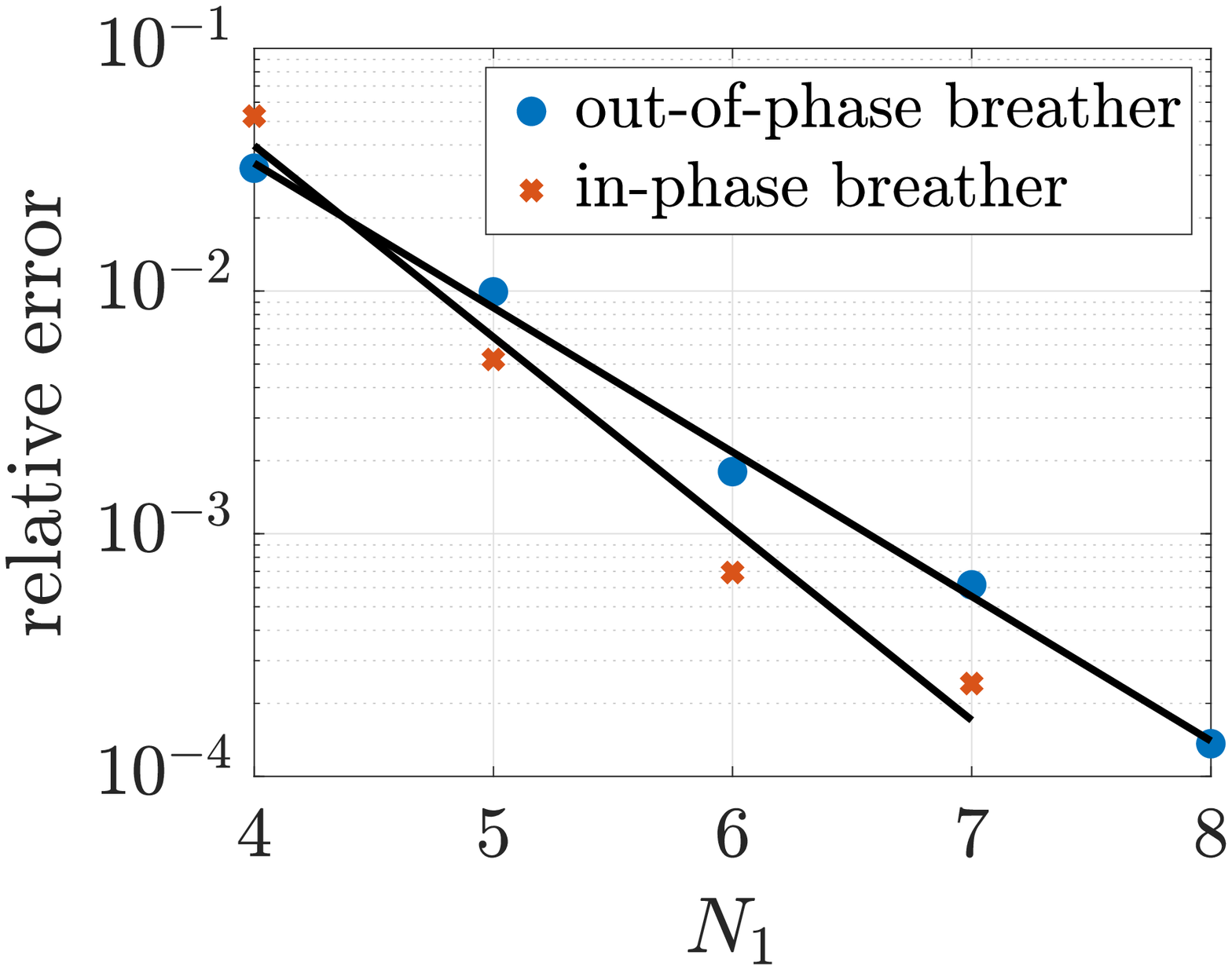} 
		\label{fig:eigerrorb} 
	\end{subfigure}
	\begin{subfigure}{0.45\linewidth}
		\caption{}
		\includegraphics[width=7.5cm]{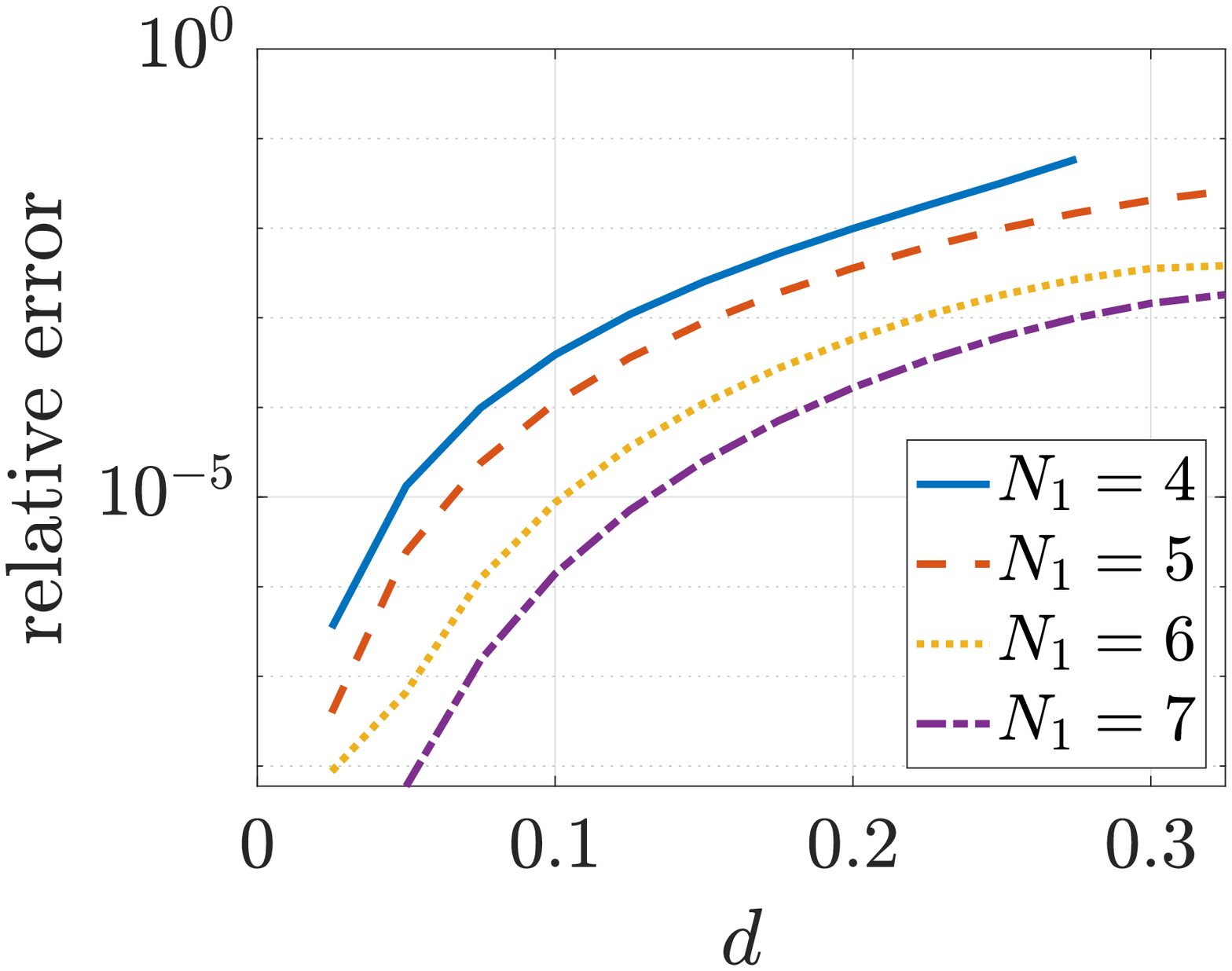} \hspace{-0.5cm}
		\label{fig:eigerrorc} 
	\end{subfigure}
	\begin{subfigure}{0.45\linewidth}
		\caption{}
		\includegraphics[width=7.5cm]{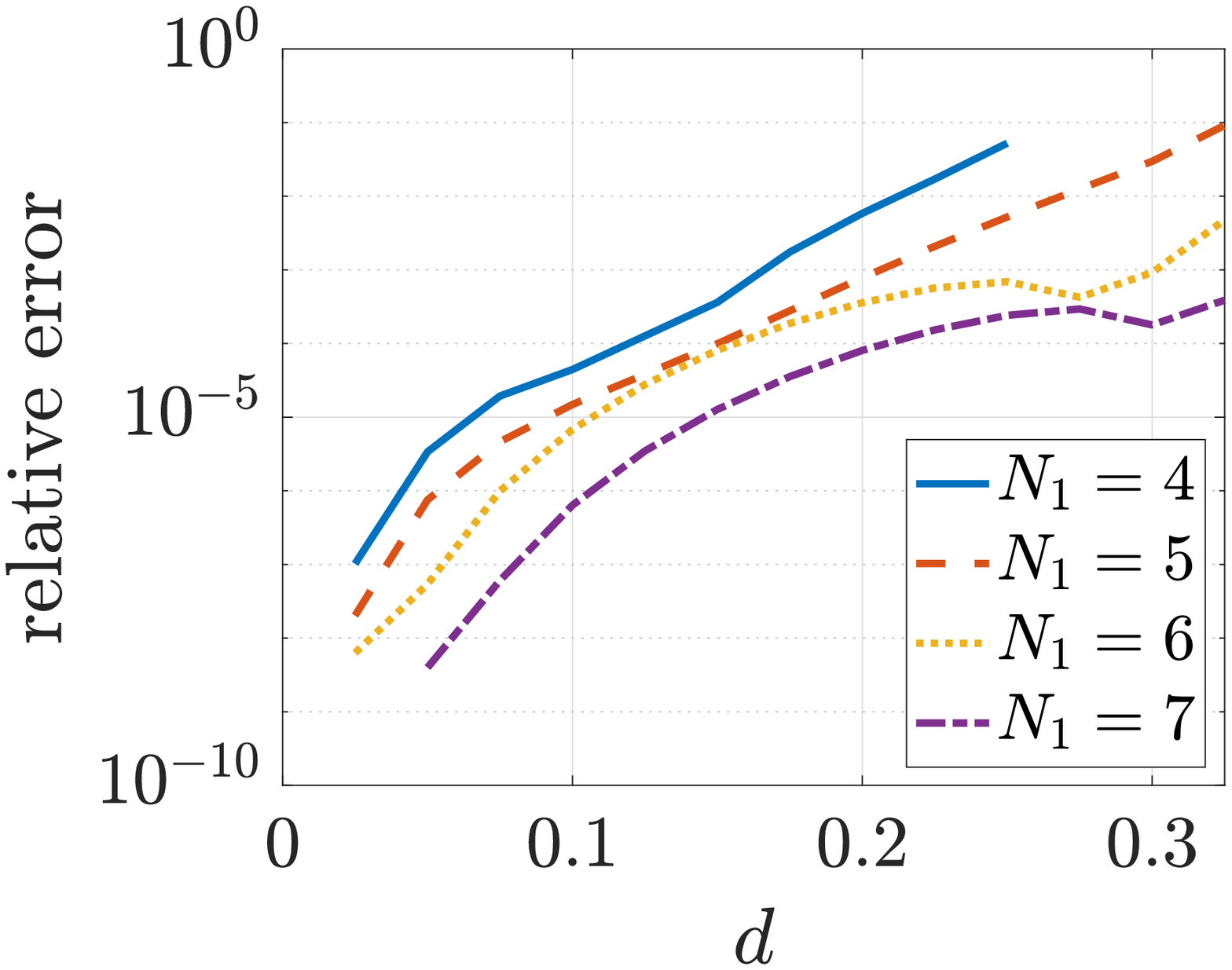} \hspace{-0.5cm}
		\label{fig:eigerrord} 
	\end{subfigure}
	\end{center}
	\caption{(a) Semilog plot of magnitude of Floquet exponent $\lambda$ corresponding to interaction eigenmode vs. $N_1$ for in-phase double breather with $d = 0.1, 0.15, 0.2, 0.25$ for discrete sine-Gordon. (b) Semilog plot of the relative error of Floquet multipliers of double breathers vs. $N_1$ for the case of 
	coupling parameter $d = 0.25$.
	Semilog plot of the relative error of interaction eigenmode computation vs. $d$ for out-of-phase double breathers (c) and in-phase double breathers (d) for breather distance $N_1 = 4,5,6,7$.}
	\label{fig:eigerror}
\end{figure}

\begin{figure}
	\begin{center}
	\begin{subfigure}{0.45\linewidth}
		\caption{}
		\includegraphics[width=7.5cm]{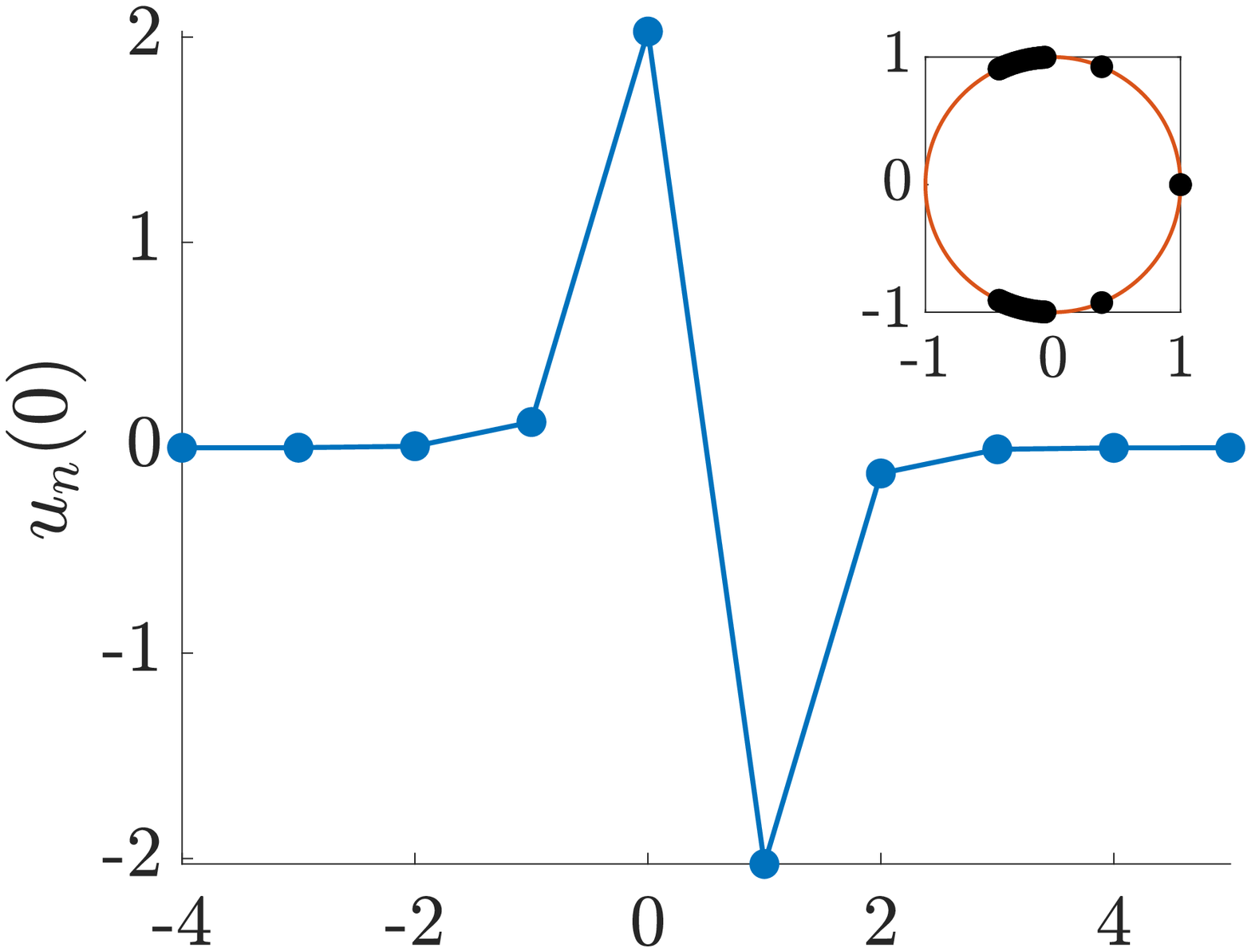} \hspace{-0.5cm}
		\label{fig:SGintersitea} 
	\end{subfigure}
	\begin{subfigure}{0.45\linewidth}
		\caption{}
		\includegraphics[width=7.5cm]{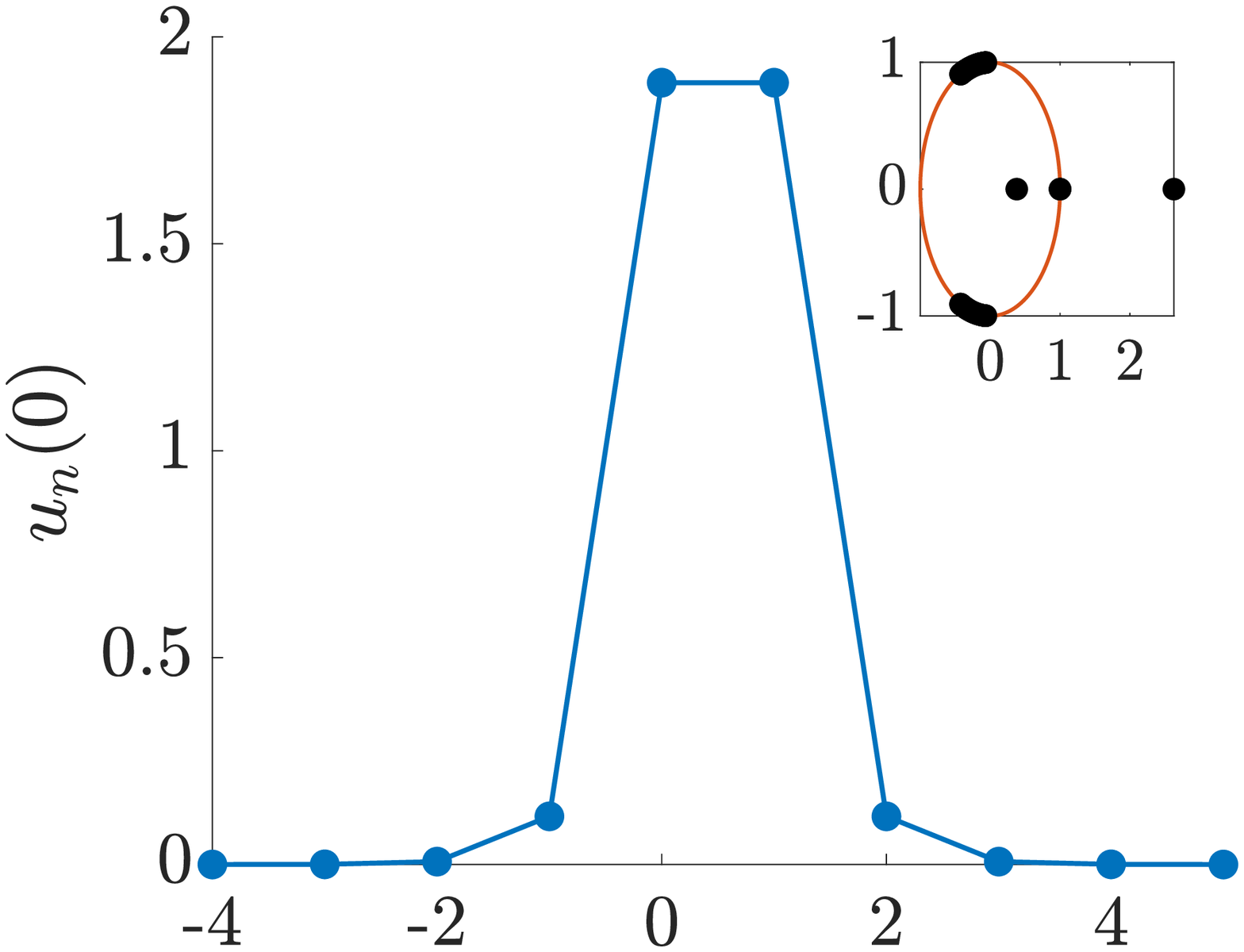} \hspace{-0.5cm}
		\label{fig:SGintersiteb} 
	\end{subfigure}
	\end{center}
	\caption{Initial condition $u_n(0)$ and Floquet spectrum (inset) for out-of-phase inter-site centered breather (a), and in-phase inter-site centered breather (b) for discrete sine-Gordon. Coupling parameter $d=0.1$.}
	\label{fig:SGintersite}
\end{figure}

For completeness here, and also per its intrinsic interest,
we show in 
\cref{fig:bifdiagSGoop1} the bifurcation diagram for an out-of-phase double breather, starting from the AC limit. The parameter continuation starts on the lower branch, where the double breather has a pair of interaction eigenmodes on the unit circle (label 1). The upper branch is a double breather whose constituent breathers
are out-of-phase, yet each encompasses two in-phase adjacent excited sites (label 7), i.e. two inter-site centered breathers, which we recall leads
to instability (\cref{fig:SGintersiteb}); this breather has two pairs of real interaction eigenmodes, in addition to a pair of interaction eigenmodes on the unit circle. The middle branch is an asymmetric double breather, comprising one single-site breather and another intersite-centered breather which is out-of-phase with it (label 8); this breather has one pair of real interaction eigenmodes, and one pair of interaction eigenmodes on the unit circle. There is a corresponding branch (not shown) where the order of the two breathers is reversed. As $d$ is increased along the lower branch, a pair of internal modes appears on the unit circle (label 2); these have opposite Krein signature from the nearby interaction eigenmode. These collide and move off of the unit circle (label 3) in the form of a complex quartet (corresponding to an oscillatory instability). At this point, a second pair of internal modes has appeared on the unit circle, which then passes between the first pair of Floquet modes (not shown). By label 4, the pair of eigenmodes which left the unit circle has rejoined the unit circle. Between label 4 and label 5, the second pair of internal modes collides (with the branch traced by label 8 and label 9) at (1,0) in a {\it subcritical} pitchfork bifurcation and moves off of the unit circle. This pitchfork is symmetry-breaking, and produces the asymmetric, middle branches of the bifurcation diagram, indicated by the
dashed line. Between label 5 and label 6, there is a turning point (fold bifurcation), in which the first pair of internal modes collides and moves off of the unit circle. 
This collision represents a saddle-center bifurcation with the
branch corresponding to label 7, designated by the dotted line, a branch
that bears generically two real multiplier pairs. We find it
quite relevant to recall that
qualitatively similar pitchfork bifurcations are seen for multi-kinks in DKG (\cite{Parker2021}*{Figure 4}) and for vortex pairs in the 2-dimensional DNLS equation (\cite{Bramburger2020}*{Figure 4}).
Indeed, this seems to represent a generic mode of disappearance
in such discrete nonlinear lattice dynamical systems of states
that do not persist all the way to the continuum limit.

\begin{figure}
	\hbox{
	\hspace{-2cm}
	\includegraphics[width=20cm]{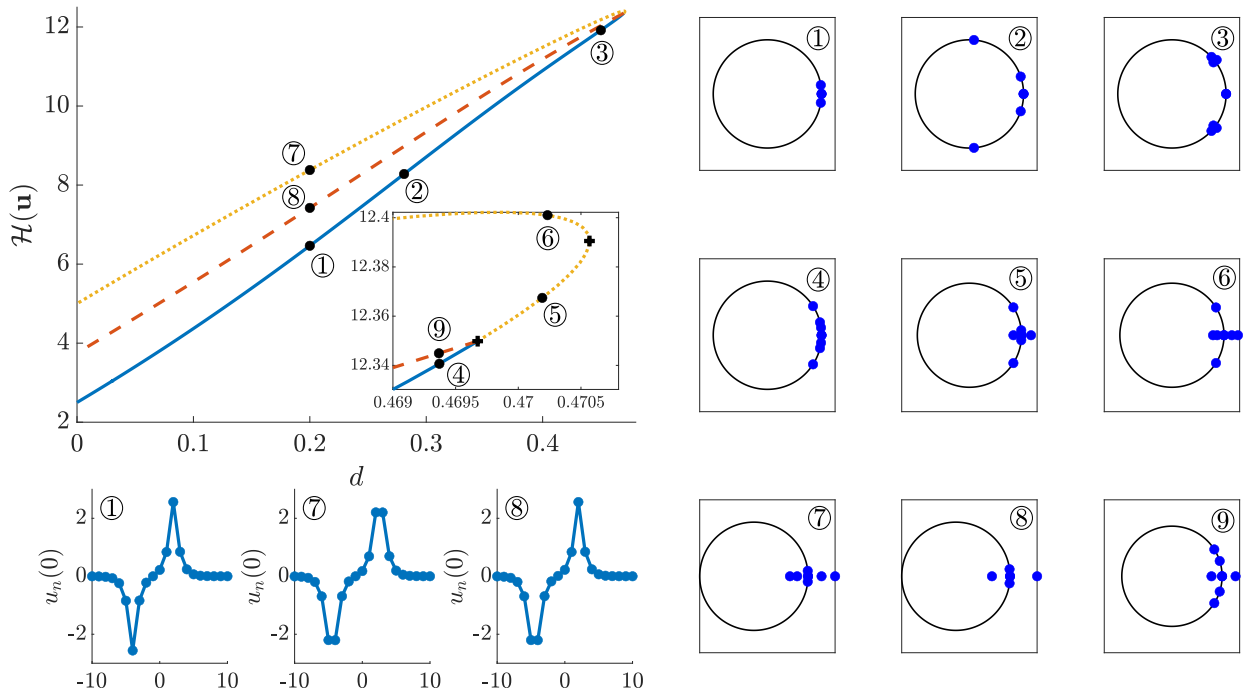} 
	}
	\caption{Bifurcation diagram plotting energy $\mathcal{H}(\uvec)$ vs. $d$ for out-of-phase double breather in the soft sine-Gordon potential with $N_1 = 6$. Points of interest are marked with black dots and labeled with circled numbers, which correspond to Floquet spectra (and in select cases also to the configurations at the bottom left). Continuous spectrum bands are not shown for clarity. Bifurcations are indicated in the inset with a black cross.}
	\label{fig:bifdiagSGoop1}
\end{figure}

\cref{fig:bifdiagSGip1} shows the bifurcation diagram for an in-phase double breather, again starting from the AC limit. The diagram is qualitatively similar, in that two bifurcations (a pitchfork
and a saddle-center) result from collisions of internal modes at (1,0). The main difference is that, since the starting in-phase breather has a pair of real interaction eigenmodes, these are not involved in any collisions with the internal modes. In both cases, there is a turning point (fold bifurcation) in the bifurcation diagram at a critical value $d_0$ of the coupling parameter $d$, at which point the parameter continuation reverses direction in $d$, again marking a saddle-center 
bifurcation. For the out-of-phase double breather, this occurs after the pitchfork bifurcation of the site-centered breathers with 
the asymmetric ones (see inset in \cref{fig:bifdiagSGoop1}).
On the contrary, for the in-phase double breather, 
the pitchfork bifurcation occurs between the branch comprising a pair of inter-site centered breathers and the asymmetric branch bearing one on-site and one inter-site centered
breather (see inset in \cref{fig:bifdiagSGip1}). A plot of the turning point $d_0$ vs. the separation distance $N_1$ suggests that $d_0$ increases linearly with $N_1$ for both the in-phase and the out-of-phase double breather (\cref{fig:SGd0}).

\begin{figure}
	\hbox{
	\hspace{-2cm}
	\includegraphics[width=20cm]{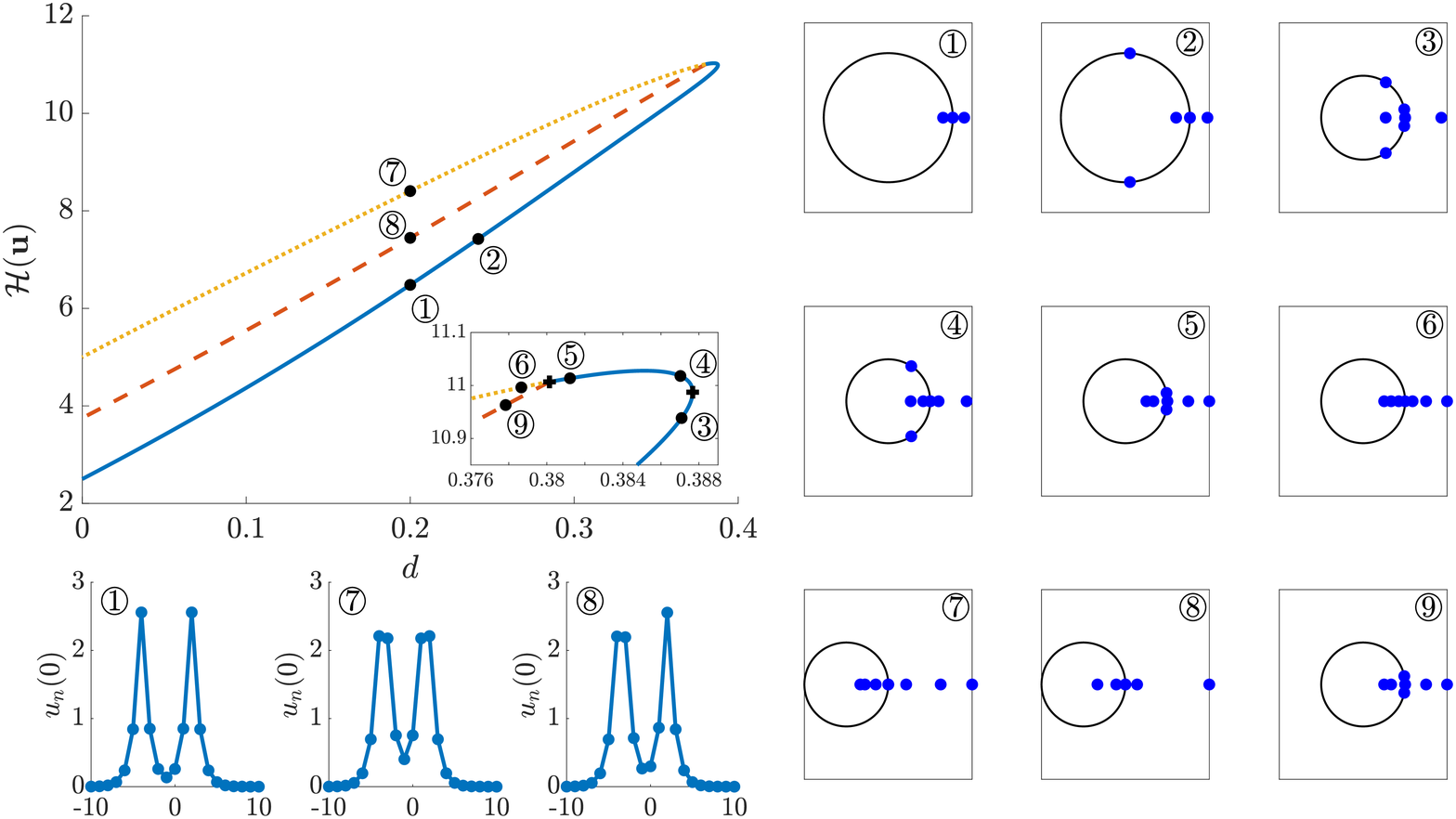} 
	}
	\caption{Bifurcation diagram plotting the energy $\mathcal{H}(\uvec)$ vs. $d$ for in-phase double breather in soft sine-Gordon potential with $N_1 = 6$. Points of interest are marked with black dots and labeled with circled letters, which correspond to Floquet spectra. 
	Again, see bottom left for select spatial configurations.
	Continuous spectrum bands are not shown for clarity. Bifurcations are indicated with a black cross. }
	\label{fig:bifdiagSGip1}
\end{figure}

\begin{figure}
	\begin{center}
	\includegraphics[width=7.5cm]{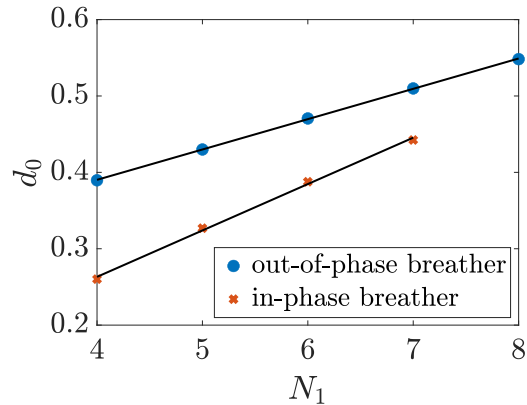} 
	\end{center}
	\caption{Turning point of parameter continuation $d_0$ vs. separation distance $N_1$ for out-of-phase and in-phase double breathers for discrete sine-Gordon.}
	\label{fig:SGd0}
\end{figure}

Finally, we demonstrate the effect of the interaction eigenmodes on the dynamics of the double breathers by performing long-term numerical timestepping experiments. For a timestepping scheme, we use a symplectic and symmetric implicit Runge-Kutta method \cite{HairerBook} to preserve the symplectic structure of the Hamiltonian system. Specifically, we use the Matlab implementation of the \texttt{irk2} scheme of order 12 from \cite{Hairer2003}.
For each experiment, we perturb the stationary breather by adding the eigenvector $(v_n, w_n)$ corresponding to the interaction eigenmode $\mu$, multiplied by a small factor $\delta$. (If the eigenvector is complex, we use the real part of the eigenvector for the perturbation).
For a double breather $u_n(t)$, let $u_L(t) = u_{-N_1^-}(t)$ and $u_R(t) = u_{N_1^+}(t)$ be the center sites (i.e. the sites of maximal excitation) of the left and the right breather, respectively. First, we look at the out-of-phase double breather (\cref{fig:timestepSGstablea}). We choose the coupling parameter $d=0.2$ and separation distance $N_1 = 4$ so that we avoid the (nonlinearity induced) 1:2 resonance between the interaction eigenmode and the continuous spectrum bands, which leads to a nonlinear instability (see \cite{cuevas-maraver2016} for details). 
Although we expect that nonlinear instabilities due to these higher order resonances will occur, they will be manifested at times beyond the time intervals of our simulations. 
When the out-of-phase double breather is perturbed, the peak amplitudes of $u_L(t)$ and $u_R(t)$ oscillate in opposite directions with frequency given by the imaginary part of $\log(\mu)/T$ (see \cref{fig:timestepSGstable}), with relative error less than 0.005, where $\mu$ is the Floquet multiplier on the unit circle corresponding to the interaction eigenmode. 

When the in-phase double breather is perturbed in the same way, the oscillations in $u_L(t)$ initially increase in both amplitude and period, while those in $u_R(t)$ decrease in both amplitude and period (\cref{fig:timestepSGunstabled}). (These are reversed if the perturbation is in the opposite direction). We can interpret this behavior by looking at the form of the corresponding eigenfunction (\cref{fig:double}, right). An addition of a small multiple of this eigenfunction to the double breather at $t=0$ causes one breather to increase in both amplitude and velocity (thus in energy), and the other to decrease in both amplitude and velocity. Since sine-Gordon is a soft potential, the breather which increases in energy also increases in period, and the breather which decreases in energy also decreases in period.
For separation distance $N_1 = 6$, as $t$ evolves, we see that the periods of $u_L(t)$ and $u_R(t)$ vary periodically in opposite directions (\cref{fig:timestepSGunstablec}), thus the phase differences between the two breathers vary periodically as well (\cref{fig:timestepSGunstablea} and \cref{fig:timestepSGunstableb}). The amplitudes of $u_L(t)$ and $u_R(t)$ also vary periodically in opposite directions (\cref{fig:timestepSGunstablee}). In effect, the perturbed in-phase double breather appears to be oscillating about the neutrally stable, out-of-phase double breather, a feature
that persists robustly over longer time scales. The growth rate of the difference in $\ell^2$ norm between the perturbed and unperturbed solutions is given by $\mu^{1/T}$, with a relative error of less than 0.001, where $\mu > 1$ is the larger of the Floquet multiplier pair (\cref{fig:timestepSGunstablef}). If the separation distance $N_1$ is smaller, the two breathers attract each other, and eventually merge into a localized, single-site breather state (see \cref{fig:timestepSGpplong} for separation distance $N_1 = 4$). 
This limiting solution is close to a higher energy, single-site breather, and is located between the initial starting sites for the pair of breathers. 
The energy density at lattice site $n$, which is plotted in the bottom panels of \cref{fig:timestepSGpplong}, is given by
\begin{equation}\label{eq:Hn}
h_n = \frac{1}{2}v_n^2 + V(u_n) 
+ \frac{d}{4}\left[ (u_n - u_{n+1})^2 + (u_n - u_{n-1})^2\right],
\end{equation}
where $v_n = \dot{u}_n$, so that the Hamiltonian \cref{eq:H} is the sum of the energy densities over the entire lattice. 
A plot of the total energy vs. time (\cref{fig:timestepSGpplongb}) indicates that energy is conserved by the numerical scheme over the simulation time, providing validation of this result.

\begin{figure}
	\begin{center}
	\begin{subfigure}{0.3\linewidth}
		\caption{}
		\includegraphics[width=5.4cm]{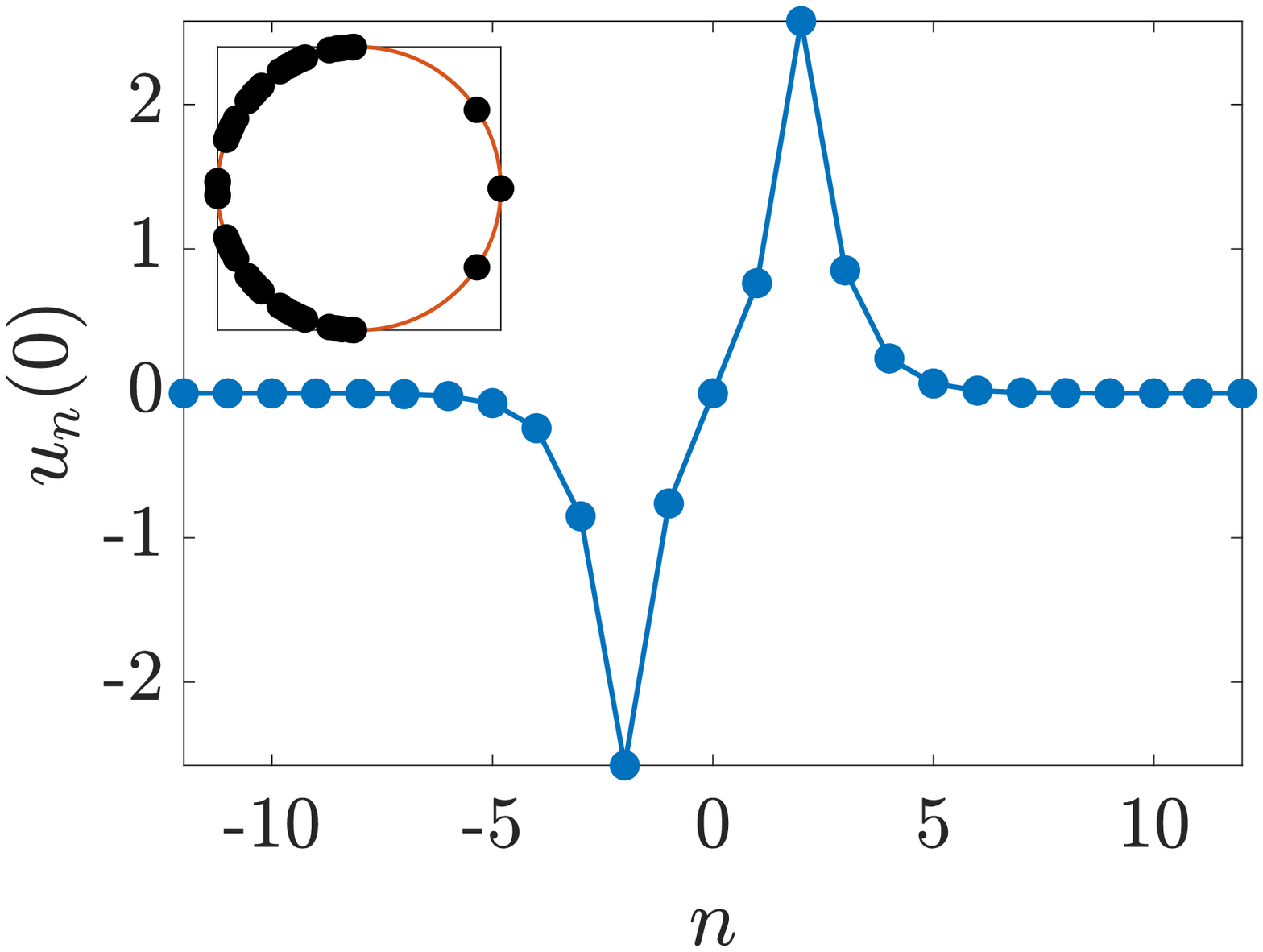}
		\label{fig:timestepSGstablea}
	\end{subfigure}
	\begin{subfigure}{0.3\linewidth}
		\caption{}
		\includegraphics[width=5.4cm]{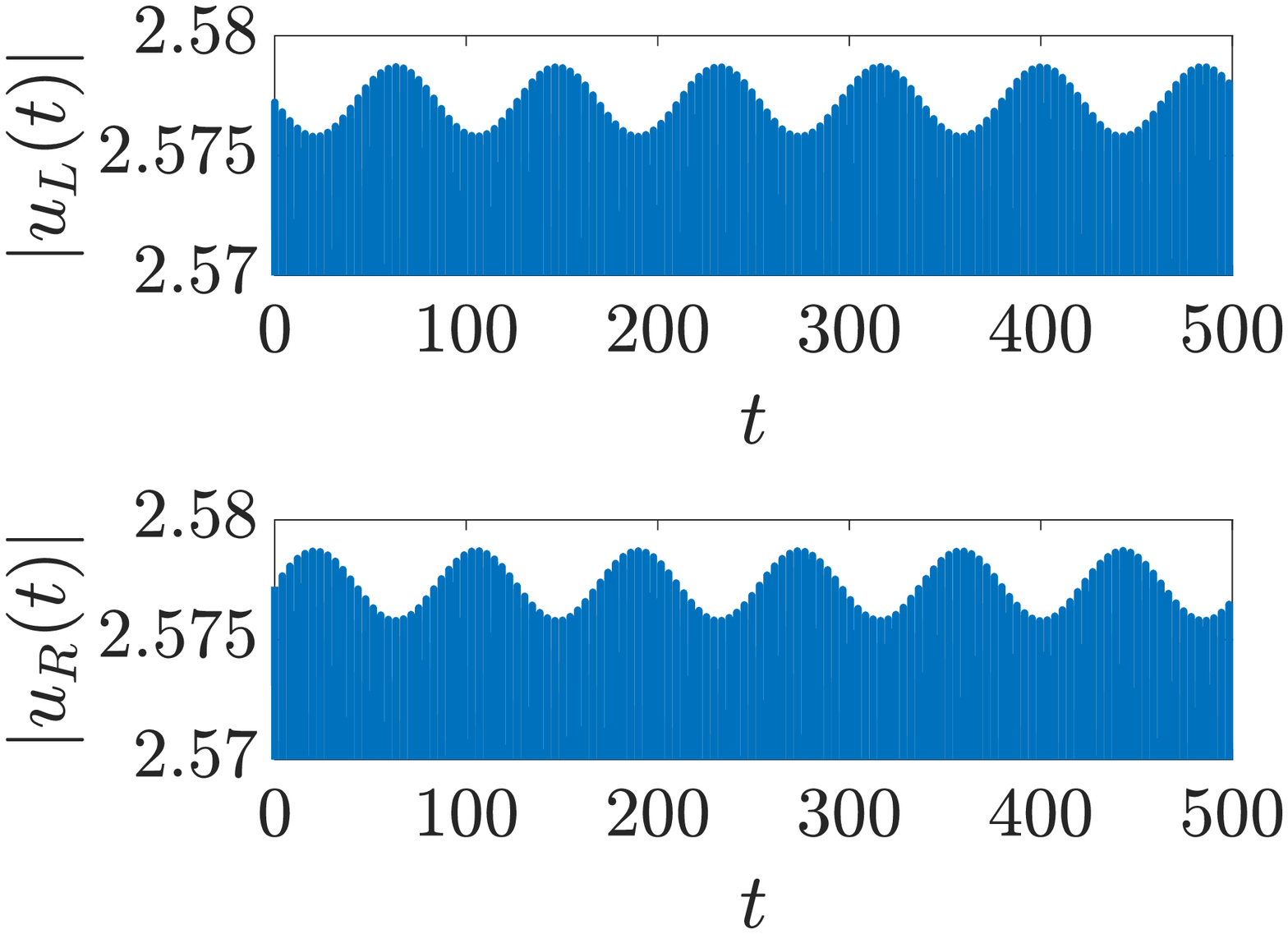}
		\label{fig:timestepSGstableb}
	\end{subfigure}
	\hspace{0.1cm}
	\begin{subfigure}{0.3\linewidth}
		\caption{}
		\includegraphics[width=5.4cm]{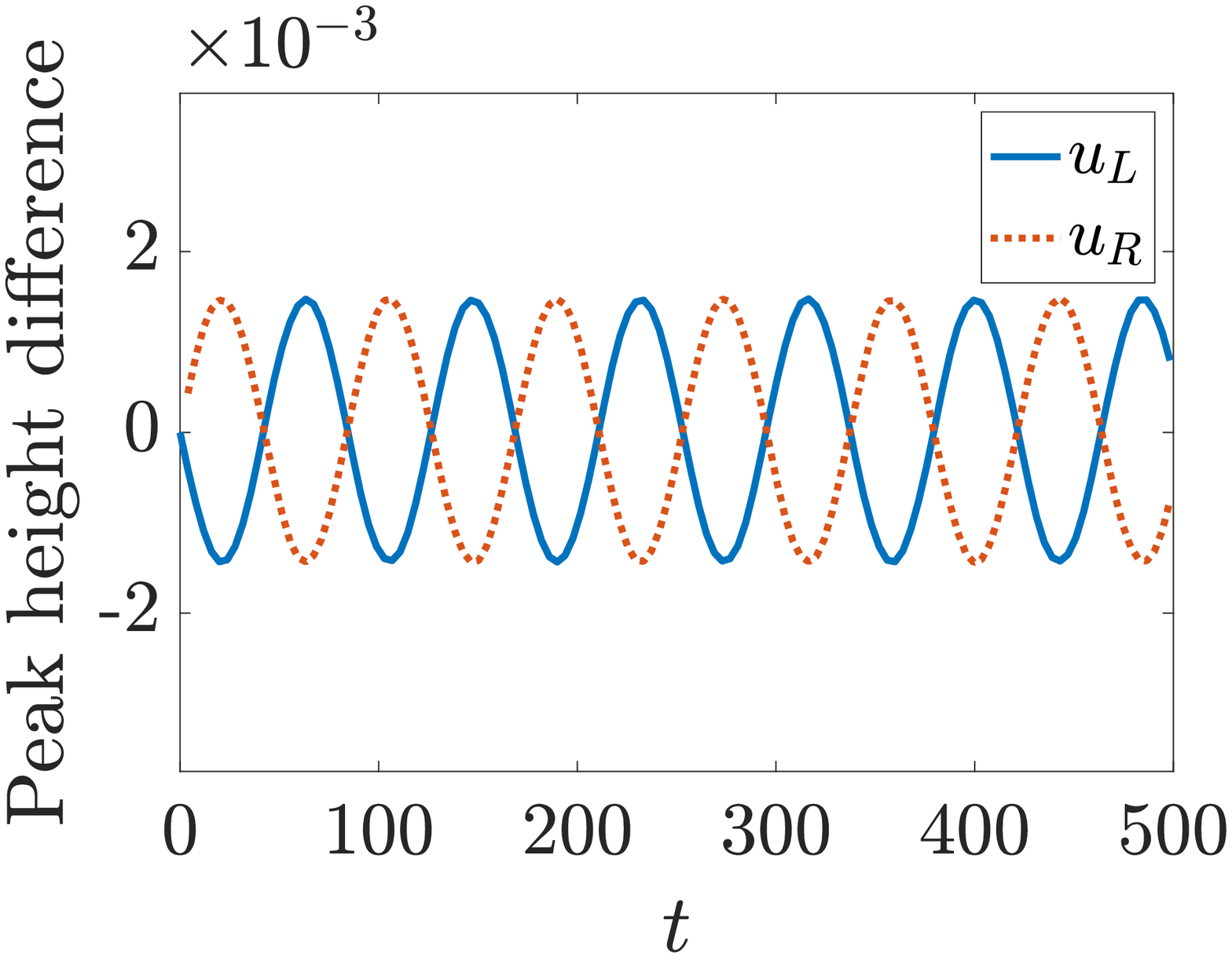}
		\label{fig:timestepSGstablec}
	\end{subfigure}
	\end{center}
	\caption{(a) Out-of-phase double breather with the Floquet spectrum shown in the inset for discrete sine-Gordon. Time evolution of the perturbation of the out-of-phase breather, showing the amplitude of the peaks of $u_L$ and $u_R$ vs. $t$ (b), and difference in this amplitude from the unperturbed breather (c).  
	Coupling parameter $d=0.2$, separation distance $N_1 = 4$, perturbation parameter $\delta = 0.01$.}
	\label{fig:timestepSGstable}
\end{figure}

\begin{figure}
	\begin{center}
	\begin{subfigure}{0.3\linewidth}
		\caption{}
		\includegraphics[width=5.25cm]{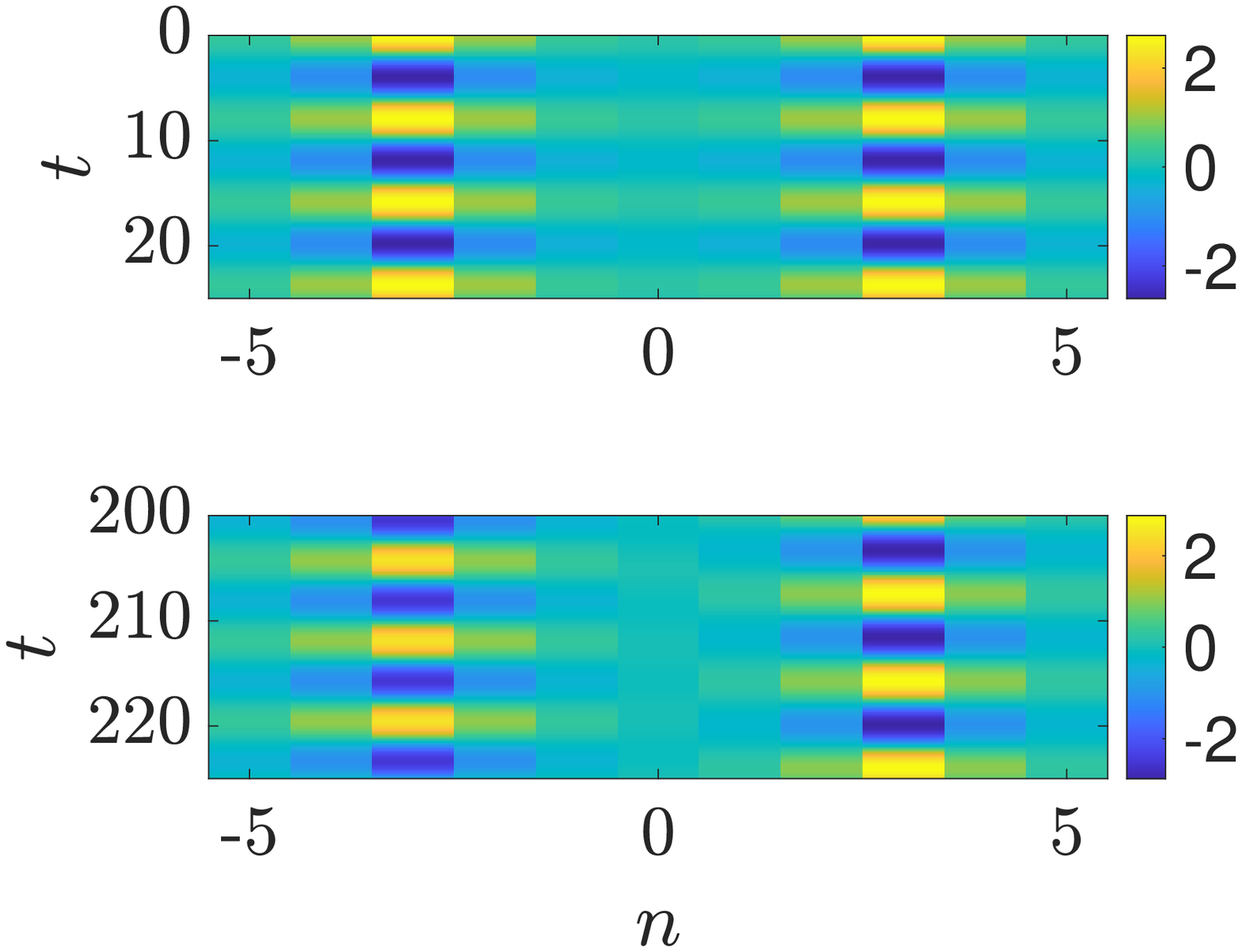} \hspace{-0.5cm}
		\label{fig:timestepSGunstablea}
	\end{subfigure}
	\begin{subfigure}{0.3\linewidth}
		\caption{}
		\includegraphics[width=5.25cm]{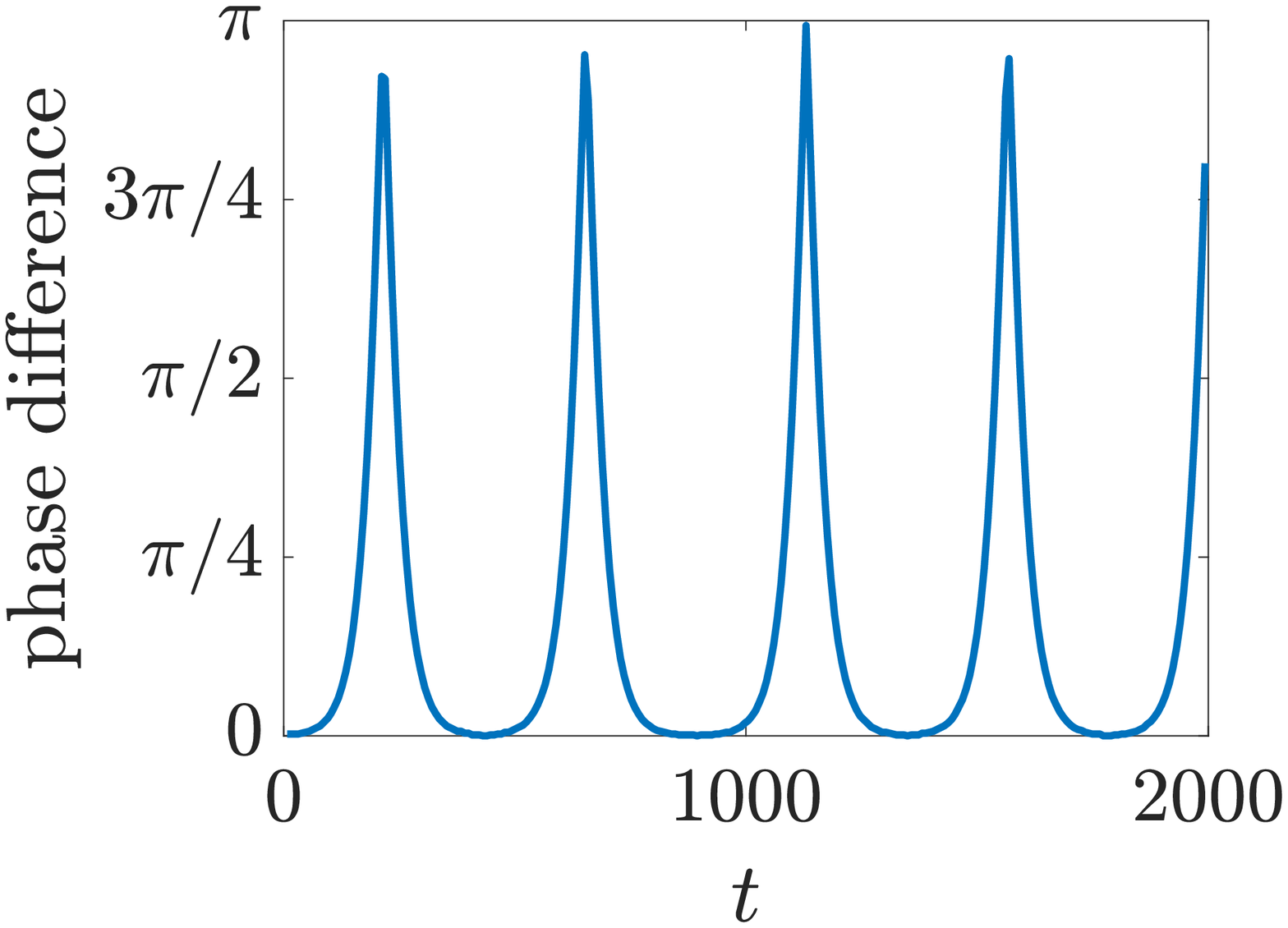} \hspace{-0.5cm}
		\label{fig:timestepSGunstableb}
	\end{subfigure}
		\begin{subfigure}{0.3\linewidth}
		\caption{}
		\includegraphics[width=5.25cm]{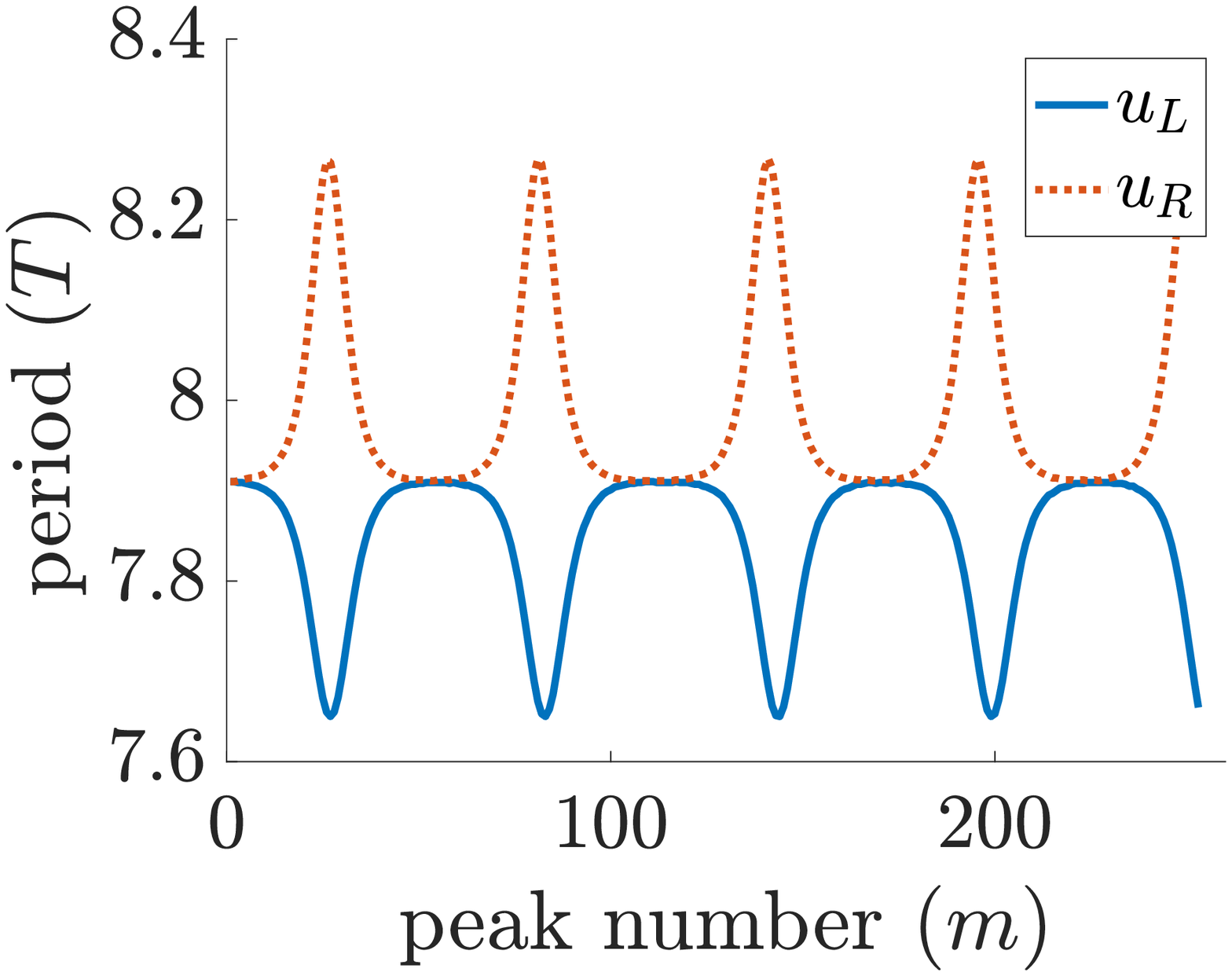} \hspace{-0.5cm}
		\label{fig:timestepSGunstablec}
	\end{subfigure}
	\begin{subfigure}{0.3\linewidth}
		\caption{}
		\includegraphics[width=5.25cm]{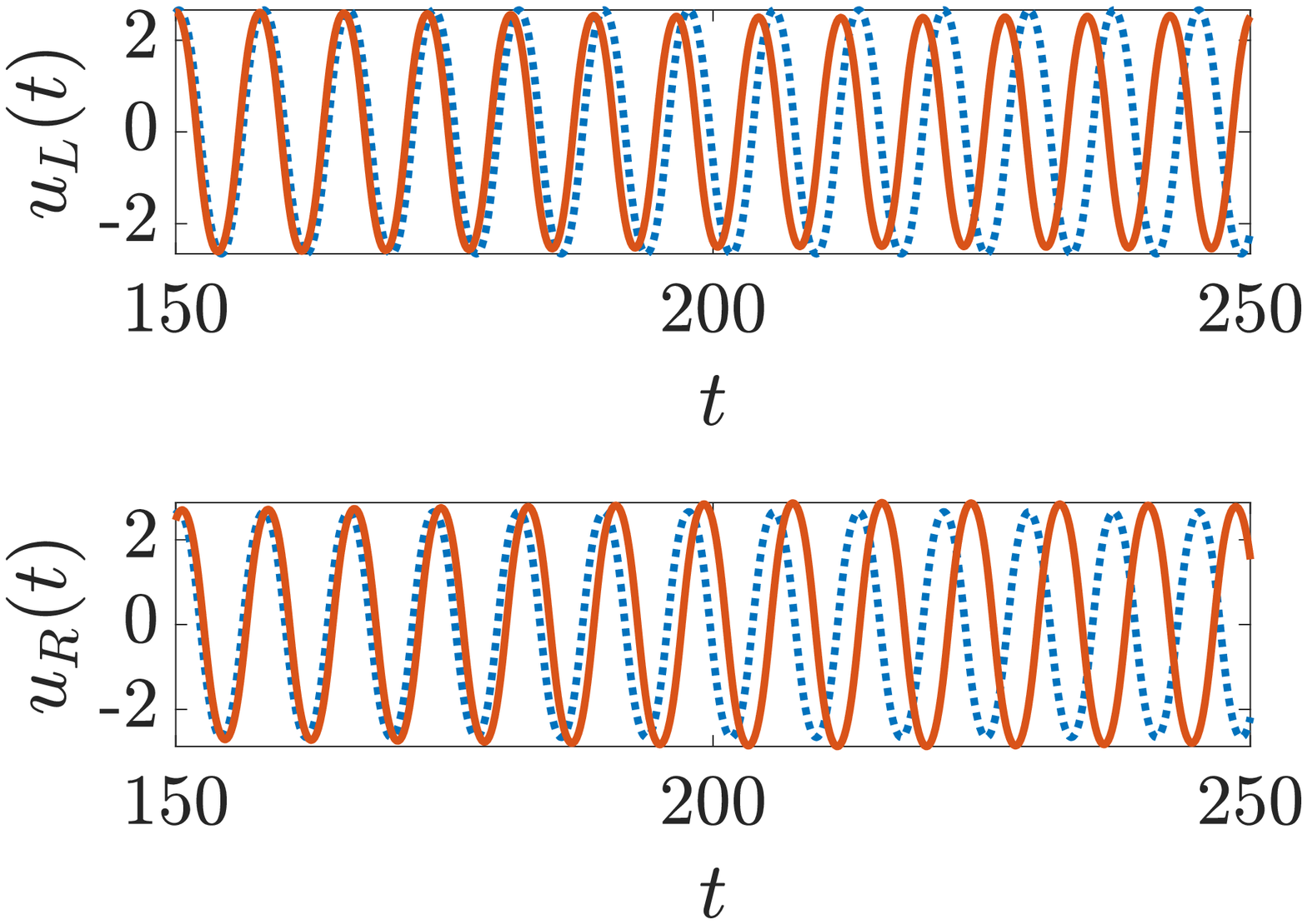} \hspace{-0.5cm}
		\label{fig:timestepSGunstabled}
	\end{subfigure}
		\begin{subfigure}{0.3\linewidth}
		\caption{}
		\includegraphics[width=5.25cm]{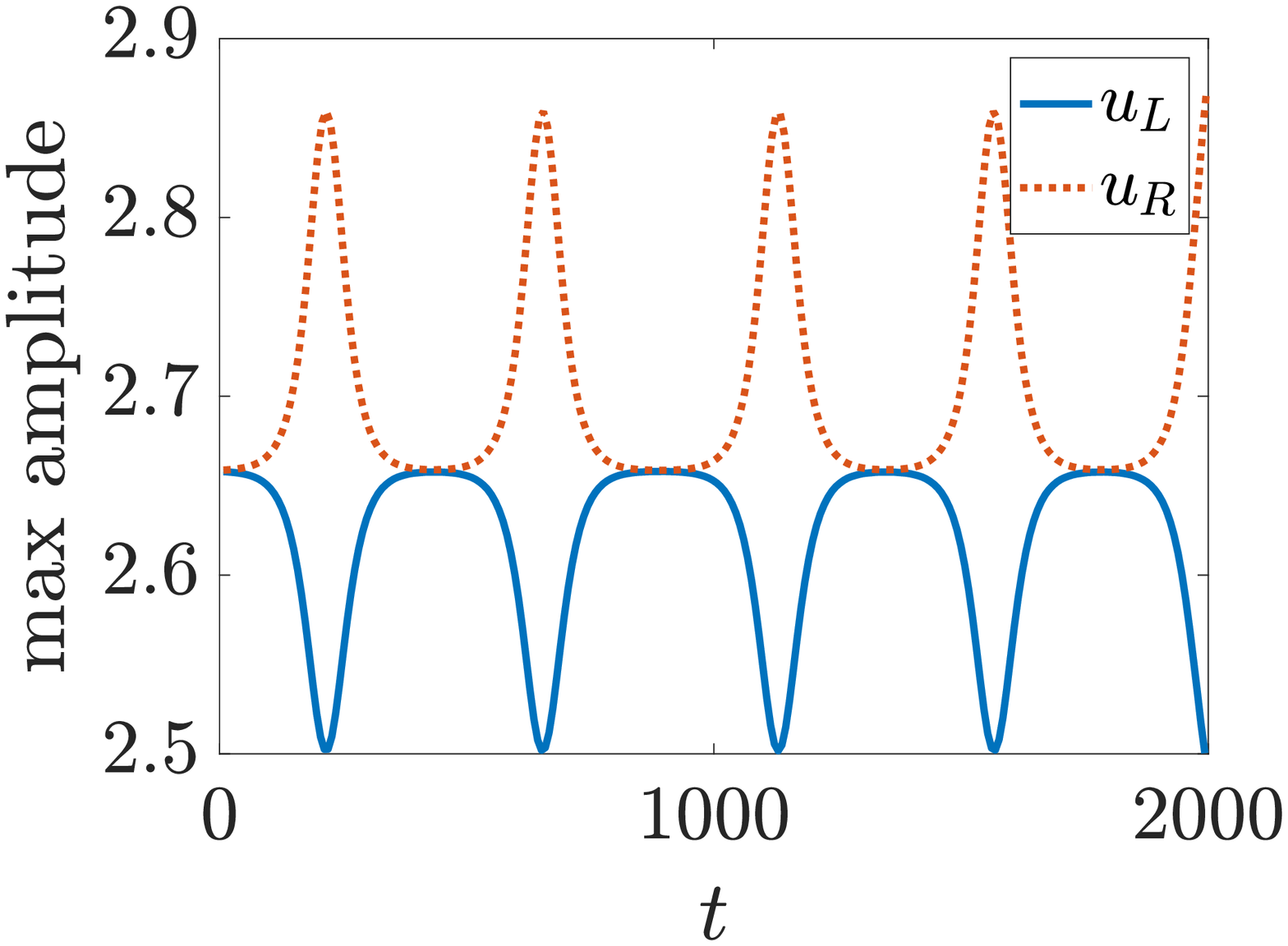} \hspace{-0.5cm}
		\label{fig:timestepSGunstablee}
	\end{subfigure}
	\begin{subfigure}{0.3\linewidth}
		\caption{}
		\includegraphics[width=5.25cm]{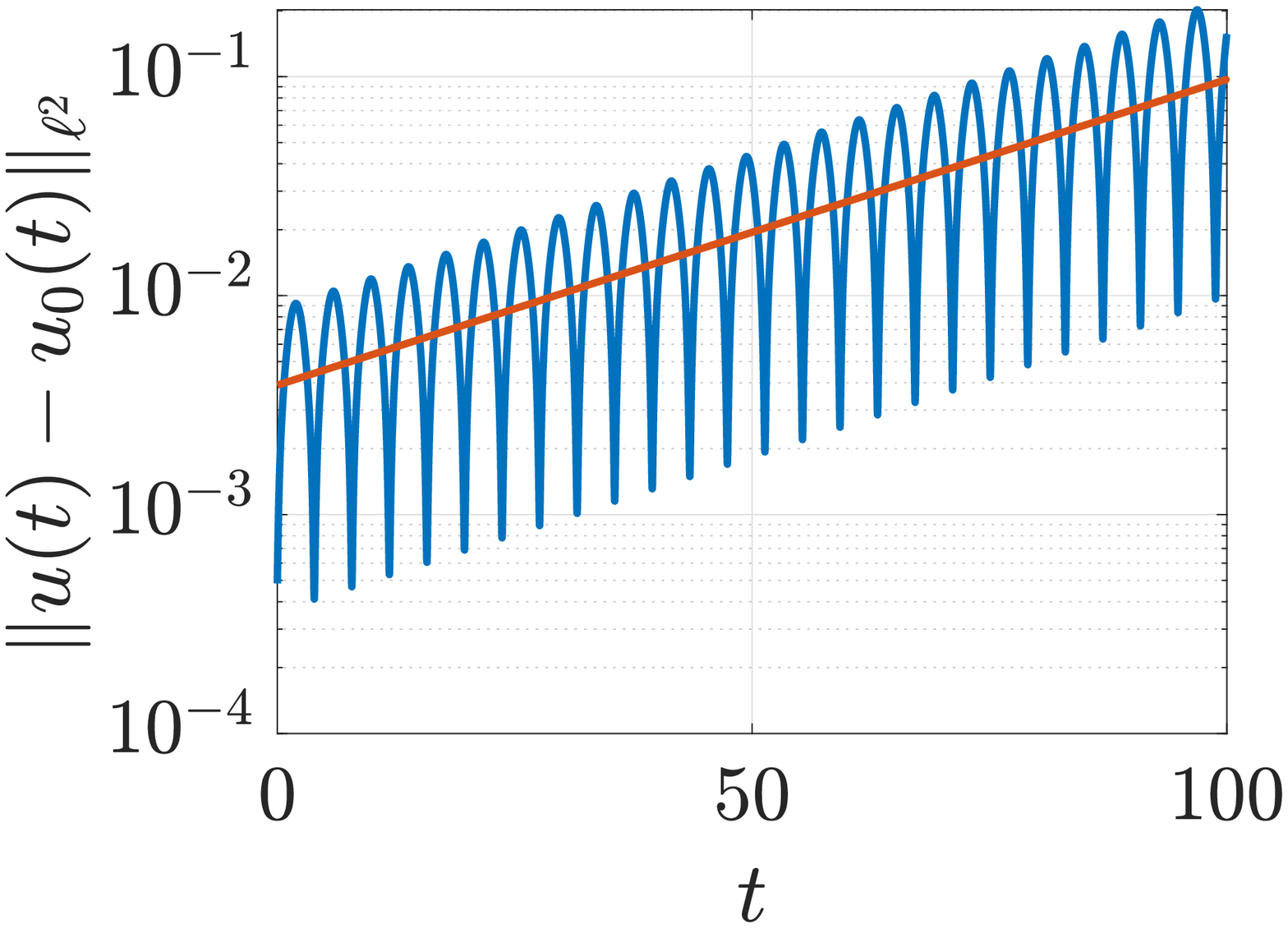} \hspace{-0.5cm}
		\label{fig:timestepSGunstablef}
	\end{subfigure}
	\end{center}
	\caption{Time evolution of the perturbation for an in-phase breather for discrete sine-Gordon. (a) colormap showing $u_n$ vs. $t$ for sites -5 to 5; breathers start in-phase (top) and are out-of-phase at approximately $t = 215$ (bottom). (b) phase difference between two breathers vs. $t$; 0 is in-phase, $\pi$ is out-of-phase. (c) period of two breathers vs. $t$; the period $T$ is measured as the distance between consecutive peaks. (d) time evolution of $u_L$ and $u_R$ for unperturbed breather (dotted blue lines) and perturbed breather (solid orange lines). (e) maximum amplitude of $u_L$ and $u_R$ vs. $t$. (f) semilog plot showing time evolution in difference of $\ell^2$ norm of perturbed and unperturbed breather, together with least squares regression line. Coupling parameter $d=0.25$, separation distance $N_1 = 6$, perturbation parameter $\delta = 0.005$.}
	\label{fig:timestepSGunstable}
\end{figure}

\begin{figure}
	\begin{center}
	\begin{subfigure}{0.45\linewidth}
		\caption{}
		\includegraphics[width=7.5cm]{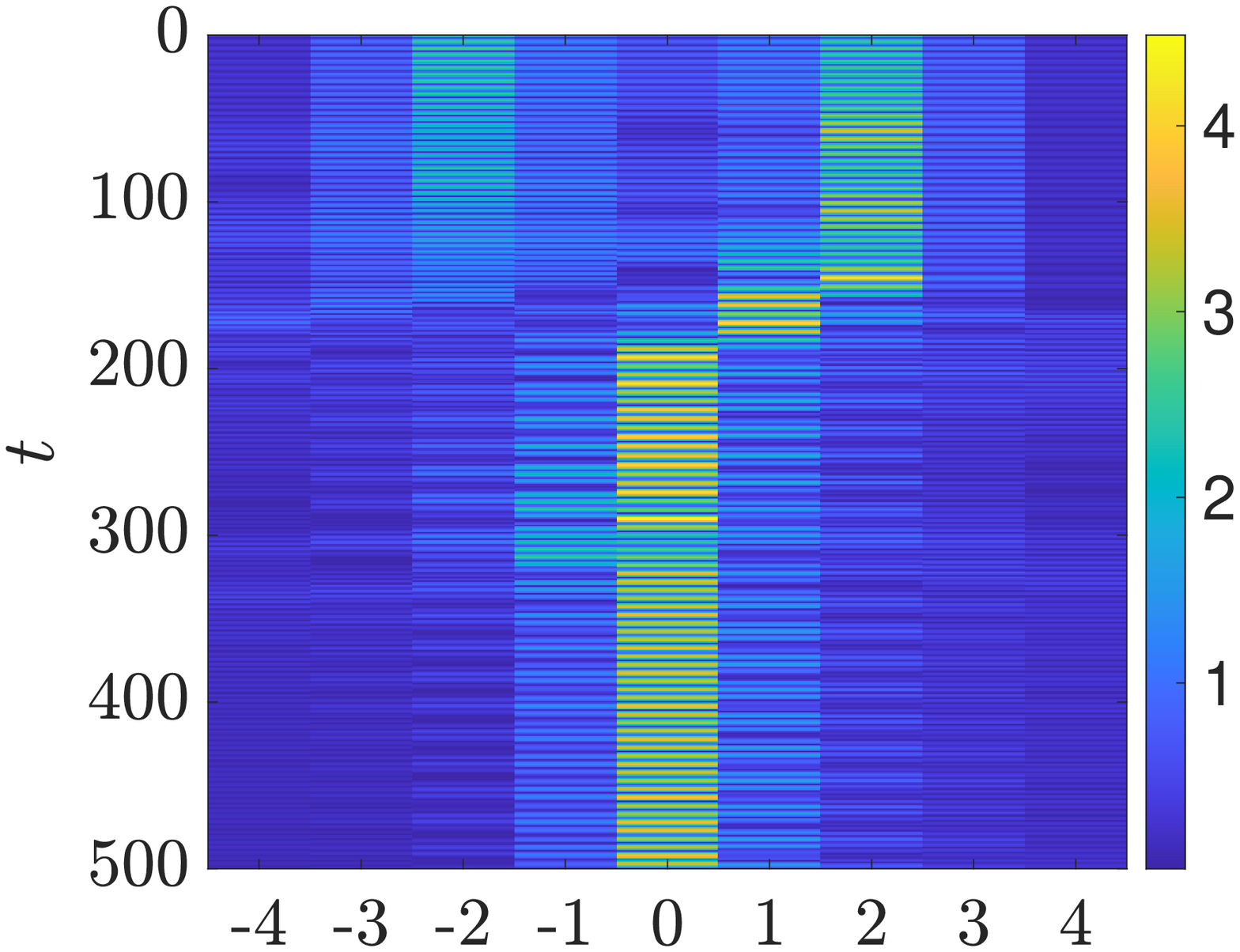}
		\label{fig:timestepSGpplonga}
	\end{subfigure}
	\begin{subfigure}{0.45\linewidth}
		\caption{}
		\includegraphics[width=7.5cm]{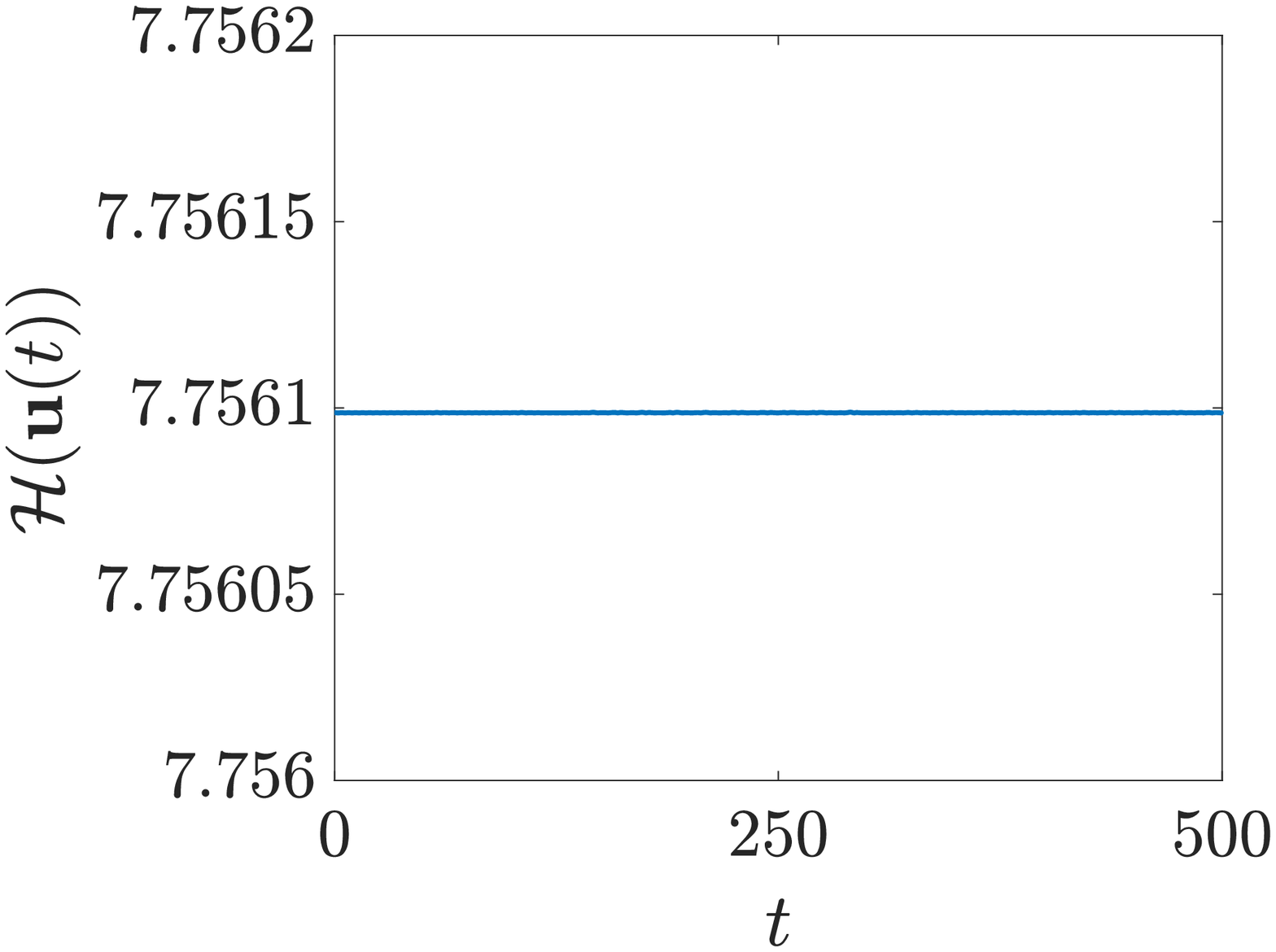}
		\label{fig:timestepSGpplongb}
	\end{subfigure}
	\begin{subfigure}{0.45\linewidth}
		\caption{}
		\includegraphics[width=7.5cm]{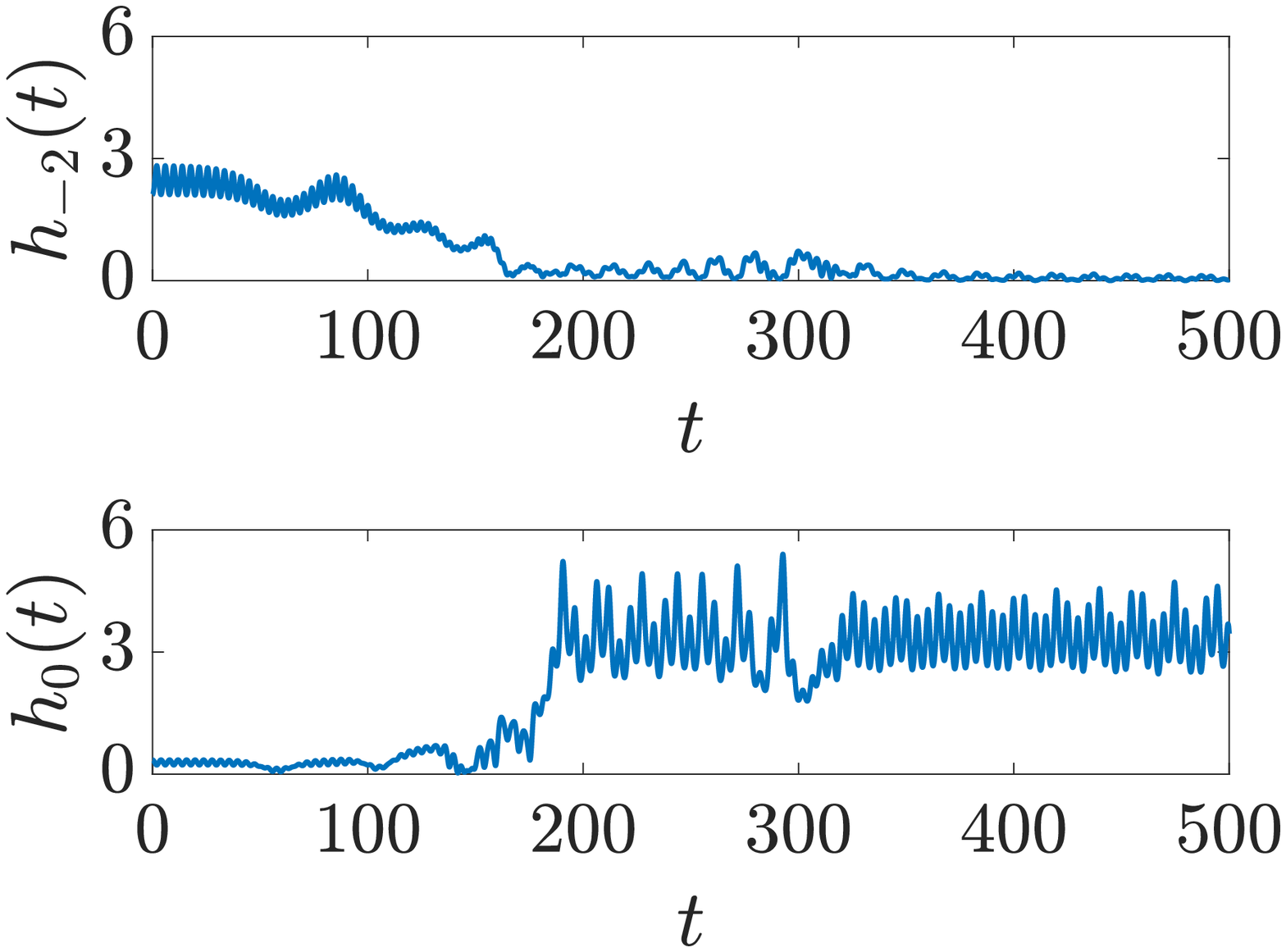}
		\label{fig:timestepSGpplongc}
	\end{subfigure}
	\begin{subfigure}{0.45\linewidth}
		\caption{}
		\includegraphics[width=7.5cm]{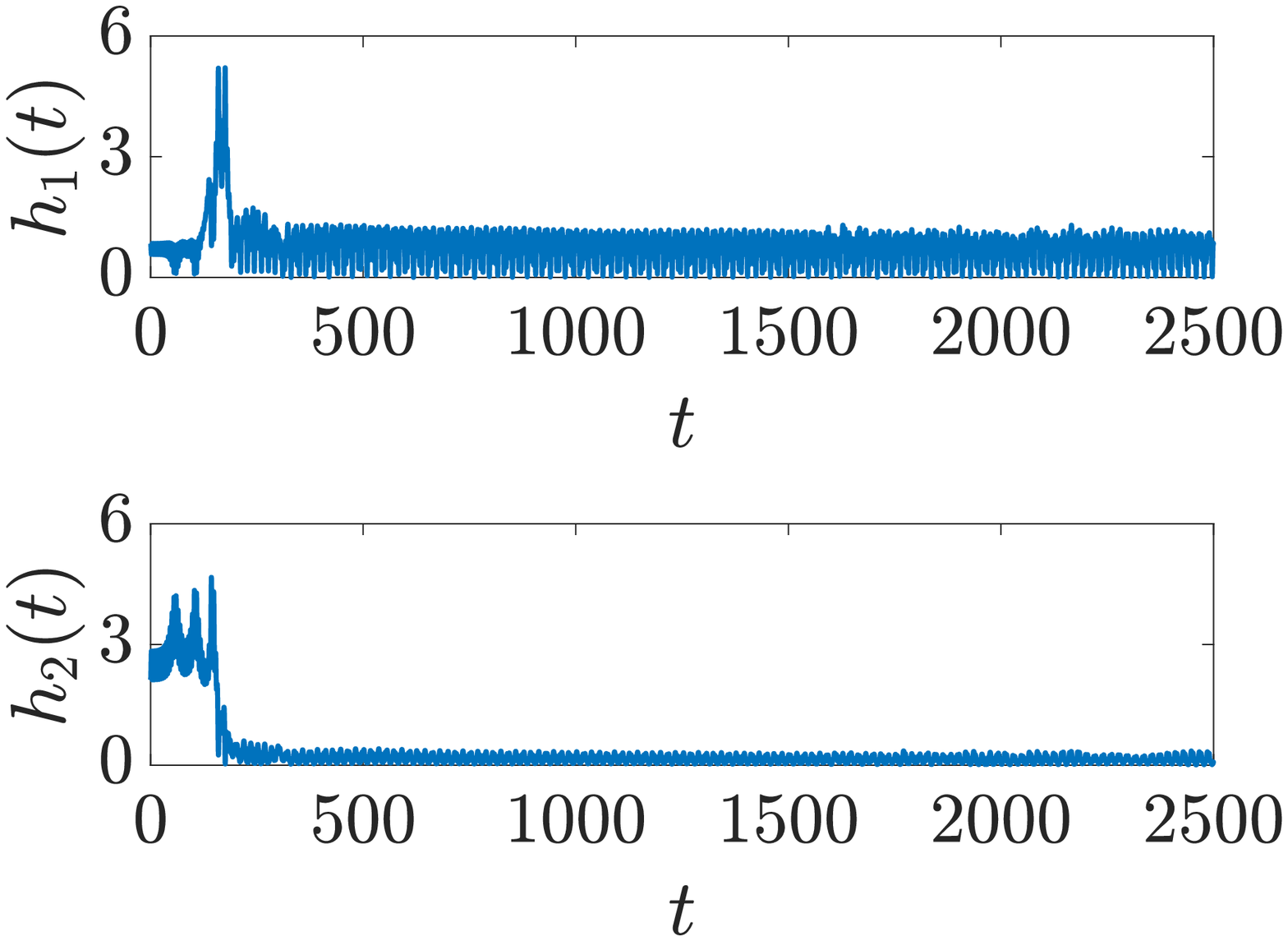}
		\label{fig:timestepSGpplongd}
	\end{subfigure}
	\end{center}
	\caption{(a) Colormap showing the evolution in $t$ of $|u_n|$ for discrete sine-Gordon. (b) Total energy $\calH(\uvec(t))$ vs. $t$ ($\calH$ is given by \cref{eq:H}). (c) and (d) Energy density vs. time for central sites of breather. Labels of sites correspond to those in (a). Separation distance $N_1 = 4$
	, coupling parameter $d=0.25$, perturbation parameter $\delta = 0.01$, lattice size parameter $L=150$.}
	\label{fig:timestepSGpplong}
\end{figure}

\subsection{Hard \texorpdfstring{$\phi^4$}{phi-4} potential}

The results for the soft $\phi^4$ potential are similar to those for sine-Gordon (which is expected, since the first two terms in the Taylor expansion of the sine-Gordon potential are qualitatively similar to the soft $\phi^4$ potential), thus we will not show them here. We instead look at the hard $\phi^4$ potential. 
For all of the plots in the section, unless otherwise noted, we use a lattice size $L = 15$, and the energy at the AC limit is $E = 1.75$, which corresponds to a period of $T = 4.1093$. 
\cref{fig:phi4sol} shows the initial conditions for the primary breather, the in-phase inter-site centered breather, and the out-of-phase inter-site centered breather, together with their Floquet spectra. In contrast to the sine-Gordon potential, the out-of-phase inter-site centered breather is unstable, while the in-phase inter-site centered breather is spectrally stable for sufficiently small $d$ (see also Figures 4 and 6 in \cite{cuevas-maraver2016}). These agree with the results of \cites{Archilla2003,Koukouloyannis2009}. As the distance between breathers $N_1$ is increased, the in-phase double breather is spectrally stable for $N_1$ odd and spectrally unstable for $N_1$ even, and the reverse is true for the out-of-phase double breather (see \cref{fig:phi4hardfloqplot}). We note that the solvability condition $K > 0$ for the hard $\phi^4$ potential; however, for the quantity $b_1$ in \cref{eq:inteigsdouble}, $b_1 > 0$ for $N_1$ even, and $b_1 < 0$ for $N_1$ odd. This agrees with the stability predictions following equation \cref{eq:inteigsdouble}, and explains the alternating eigenvalue pattern seen as the distance $N_1$ is increased. A summary of the Floquet spectral pattern for both potentials is given in \cref{table:spec}.

\begin{figure}
	\begin{center}
	\begin{subfigure}{0.3\linewidth}
		\caption{}
		\includegraphics[width=5.5cm]{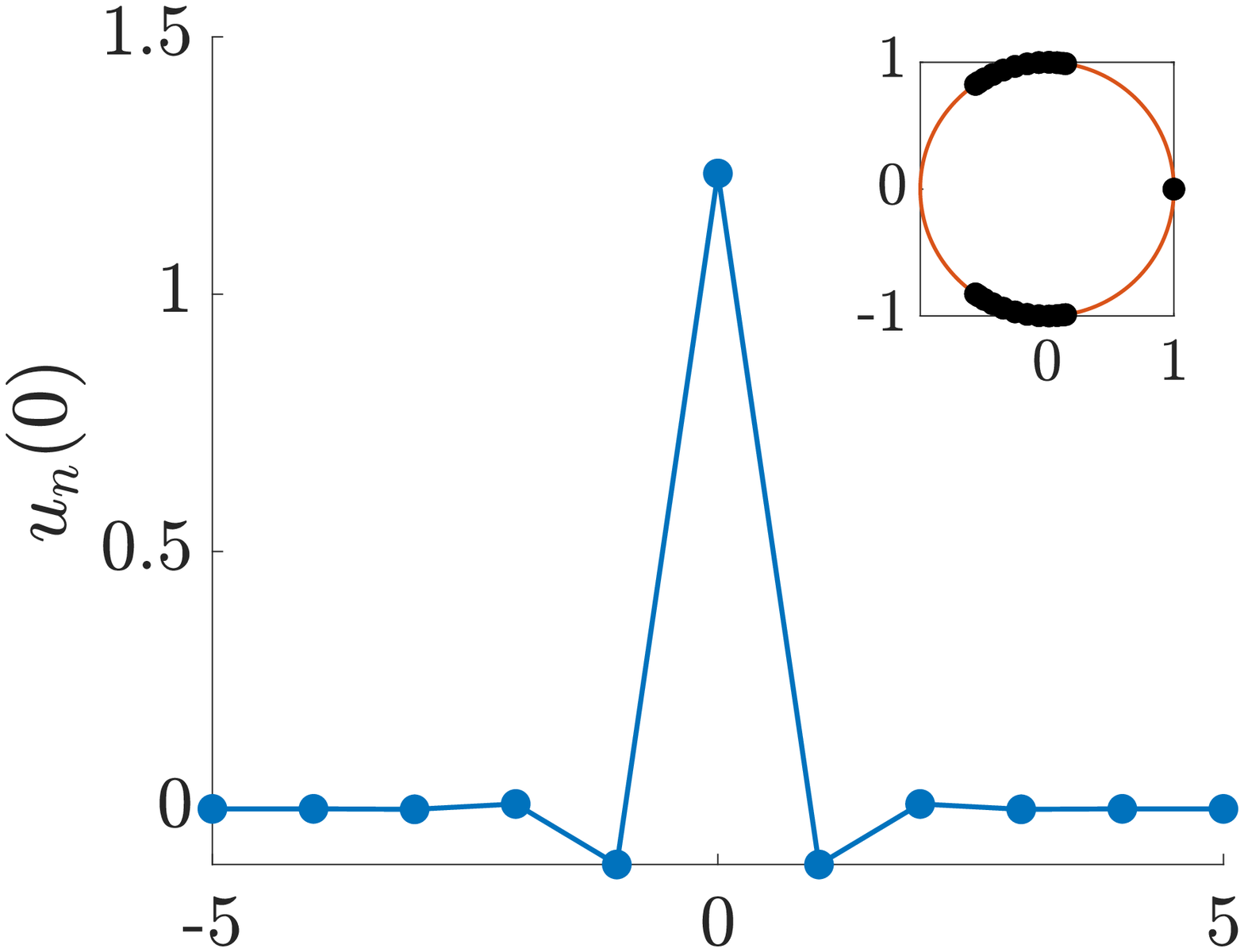} \hspace{-0.5cm}
		\label{fig:phi4sola} 
	\end{subfigure}
	\begin{subfigure}{0.3\linewidth}
		\caption{}
		\includegraphics[width=5.5cm]{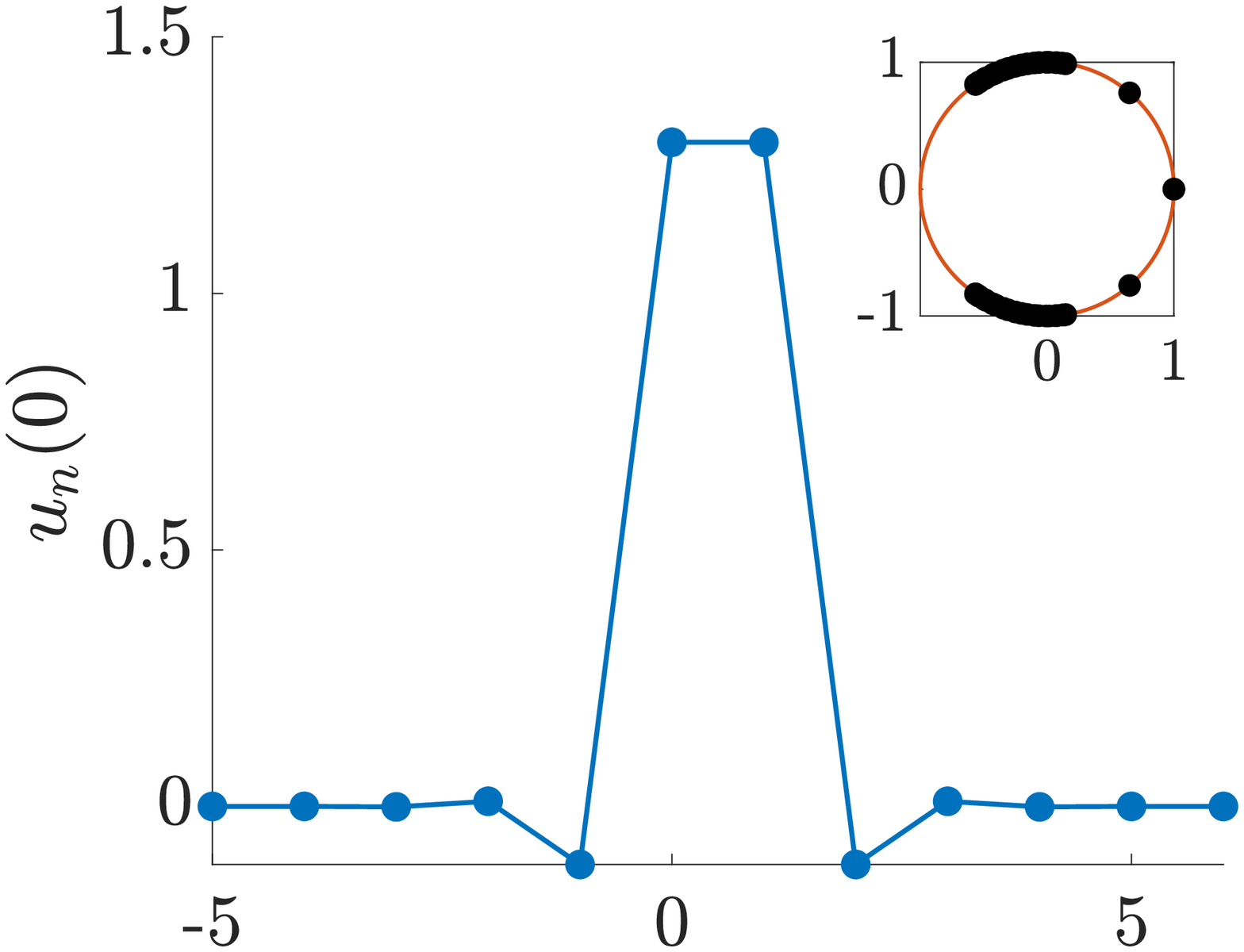} \hspace{-0.5cm}
		\label{fig:phi4solb} 
	\end{subfigure}
	\begin{subfigure}{0.3\linewidth}
		\caption{}
		\includegraphics[width=5.5cm]{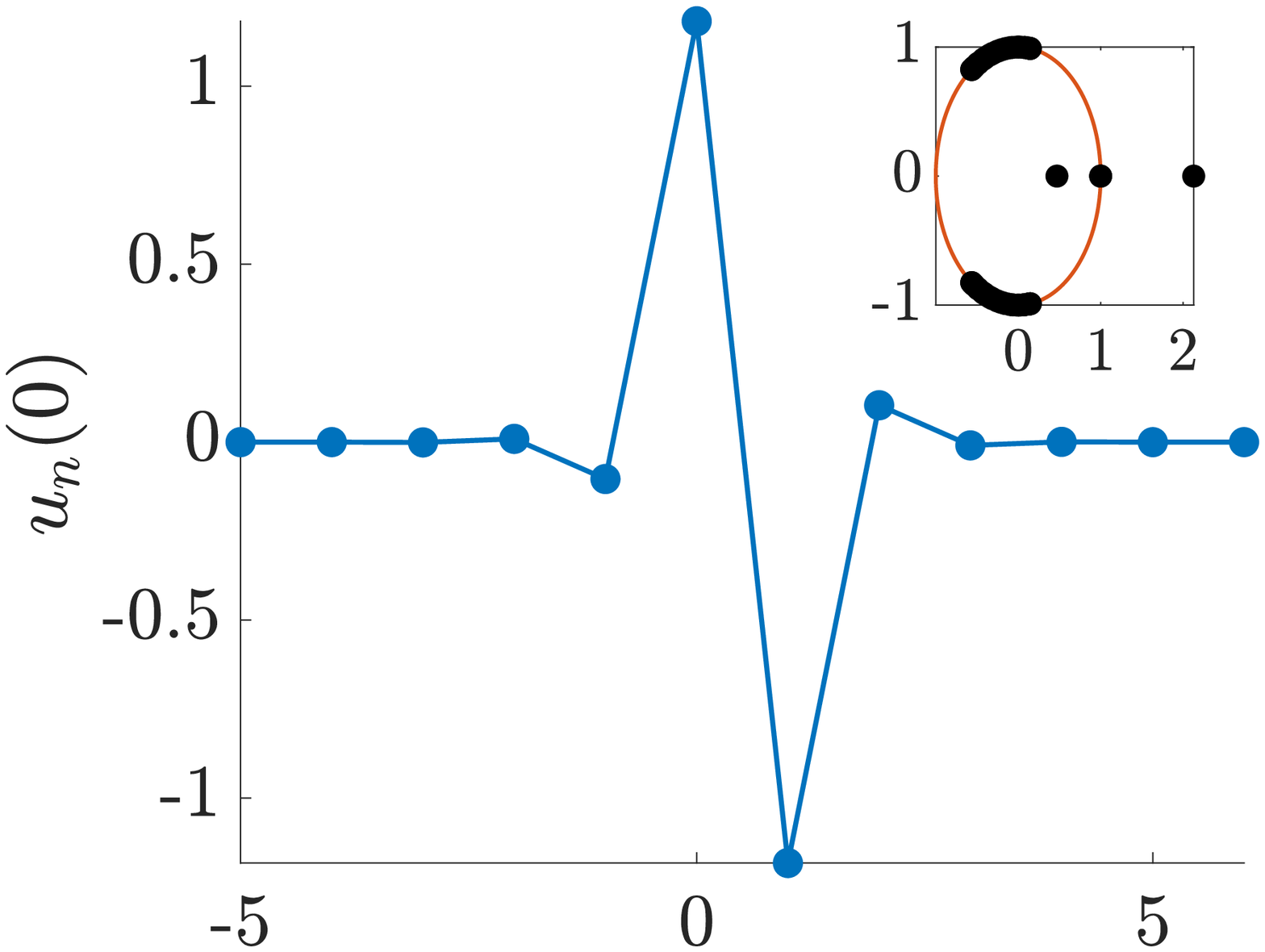} 
		\label{fig:phi4solc} 
	\end{subfigure}
	\end{center}
	\caption{Initial condition $u_n(0)$ and Floquet spectrum (inset) for primary breather (a), in-phase inter-site centered breather (b), and out-of-phase inter-site centered breather (c) for hard $\phi^4$ potential with coupling parameter $d=0.1$. }
	\label{fig:phi4sol}
\end{figure}

\begin{figure}
	\includegraphics[width=15cm]{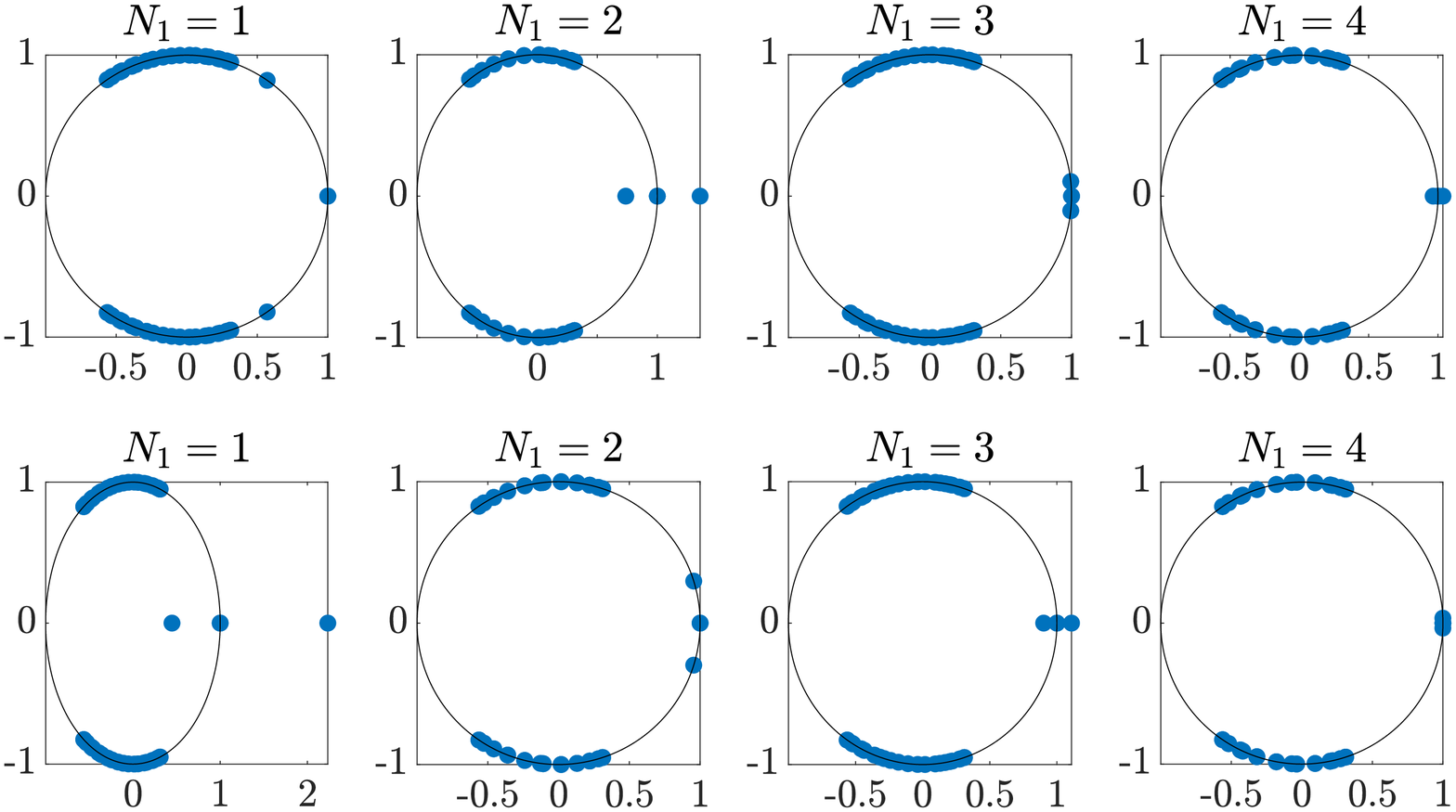}
	\caption{Floquet spectra for $N_1 = 1, 2, 3, 4$ for in-phase double breather (top) and out-of-phase double breather (bottom) for hard $\phi^4$ potential with $d = 0.125$.}
	\label{fig:phi4hardfloqplot}
\end{figure}

\cref{fig:bifdiagphi4} shows the bifurcation diagram for out-of-phase and in-phase double breathers for $N_1 = 6$. As noted above, the out-of-phase double breather is spectrally stable, while the in-phase double breather is unstable. For $N_1$ odd, this pattern is reversed (not shown). The upper branches are double breathers comprising two out-of-phase inter-site centered breathers, where we recall that this is the unstable configuration for the inter-site centered breather. The middle branch is an asymmetric double breather, comprising one single-site breather and one out-of-phase inter-site centered breather. \cref{fig:phi4eigerror} shows the relative error in the computation of the Floquet interaction eigenmodes on the lower branches for both the in-phase and out-of-phase double breathers for both even and odd $N_1$, using the formula \cref{eq:inteigsdouble}. As for the sine-Gordon equation, the relative error increases with $d$. We note that the relative error is less than $10^{-2}$ up to approximately $d = 0.25$, which is close to the turning points of the bifurcation diagrams. 

\begin{figure}
	\hbox{
	\hspace{-2cm}
	\includegraphics[width=20cm]{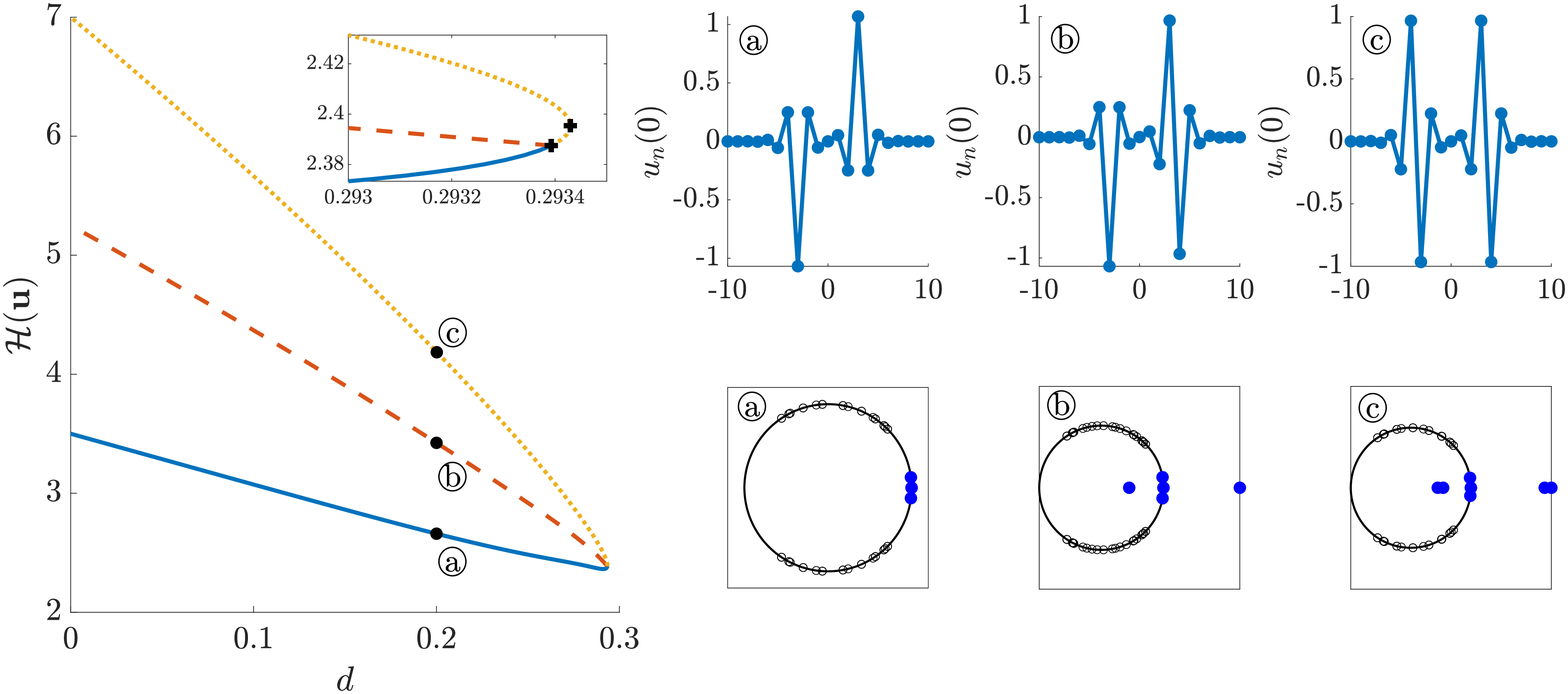}
	}
	\hbox{
	\hspace{-2cm}
	\includegraphics[width=20cm]{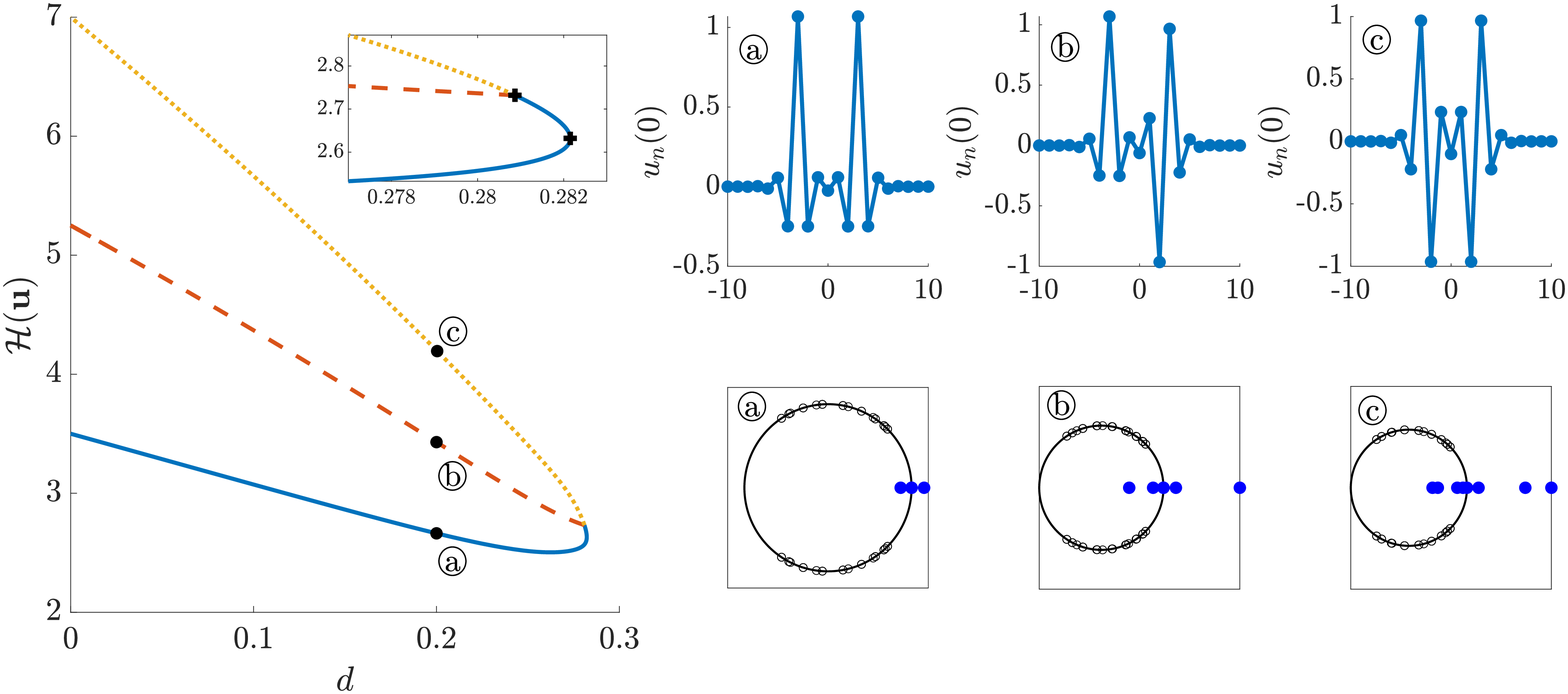}
	}
	\caption{Bifurcation diagram plotting energy $\mathcal{H}(\uvec)$ vs. $d$ for hard $\phi^4$ potential for out-of-phase (top) and in-phase (bottom) double breather with $N_1 = 6$. Solutions and Floquet spectra on right correspond to labeled points on left.}
	\label{fig:bifdiagphi4}
\end{figure}

\begin{figure}
	\begin{center}
	\begin{subfigure}{0.45\linewidth}
		\caption{}
		\includegraphics[width=7.5cm]{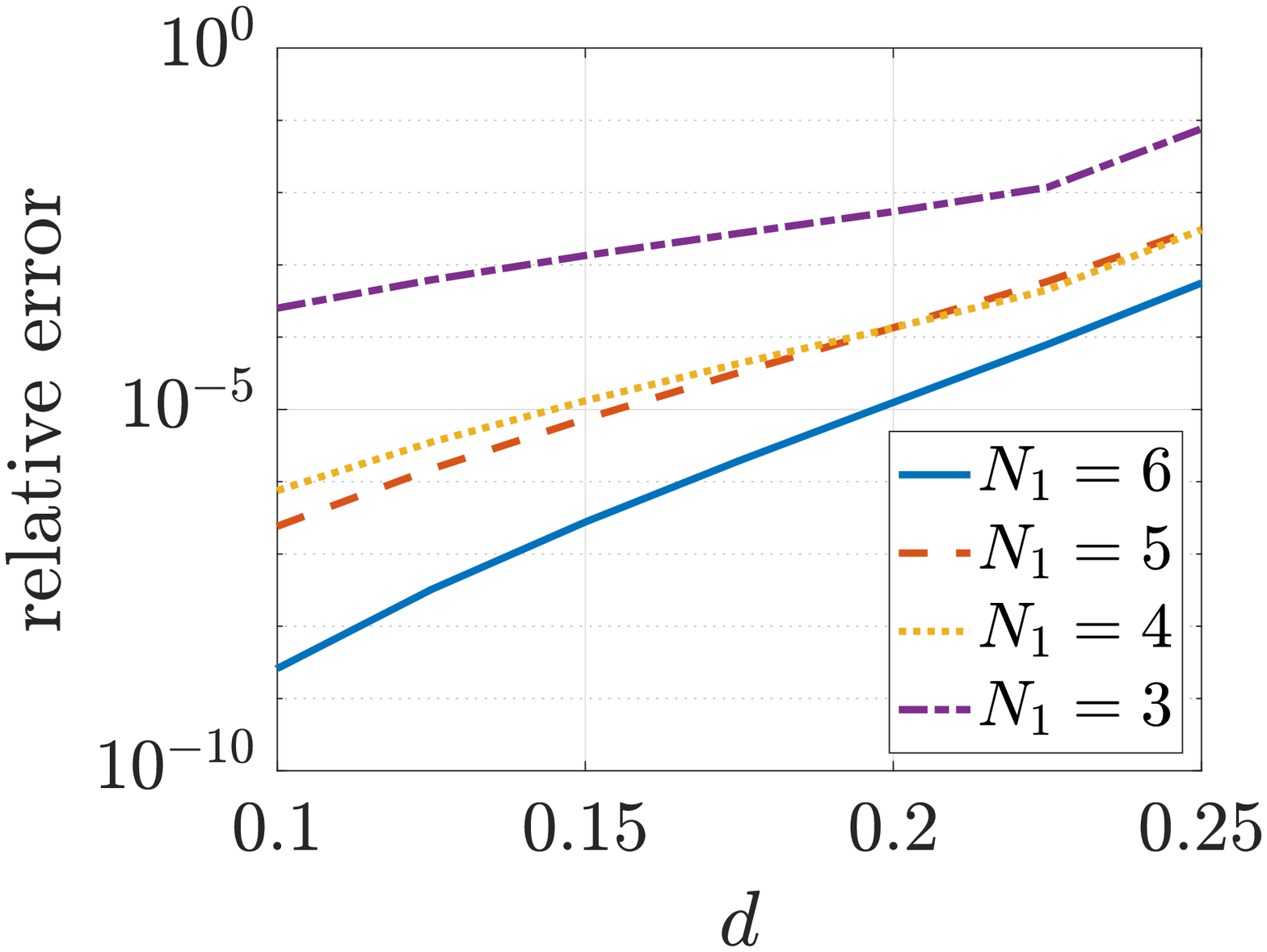} 
		\label{fig:phi4eigerrora} 
	\end{subfigure}
	\begin{subfigure}{0.45\linewidth}
		\caption{}
		\includegraphics[width=7.5cm]{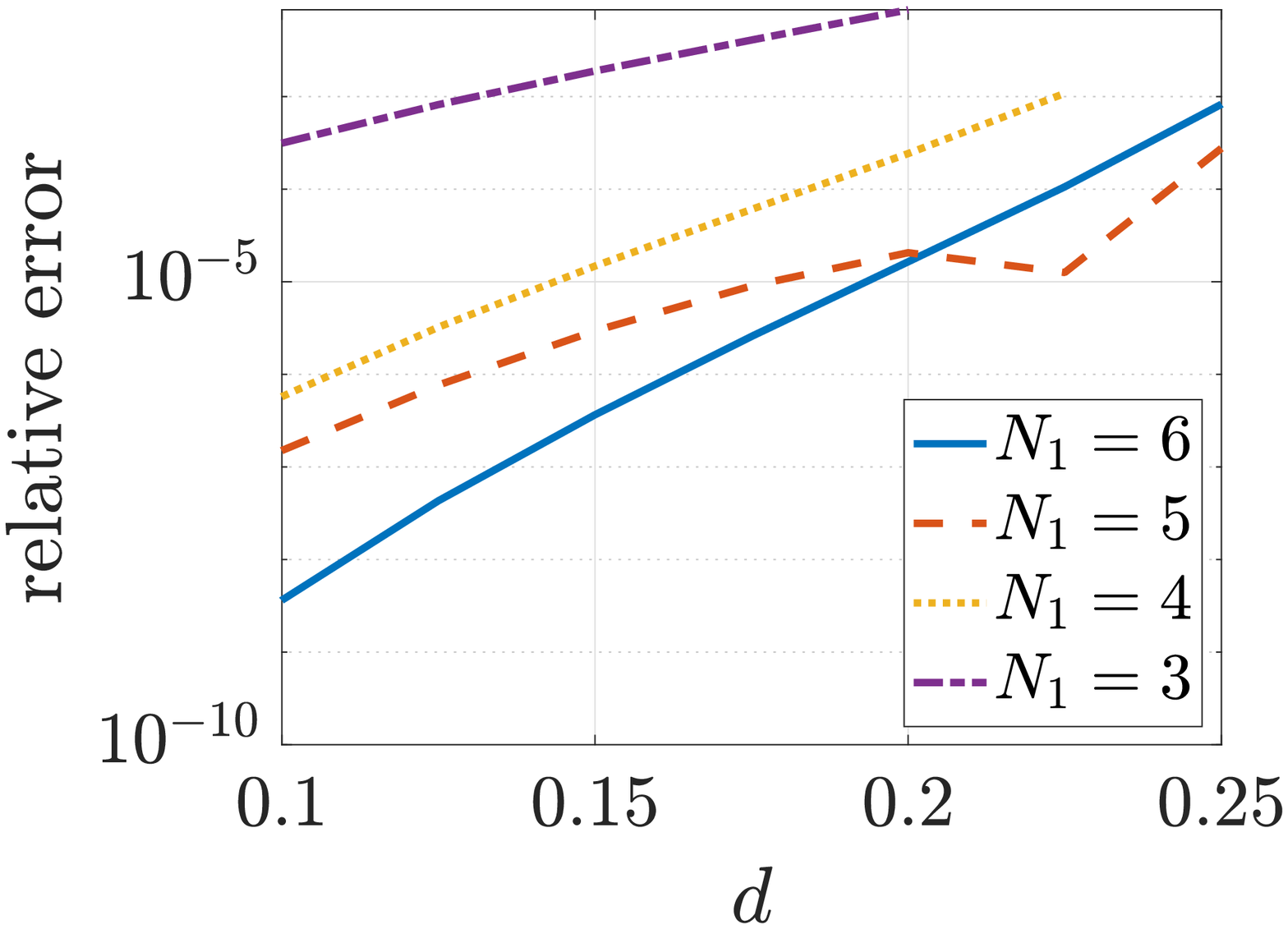} 
		\label{fig:phi4eigerrorb} 
	\end{subfigure}
	\end{center}
	\caption{Semilog plot of relative error of interaction eigenmode computation vs. $d$ for in-phase (a) and out-of-phase (b) double breathers for hard $\phi^4$ potential with $N_1 = 3, 4, 5, 6$.}
	\label{fig:phi4eigerror}
\end{figure}

\begin{table}
\begin{tabular}{llccll}\toprule
Potential & $N_1$ & sign of $K$ & sign of $b_1$ & in-phase & out-of-phase \\ \midrule
sine-Gordon (soft) & even & $-$ & $-$ & unstable & stable \\ 
                   & odd  & $-$ & $-$ & unstable & stable \\
$\phi^4$ (hard)    & even & $+$ & $+$ & unstable & stable \\
                   & odd  & $+$ & $-$ & stable & unstable \\ \bottomrule
\end{tabular}
\caption{Summary of Floquet interaction eigenmode pattern for double breathers for soft sine-Gordon potential and hard $\phi^4$ potential.}
\label{table:spec}
\end{table}

\section{Conclusions and future directions}\label{sec:conc}

In this paper, we studied multi-breather solutions to the discrete Klein-Gordon equation,
with an emphasis on a theoretical analysis of multi-breathers with well-separated excited sites, and their existence, spectral stability and dynamical properties. Specifically, we looked at an approximation of the system on a finite-dimensional Hilbert space, and used Lin's method to construct multi-breather solutions to this system. Furthermore, we used Lin's method again to reduce the eigenvalue problem to a low-dimensional, matrix equation, which can then be solved. This can be done as long as the distances between copies of the primary breather are sufficiently large. The results from this approximation are in very good agreement with direct numerical computations of the Floquet spectrum, both for soft and hard potentials. The key determining factor for the Floquet spectral pattern is the phase differences between adjacent copies of the primary, single-site breather. In addition, our results
showcased the crucial difference between the soft sine-Gordon equation and the hard $\phi^4$ equation; in the latter case, the Floquet spectral pattern depends in addition on whether the distances in lattice points between consecutive copies of the primary breather are even or odd, while in the former case, it does 
not (see also, e.g., the analogy  with the interaction of dark
solitons in the defocusing DNLS model of~\cite{Pelinovsky_2008}). 

Avenues of further research include exploring breathers and multi-breathers either in Klein-Gordon lattices with asymmetric potentials, such as the Morse potential, or in other models 
involving beyond-nearest-neighbor interactions, discussed, e.g., in~\cite{PENATI201992} and references therein. One class of models which would be interesting to study concerns equations in which the potential involves off-site terms, such as the FPUT lattice. A natural extension of this would be to look at higher dimensional lattices, both for on-site potentials such as in a higher dimensional Klein-Gordon model, and more complex potentials, such as higher dimensional FPUT lattices. In two dimensions, there are many possible regular lattice models, including square, triangular, and honeycomb, and these different geometries may exhibit qualitatively different behavior. Another possibility would be to look at more complex solutions. One class of solutions which remains unexplored is asymmetric multi-breathers, solutions comprising breathers which have different, but commensurate, fundamental periods. \cref{fig:brka} shows an example of an asymmetric double breather solution to the discrete sine-Gordon equation; the left breather has a fundamental period of 8, and the right breather has a fundamental period of 16, so the overall breather has a period of 16. While these solutions can be constructed numerically by parameter continuation from the AC limit, preliminary numerical experiments suggest that they only exist for small values of the coupling parameter $d$, and that they are spectrally unstable.
In addition, preliminary numerical experiments suggest that there are solutions comprising both breathers and kinks. See \cref{fig:brkb} for an example of a solution to the discrete sine-Gordon equation in which a breather is connected to a kink. Finally, although the finite-dimensional approximation used in the paper yields results which are in very good agreement with those found numerically, it may be still possible to remove this restriction and prove the results for the full system, although this would most likely involve an entirely different approach.

\begin{figure}
	\begin{center}
	\begin{subfigure}{0.45\linewidth}
		\caption{}
		\includegraphics[width=7.5cm]{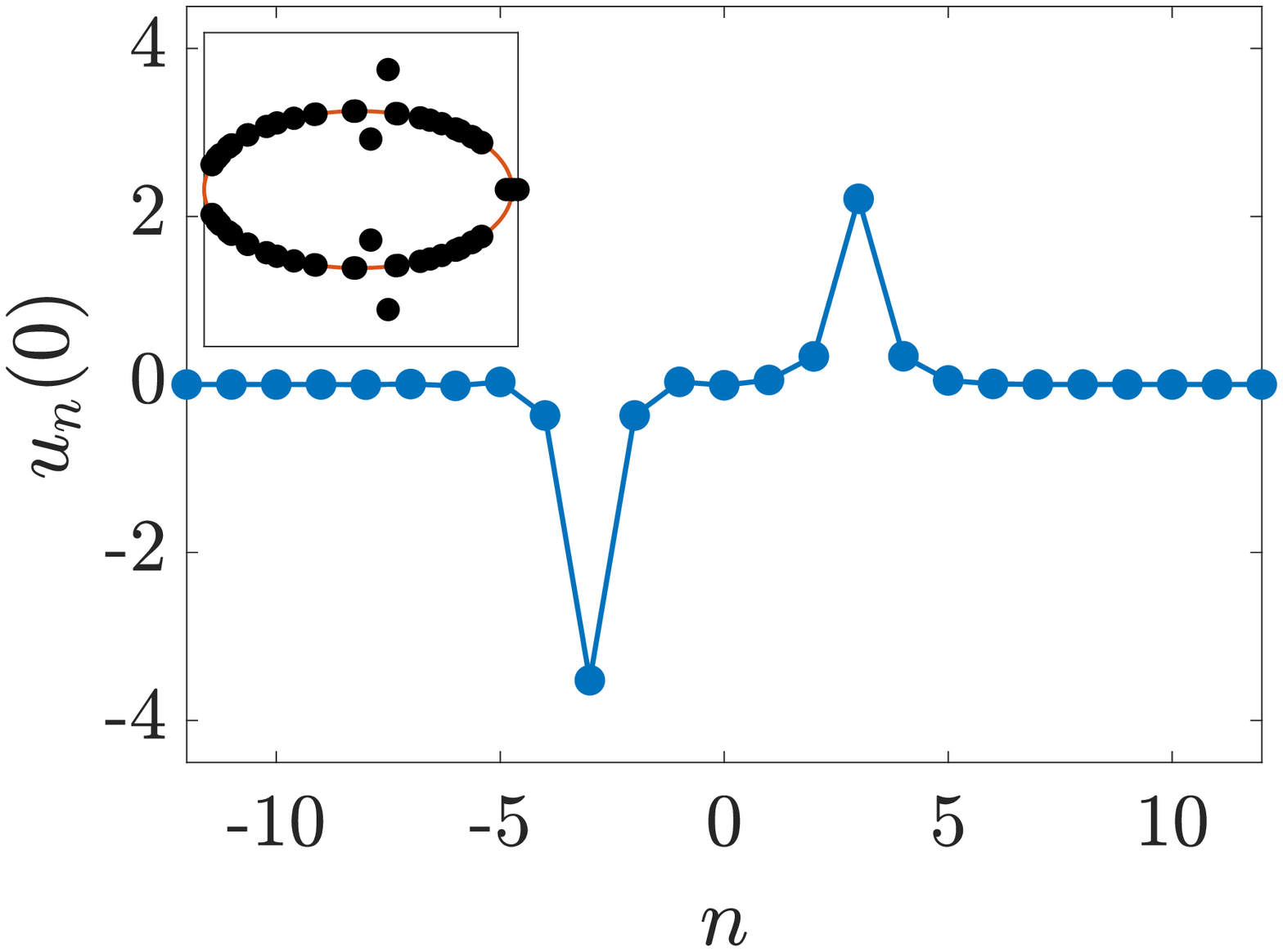} 
		\label{fig:brka} 
	\end{subfigure}
	\begin{subfigure}{0.45\linewidth}
		\caption{}
		\includegraphics[width=7.5cm]{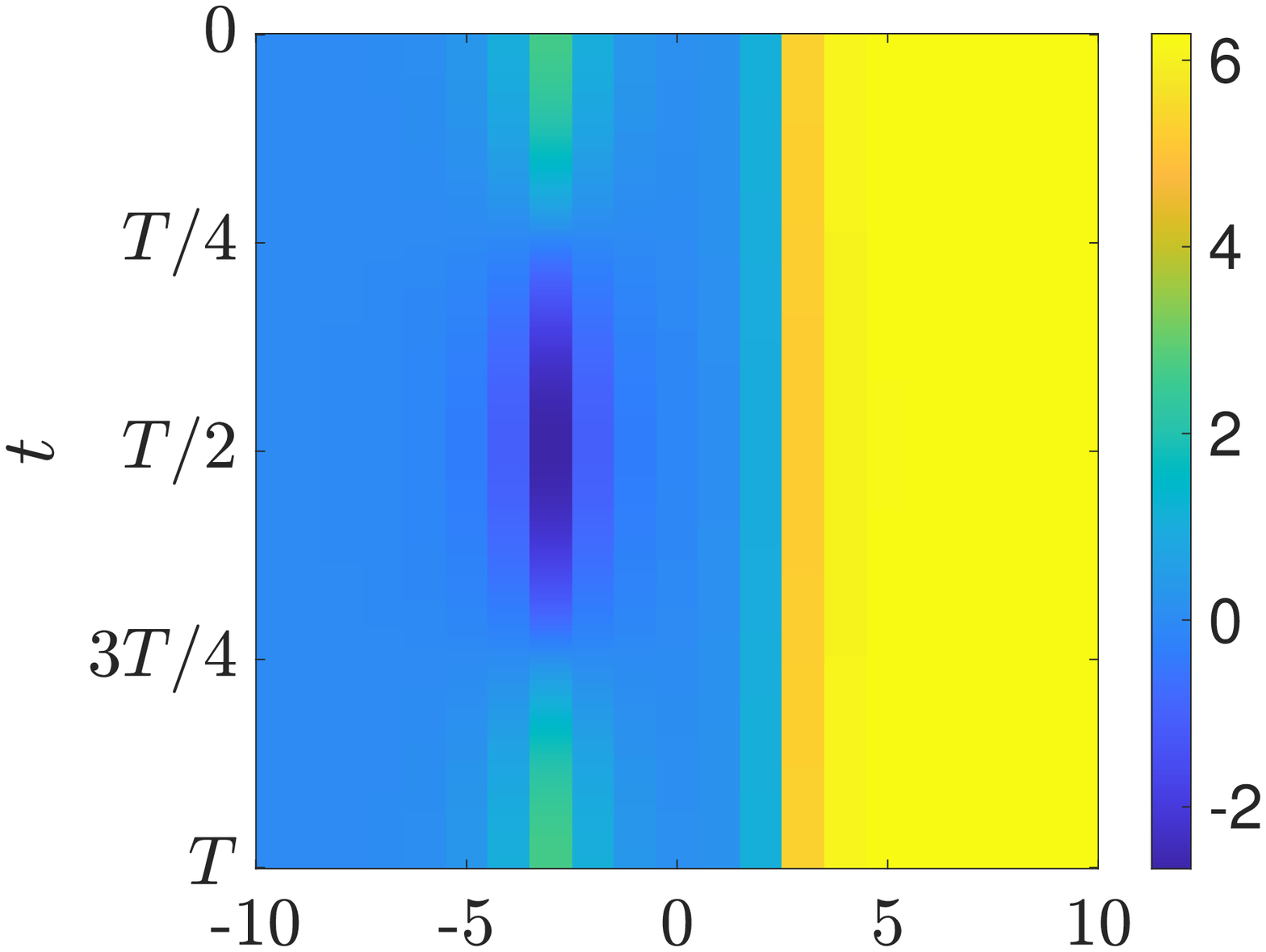} 
		\label{fig:brkb} 
	\end{subfigure}
	\end{center}
	\caption{(a) Initial condition for asymmetric double breather solution with Floquet spectrum in inset for discrete sine-Gordon, coupling parameter $d = 0.075$. (b) Colormap showing solution to discrete sine-Gordon equation comprising one breather (on left) and one kink (on right), coupling parameter $d = 0.25$. Time evolution is over one period $T$ of the breather.}
	\label{fig:brk}
\end{figure}

\vspace{0.5cm}

\paragraph{\textbf{Acknowledgments}}

This material is based upon work supported by the U.S. National Science Foundation under the RTG grant DMS-1840260 (R.P. and A.A.) and DMS-1809074 (P.G.K.). J.C.-M. acknowledges support from EU (FEDER program 2014-2020) through both Consejería de Econom\'{\i}a, Conocimiento, Empresas y Universidad de la Junta de Andaluc\'{\i}a (under the projects P18-RT-3480 and US-1380977), and MICINN and AEI (under the projects PID2019-110430GB-C21 and PID2020-112620GB-I00). R.P. would also like to thank Graham Cox and Bj\"orn Sandstede for their helpful comments and suggestions.

\appendix

\section{Proof of \texorpdfstring{\cref{th:spectrum}}{Theorem 2} }\label{app:specproof}

The proof is a straightforward adaptation of the proof of \cite{Parker2020}*{Theorem 2}. Let $Y(n) = (y(n), \tilde{y}(n)) = \omega \partial_\omega U(n)$. We take as an ansatz for the eigenfunction $W(n)$ the following piecewise linear combination
\begin{equation}\label{eq:Wansatz}
W_i^\pm(n) = c_i ( \dot{U}_i^\pm(n) + \lambda Y_i^\pm(n) ) + \tilde{W}_i^\pm(n),
\end{equation}
where $c_i \in \C$, and $W_i^\pm(n)$ is defined on the same interval as $U_i^\pm(n)$ in \cref{eq:Upiecewise}. Substituting \cref{eq:Wansatz} into \cref{eq:dynEVPM}, using the relations
\begin{equation}
\begin{aligned}
\dot{U}_i^\pm(n+1) &= DF_M(U(n))\dot{U}_i^\pm(n) \\ 
\qquad Y_i^\pm(n+1) &= DF_M(U(n)) Y_i^\pm(n) + 2 \partial_t B \dot{U}_i^\pm(n),
\end{aligned}
\end{equation}
and simplifying, we obtain the equation
\begin{equation}\label{eq:EVPpiecewise}
\begin{aligned}
\tilde{W}_i^\pm(n+1) &= DF_M(\sigma_i Q_M(n)) \tilde{W}_i^\pm(n) + \lambda^2 c_i B[ 2 \partial_t Y_i^\pm(n) + \dot{U}_i^\pm(n)] \\
&\quad+ [G_i^\pm(n) + (2 \lambda \partial_t + \lambda^2) B] \tilde{W}_i^\pm(n) + c_i \lambda^3 B Y_i^\pm(n),
\end{aligned}
\end{equation}
where
\begin{equation}\label{eq:Gipm}
G_i^\pm(n) = DF_M(U_i^\pm(n)) - DF_M(\sigma_i Q_M(n)).
\end{equation}
In addition to solving \cref{eq:EVPpiecewise}, the eigenfunction $W_i^\pm$ must satisfy matching conditions at $n = \pm N_i$ and $n = 0$. As in \cites{Parker2020,Parker2021,Sandstede1998}, this will in general not be possible. Instead, we solve the system 
\begin{equation}\label{eq:EVPsystem}
\begin{aligned}
\tilde{W}_i^\pm(n+1) &= DF_M(\sigma_i Q_M(n)) \tilde{W}_i^\pm(n) + \lambda^2 c_i B[ 2 \partial_t Y_i^\pm(n) + \dot{U}_i^\pm(n)] \\
&\quad+ [G_i^\pm(n) + (2 \lambda \partial_t + \lambda^2) B] \tilde{W}_i^\pm(n) + c_i \lambda^3 B Y_i^\pm(n) \\
\tilde{W}_i^+(N_i^+) &- \tilde{W}_{i+1}^-(-N_i^-) = C_i c \\
\tilde{W}_i^+(0) &- \tilde{W}_i^-(0) \in \C Z_M(n),
\end{aligned}
\end{equation}
where
\begin{equation}
\begin{aligned}
C_i c &= [ \dot{U}_{i+1}^-(-N_i^-) + \lambda Y_{i+1}^-(-N_i^-) ]c_{i+1} 
- [ \dot{U}_i^+(N_i^+) + \lambda Y_i^+(N_i^+) ] c_i.
\end{aligned}
\end{equation}
A solution to \cref{eq:EVPsystem} is an eigenfunction if and only if the $m$ jump conditions
\begin{equation}\label{eq:jump1}
\begin{aligned}
\xi_i &= \langle \sigma_i Z_M(0), \tilde{W}_i^+(0) - \tilde{W}_i^-(0) \rangle = 0
&& \qquad i = 1, \dots, m
\end{aligned}
\end{equation}
are satisfied.

We proceed as in the proof of \cite{Parker2020}*{Theorem 2}. Since $DF_M(0)$ is hyperbolic, and 
\[
\| Q_M(n) \|_{L^2_\per([0,T])} \leq C r_M^{-|n|}
\]
by the stable manifold theorem, we can adapt the results of \cite{Palmer1988} to decompose the evolution operator $\Phi(m,n)$ of the variational equation \cref{eq:vareq} in exponential dichotomies on $\Z^\pm$. We then rewrite \cref{eq:EVPpiecewise} as a fixed point problem using the discrete variation of constants formula (see, for example, \cite{Parker2020}*{Lemma 3}), and project onto the stable and unstable subspaces of the exponential dichotomy. From there, we follow the steps in \cite{Parker2020}, using the estimates \cref{eq:Uestimates}, to obtain a unique solution to \cref{eq:EVPsystem}. The jump conditions \cref{eq:jump1} become
\begin{equation}\label{eq:jump2}
\begin{aligned}
\xi_i &= \sigma_i \sigma_{i+1} \langle Z_M(N_i^+), \dot{Q}_M (-N_i^-) \rangle (c_{i+1} - c_i) \\
&\qquad +\sigma_i\sigma_{i-1} \langle  Z_M(-N_{i-1}^-), \dot{Q}_M(N_{i-1}^+) \rangle (c_i - c_{i-1}) \\
&\qquad -\frac{1}{d} \lambda^2 \sum_{n=-\infty}^\infty \left\langle Z_M(n+1), B[ 2 \partial_t Y(n) + \dot{Q}_M(n)]\right\rangle + R(\lambda)_i(c),
\end{aligned}
\end{equation}
where the remainder term has the uniform bound
\[
\| R(\lambda)(c)\|_{X_M^2} \leq C \left( r_M^{-N} + |\lambda|\right)^3.
\]
Evaluating the inner products, we obtain for the jumps
\begin{equation}\label{eq:jump3}
\begin{aligned}
\xi_i &= a_i (c_{i+1} - c_i) + \tilde{a}_{i-1} (c_i - c_{i-1}) + \frac{1}{d} \lambda^2 K + R(\lambda)_i(c),
\end{aligned}
\end{equation}
where $a_i$ and $\tilde{a}_{i-1}$ are defined by \cref{eq:ai}, and $K$ is defined by \cref{eq:M}. Equation \cref{Elambda} follows by writing the jump conditions \cref{eq:jump3} in matrix form as
\begin{equation}\label{eq:matrixform}
\left( A + \frac{1}{d} K \lambda^2 I + R(\lambda) \right) c = 0
\end{equation}
on $\R^m$, where $c = (c_1, \dots, c_m)^T$, and the matrix $A$ is defined by \cref{eq:matrixA}. Equation \cref{eq:matrixform} has a nontrivial solution if and only if its determinant is 0. 

\bibliographystyle{amsplain}
\bibliography{DKG.bib}

\end{document}